\newif\ifabstract

\abstractfalse 
\newif\iffull
\ifabstract \fullfalse \else \fulltrue \fi

\ifabstract
\documentclass[11pt]{article}
\usepackage{fullpage}
\fi

\iffull
\documentclass[11pt]{article}
\usepackage{fullpage}
\fi

\usepackage{amssymb,amsmath}
\usepackage{amsthm}
\usepackage{graphicx}
\usepackage{subfigure}
\usepackage{graphics}
\usepackage{url}
\usepackage{hyperref}
\usepackage{color}

\newtheorem{theorem}{Theorem}[section]
\newtheorem{lemma}[theorem]{Lemma}

\newtheorem{corollary}[theorem]{Corollary}
\newtheorem{definition}[theorem]{Definition}

\newcommand{\DEF}[1]{{\em #1\/}}

\newcommand{\genus}{\mathsf{genus}}

\newcommand{\eg}{\mathsf{eg}}

\renewcommand{\phi}{\varphi}
\newcommand{\eps}{\varepsilon}

\newcommand{\comment}[1]{}

\newcommand{\randbem}[3]{}%
\newcommand{\ken}[1]{\randbem{\red}{K}{#1}}
\ifabstract

\fi
\iffull

\fi

\title{Beyond the Euler characteristic:\\ Approximating the genus of general graphs}
\author{Ken-ichi Kawarabayashi\thanks{National Institute of Informatics, 2-1-2, Hitotsubashi, Chiyoda-ku, Tokyo, Japan.
    \texttt{k\_keniti@nii.ac.jp}.
    Supported by JST ERATO Kawarabayashi Large Graph Project.}
\and
Anastasios Sidiropoulos\thanks{
Dept.~of Computer Science and Engineering, and Dept.~of Mathematics, The Ohio State University. Columbus, OH, 43210.
\texttt{sidiropoulos.1@osu.edu}.
Supported by NSF grant CCF 1423230.
}
}

\date{\today}

\begin{document}

\pagenumbering{gobble}

\maketitle

\begin{abstract}
Computing the Euler genus of a graph is a fundamental problem in graph theory
and topology. It has been shown to be NP-hard by Thomassen \cite{thomassen}
and a linear-time fixed-parameter algorithm has been obtained by Mohar \cite{mohar2}. Despite extensive study, the approximability of the Euler genus remains wide open.
While the existence of an $O(1)$-approximation is not ruled out, the currently best-known upper bound is a trivial $O(n/g)$-approximation that follows from bounds on the Euler characteristic.

In this paper, we give the first non-trivial approximation algorithm for this problem.
Specifically, we
present a polynomial-time algorithm which given a graph $G$ of Euler genus $g$ outputs an embedding of $G$ into a surface of Euler genus $g^{O(1)}$.
Combined with the above $O(n/g)$-approximation, our result also implies a $O(n^{1-\alpha})$-approximation, for some universal constant $\alpha>0$.
%

Our approximation algorithm also has implications for the design of algorithms on graphs of small genus.
Several of these algorithms require that an embedding of the graph into a surface of small genus is given as part of the input.
Our result implies that many of these algorithms can be implemented even when the embedding of the input graph is unknown.
\end{abstract}

\newpage
\setcounter{page}{1}
\pagenumbering{arabic}

\section{Introduction}

A \emph{drawing} of a graph $G$ into a surface ${\cal S}$ is a mapping $\phi$
that sends every vertex $v\in V(G)$ into a point $\phi(v)\in {\cal
  S}$ and every edge into a simple curve connecting its endpoints, so
that the images of different edges are allowed to intersect only at
their endpoints.
The \DEF{Euler genus} of a surface ${\cal S}$, denoted by $\eg({\cal S})$, is defined to be $2-\chi({\cal S})$, where $\chi({\cal S})$ is the Euler characteristic of ${\cal S}$.
This parameter coincides with the usual notion of genus, except that it is twice as large if the surface is orientable.
For a graph $G$, the Euler genus of $G$, denoted by $\eg(G)$, is defined to be the minimum Euler genus of a surface ${\cal S}$, such that $G$ can be embedded into ${\cal S}$.

In this paper we consider the following basic problem:
Given a graph $G$, approximate $\eg(G)$.
This is a fundamental problem in graph theory and its exact version is one of the basic problems listed by Garey and Johnson \cite{garey2002computers}.
Part of the original motivation for the study of the genus of graphs goes back to the Heawood problem which concerns the maximum chromatic number of graphs embeddable in a fixed surface.
The solution of the Heawood problem turned out to be equivalent to determining the genus of complete graphs (cf. \cite{ringel}).
The practical interest for planar embeddings, and more generally, embeddings into low-genus surfaces arises, for instance, in problems concerning VLSI.
Moreover, ``nearly planar'' networks can be used to model a plethora of natural objects and phenomena.
Algorithmic interest comes from the fact that graphs of bounded genus naturally generalize the family of planar graphs and share many important properties with them.
Moreover, graphs of small genus play a central role in the seminal work of Robertson and Seymour on graph minors and the proof of Wagner's
conjecture.

Apart from bounds on the genus of specific families of graphs, there are no general results available.
This can be explained by the result of Thomassen \cite{thomassen}
who showed that computing the genus of a given graph exactly is NP-hard.
Nevertheless, closely related problems have been extensively studied by many researchers.
For example, a seminal result of Hopcroft and Tarjan \cite{lin1} gives  a linear time algorithm for testing planarity of graphs, and for computing a planar embedding if one exists.
Extending this planarity result, many researchers have focused on the case when the Euler genus $g$ is a fixed constant.
Filotti, Miller, and Reif \cite{FMR} were
the first to give a polynomial time algorithm for this problem.
In their solution, the degree of the polynomial bound on the time
complexity depends on $g$.
Djidjev and Reif \cite{dr} improved the result of \cite{FMR} by
presenting a polynomial time algorithm for each fixed orientable
surface, where the degree of the polynomial is fixed. In addition,
linear time algorithms have been devised for embedding graphs into the
projective plane \cite{projective} and the torus \cite{torus}.
Mohar \cite{mohar1,mohar2} finally gave a linear time algorithm for embedding a graph into an arbitrary fixed surface.
This is one of the deepest results in this area, generalizing linear time algorithms for planarity \cite{lin2,lin4,lin1,lin3}.
A relatively simple linear-time algorithm was given by Kawarabayashi, Mohar, and Reed \cite{DBLP:conf/focs/KawarabayashiMR08}.
From the work of Robertson and Seymour \cite{RobertsonS90b}, the family of graphs of genus at most $g$ is characterized as the
class of graphs that exclude as a minor all graphs from a finite
family. However, this family of excluded minors is not known
explicitly even for small values of $g$ and can generally be very large  (it contains two graphs for $g=0$ \cite{kuratowski1930probleme, wagner1937eigenschaft}, and 35 graphs for $g=1$ \cite{DBLP:journals/jgt/Archdeacon81, DBLP:journals/jct/GloverHW79}).
The dependence of the running time of all of the above mentioned exact algorithms is at least exponential in $g$.

\paragraph{Our results.}

We consider the problem of approximating $\eg(G)$, when $\eg(G)$ is not fixed. Perhaps surprisingly, despite its central importance, essentially nothing is known for this problem on general graphs.
Let us first briefly describe what is currently known.
Euler's characteristic implies that any $n$-vertex graph of Euler genus $g$ has at most $O(n + g)$ edges. Since any graph can be drawn into a surface that has one handle for every edge, this immediately implies
a O$(n/g)$-approximation, which is a $\Theta(n)$-approximation in the worst case.
In other words, even though we currently cannot exclude the existence of an $O(1)$-approximation, the state of the art only gives a trivial $O(n)$-approximation.

We give the first non-trivial approximation algorithm for $\eg(G)$ on general graphs.
Our result can be summarized as follows.

\begin{theorem}[Approximating the Euler genus of general graphs]\label{thm:main1}
There exists a polynomial-time algorithm which given a graph $G$ and an integer $g$, either correctly decides that $\eg(G)>g$, or outputs an embedding of $G$ into a surface of Euler genus $O(g^{256} \log^{189}n)$.
\end{theorem}

Combined with the above trivial $O(n)$-approximation, our result implies the first non-trivial approximation algorithm for approximating the Euler genus of general graphs.
\begin{corollary}
There is a polynomial-time $O(n^{1-\alpha})$-approximation algorithm for Euler genus, for some universal constant $\alpha>0$.
\end{corollary}

  Kawarabayashi, Mohar and Reed \cite{DBLP:conf/focs/KawarabayashiMR08} have obtained an
  exact algorithm for computing $\eg(G)$ with running time
  $2^{O(\eg(G))} n$.
  This implies a polynomial-time algorithm when $\eg(G) = O(\log n)$.
  Combining this result with Theorem~\ref{thm:main1}, we immediately obtain the following Corollary.

\begin{corollary}\label{cor:main1}
There exists a polynomial-time algorithm which given a graph $G$ and an integer $g$, either correctly decides that $\eg(G)>g$, or outputs an embedding of $G$ into a surface of Euler genus $g^{O(1)}$.
\end{corollary}

\paragraph{Previous work on approximating the genus of graphs.}

For special classes of graphs, some $o(n)$-approximation guarantees are known.
Chekuri and Sidiropoulos \cite{DBLP:conf/focs/ChekuriS13} have recently obtained a polynomial-time algorithm which given a graph $G$ of maximum degree $\Delta$ computes an embedding of $G$ into a surface of Euler genus at most $\Delta^{O(1)} (\eg(G))^{O(1)} \log^{O(1)} n$.
We remark that there are graphs of maximum degree three that have Euler genus $\Omega(n)$.
Therefore, this result does not imply anything better that a $\Theta(n)$-approximation for the Euler genus of general graphs.
Mohar has obtained a $O(1)$-approximation for graphs $G$ that contain a vertex $a$ such that $G-a$ is planar and 3-connected (note that this is a special class of 1-apex graphs).
Finally, Makarychev, Nayyeri, and Sidiropoulos \cite{MNS12} obtained an
algorithm that given a Hamiltonian graph $G$ along with a Hamiltonian path $P$, computes an embedding of $G$ into a surface of Euler genus  $g^{O(1)} \log^{O(1)} n$ where $g$ is the orientable genus of $G$.
We remark that the Hamiltonicity assumption is a major restriction of this algorithm.
On the lower-bound side, Mohar \cite{Mo-apex} showed that computing $\eg(G)$ remains NP-hard even when the input is a 1-apex graph.
We emphasize that essentially no inapproximability result is known for $\eg(G)$, even on graphs of bounded degree.



\paragraph{Further algorithmic implications.}

Our result has a general consequence for the design of algorithms on graphs of small genus.
Most of the known algorithms for problems on such graphs require that an embedding of the graph is given as part of the input.
Our result implies that  many of these algorithms can be implemented even when the embedding is unknown.
Prior to our work, such a general reduction was known only for the case of graphs of bounded degree (see \cite{DBLP:conf/focs/ChekuriS13}).
As an illustrative example, consider the  Asymmetric TSP. For this problem, Erickson and Sidiropoulos \cite{DBLP:conf/compgeom/EricksonS14} recently gave a $O(\log g/\log\log g)$-approximation algorithm for graphs \emph{embedded} into a surface of Euler genus $g$.
Our result implies a polynomial-time algorithm with the same asymptotic approximation guarantee for the case of graphs that are \emph{embeddable} into a surface of Euler genus $g$ (that is, without having an embedding as part of the input).
We refer the reader to \cite{DBLP:conf/focs/ChekuriS13} for a more detailed discussion of such implications.

\subsection{Overview of the algorithm}

We now give a high-level description of our approach.
Our algorithm builds on the recent approximation algorithm for the bounded degree case, due to Chekuri and Sidiropoulos \cite{DBLP:conf/focs/ChekuriS13}.

\paragraph{Tools from the bounded-degree case.}
The approach from \cite{DBLP:conf/focs/ChekuriS13} is based on ideas from the fixed-parameter case and the theory of graph minors \cite{DBLP:conf/focs/KawarabayashiMR08, RS5, DBLP:journals/jct/RobertsonST94}.
However, we remark that the implementation of certain steps from the exact case is quite challenging due to the fact that the parameters are not fixed in the approximate setting.

The algorithm of \cite{DBLP:conf/focs/ChekuriS13} proceeds as follows.
First, while the input graph has sufficiently large treewidth, it finds a subgraph that can be removed without significantly affecting the solution.
This is done by computing a \emph{flat} grid minor.
Here, a planar subgraph $\Gamma$ of some graph $G$ is called \emph{flat} (w.r.t.~$G$) if there exists a planar drawing $\phi$ of $\Gamma$, such that for all edges $\{u,v\} \in E(\Gamma)$, with $u\in V(\Gamma)$, and $v\in V(G) \setminus V(\Gamma)$,  $u$ is on the outer face of $\phi$.
Eventually, they arrive at a graph of small treewidth.
For such a graph they can compute a small set of vertices $X$ whose removal leaves a planar graph.
Since in their case the degree is bounded, they can add $X$ back to the planar graph by introducing at most a constant number of handles for every vertex in $X$.
In summary, the approach of Chekuri and Sidiropoulos \cite{DBLP:conf/focs/ChekuriS13} reduces the problem of computing the genus of a graph to the following two sub-problems:

\begin{description}
\item{\textbf{Sub-problem 1: Computing flat grid minors.}}
Suppose that we are given a graph $G$ of genus $g$ and large treewidth, say $t>g^{c}$, for some sufficiently large constant $c$.
We wish to find a \emph{flat} subgraph that contains a $(c'g \times c'g)$-grid minor, for some sufficiently large constant $c'$.

\item{\textbf{Sub-problem 2: Embedding $k$-apex graphs.}}
Given a graph $G$ and some $X\subseteq V(G)$, such that $H=G\setminus X$ is planar, we wish to compute an embedding of $G$ into a surface of genus $g^{O(1)} \cdot |X|^{O(1)}$.
In general, there might be edges between the vertices in $X$.
We may remove all such edges, and add them to the final embedding by increasing the resulting genus by at most an additive factor $O(|X|^2)$, which does not affect our asymptotic bounds.
We may therefore assume in the rest of the this high-level overview that $X$ is an independent set.
\end{description}

Chekuri and Sidiropoulos \cite{DBLP:conf/focs/ChekuriS13} obtain  algorithms for both of these sub-problems. Indeed, the second problem is
trivial for them.
Unfortunately, since we are dealing with graphs of unbounded degree, their algorithms are not applicable in our case.
We next describe our algorithms for these sub-problems on general graphs.

\paragraph{Computing flat grid minors.}
Our algorithm for Sub-problem 1 follows an approach similar to the one used for the bounded-degree case in \cite{DBLP:conf/focs/ChekuriS13}.
We start by removing a small number of vertices that make the graph planar.
Since the original graph has large treewidth, it must also have a large grid minor.
The removal of a small number of vertices can only destroy a small part of this grid minor.
The main difficulty is to prove that some part of this remaining grid minor must be flat in the original graph.
We establish this property by arguing that if no such flat grid minor exists, then the graph must contain a $K_{3,b\cdot \eg(G)}$ minor, for some sufficiently large constant $b$, which contradicts the fact that the Euler genus of $G$ is $\eg(G)$.

\paragraph{Embedding $k$-apex graphs.}
We now discuss our algorithm for Sub-problem 2.
Recall that a graph $G$ is called $k$-apex if there exists some $X\subseteq V(G)$, with $|X|\leq k$, such that $H=G\setminus X$ is planar.
Our algorithm for approximating the Euler genus of $k$-apex graphs is the main technical contribution of this work.
Indeed, prior to our work, a similar algorithm was only known for special cases of $1$-apex graphs, and even the case of 2-apex graphs was completely open.
The problem is that each apex in $X$ may have many (e.g.~$\Omega(n)$) neighbors in $G\setminus X$. This makes a major difference between
our proof and the bounded-degree case in \cite{DBLP:conf/focs/ChekuriS13}.
This is because in the latter case, there is only a small number of edges between $X$ and $G-X$, so it is possible to add a handle for each edge. On the other hand, in our case, we cannot do this simply because we may have to add linearly many handles for each vertex in $X$. Most of the technical effort in this paper goes into bounding the number of handles added in this step.

Our algorithm for $k$-apex graphs proceeds in several steps.
At each step we simplify the graph via a sequence of operations.
Roughly speaking, every simplification operation either reduces the number of apices, or it simplifies the structure of the planar piece $H$.
Let us now describe the key ingredients of our approach in more detail.

\textbf{1. Simplification via vertex splitting.}
We introduce an operation called \emph{vertex splitting}.
This allows us to ``split'' a vertex of the planar piece into two vertices, as depicted in Figure \ref{fig:splitting_embedding}.
The benefit of this operation is that given an embedding of the new graph, we can efficiently compute an embedding of the original graph, without significantly increasing the Euler genus of the underlying surface (see Figure \ref{fig:splitting_embedding}).

\begin{figure}
\begin{center}
\scalebox{0.55}{\includegraphics{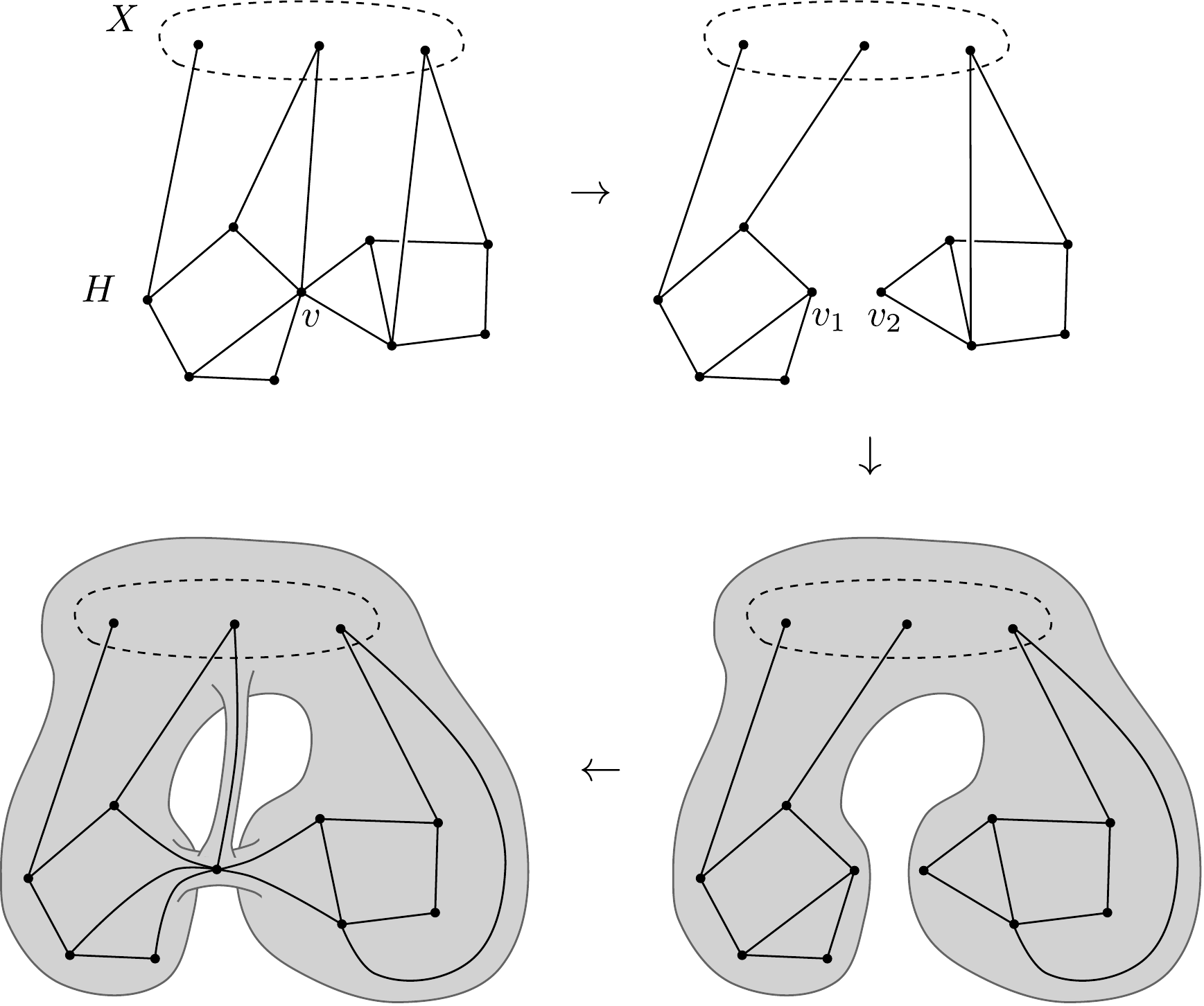}}
\caption{Example of performing a vertex $(H, \phi)$-splitting operation on a vertex $v$: The original graph $G$ with some $X\subseteq V(G)$ such that $H=G\setminus X$ is planar (top left), and the resulting graph $G'$ (top right).
Example of extending an embedding of $G'$ (bottom right) to an embedding of $G$ (bottom left).}
\label{fig:splitting_embedding}
\end{center}
\end{figure}

\textbf{2. Reduction to the 2-apex case.}
A key step in our algorithm is to reduce the problem of embedding $k$-apex graphs to the problem of embedding $2$-apex graphs.
This is done in several steps, by performing appropriate sequences of splitting operations.
We first compute a sequence of splitting operations such that in every resulting planar piece, every connected component is either incident to at most two apices, or every 1-separator is incident to at most one apex.
In the former case, we have obtained a 2-apex instance, which we show how to handle below.
In the latter case, we compute another sequence of splitting operations such that in the resulting graph, every component is either ``nearly locally 2-apex'' or has a 2-connected planar piece.
We shall deal with each one of these cases separately.

\textbf{3. Embedding $2$-apex graphs and their generalizations.}
Our algorithm for 2-apex graphs starts by decomposing the input graph into simpler pieces using a sequence of splitting operations.
In the resulting graph every piece is either 2-connected planar, or it has a 2-connected planar piece.
Therefore, the case of 2-apex graphs is reduced to the case of $2$-apex graphs with a 2-connected planar piece, which we  addressed below.
In reality, our algorithm has to embed graphs that can be more complicated than 2-apex graphs.
More specifically, we need to design an algorithm for embedding graphs that have at most one maximal 2-connected component that is $k$-apex, and a small number of remaining components (not necessarily 2-connected) that are 2-apex.
We call these graphs \emph{nearly locally 2-apex}.
Our algorithm for embedding these graphs uses similar ideas to the 2-apex case, but needs to perform a significantly more complicated decomposition step.

\textbf{4. Embedding $k$-apex graphs with a 2-connected planar piece.}
Next, we obtain an algorithm for embedding $k$-apex graphs where the planar piece $H$ is 2-connected.
\iffull
This is done by splitting any 2-connected $k$-apex graph of small genus into a small number of simpler structures, that we call \emph{centipedes} and \emph{butterflies} (see Figure \ref{fig:butterfly}).
\fi
These are special subgraphs with at most four apices.
Ultimately, the problem of embedding these graphs can be reduced to the problem of embedding 1-apex graphs, which we address below.

\textbf{6. Embedding 1-apex graphs.}
Finally, after the above sequence of reductions, we arrive at a small number of instances that are 1-apex.
Unfortunately, even the case of approximating the Euler genus of 1-apex graphs was open prior to our work.
In fact, the only previous result was a $O(1)$-approximation for the orientable genus of 1-apex graphs where the planar piece is 3-connected, due to Mohar \cite{DBLP:journals/jct/Mohar01}.
We remark that the 3-connectedness assumption simplifies the problem significantly.
We overcome this limitation by generalizing Mohar's argument using the theory of SPQR decompositions (for the definition of an SPQR decomposition and further exposition we refer the reader to \cite{DBLP:conf/focs/BattistaT89}).
Finally, since we are dealing with the Euler genus instead of orientable genus, we also have to extend Mohar's argument to the non-orientable case.

\textbf{7. Further complications: Extremities.}
In the above description of the key steps of our approach we have omitted certain complications that arise when splitting the input graph into simpler subgraphs.
More specifically, performing a splitting operation can occasionally create a certain number of components that are simple to embed.
We call such components \emph{extremities} (see Figure \ref{fig:extremity}).
This does not significantly affect our approach at the high level, but makes the statements of our intermediate reduction steps somewhat more technical.

\subsection{Organization}
The rest of the paper is organized as follows.
Section \ref{sec:prelim} gives some basic definitions and facts.
Section \ref{sec:algo} states formally the reduction of the problem to Sub-problem 1 (computing a flat grid minor) and Sub-problem 2 (embedding $k$-apex graphs), and presents our main algorithm.
Section \ref{sec:splitting} presents the vertex splitting operation and proves some of its basic properties.
It also formalizes the issues concerning extremities that arise when performing vertex splitting operations.
Section \ref{sec:k-apex} presents our algorithm for embedding $k$-apex graphs (i.e.~Sub-problem 1).
\iffull
This algorithm uses several other algorithms as sub-routines, which are discussed in subsequent Sections.
Section \ref{sec:2-connected} presents our algorithm for embedding $k$-apex graphs with a 2-connected planar piece.
Section \ref{sec:1-separators} introduces certain intermediate results that allow us to generalize the algorithm for a 2-connected planar piece; in particular, we present tools that allow us to simplify a graph by splitting along certain kinds of 1-separators.
These tools will later be used when dealing with nearly locally 2-apex graphs.
Section \ref{sec:2-apex} presents our algorithm for 2-apex graphs.
Section \ref{sec:2-apex_generalizations} combines the results from Sections \ref{sec:1-separators} and \ref{sec:2-apex} to obtain an algorithm for embedding generalizations of 2-apex graphs, namely, nearly locally 2-apex graphs.
Section \ref{sec:1-apex} presents our algorithm for embedding 1-apex graphs, generalizing Mohar's theorem on face covers \cite{Mo-apex}.
Section \ref{sec:flat_grids} presents our algorithm for computing a flat grid minor (i.e.~Sub-problem 1).
Finally, in Section \ref{sec:CS_summary}, for the sake of completeness, we give a high-level overview of the reduction from the general problem to Sub-problems 1 and 2, from the work of Chekuri and Sidiropoulos \cite{DBLP:conf/focs/ChekuriS13}.
\fi

\section{Preliminaries}\label{sec:prelim}



Before proceeding, we review some basic definitions and facts used throughout this paper.
We use $n$ to denote the number of vertices.
For basic graph theoretic definitions we refer the reader to the book
by Diestel \cite{diestel}, and for an in-depth treatment of topological graph theory, to the monograph by Mohar and Thomassen \cite{MT}.
We will only consider \DEF{$2$-cell embeddings} of graphs into surfaces; that is, we always assume that every face is homeomorphic to a disk.
Such embeddings can be
represented combinatorially by means of a \DEF{local rotation} and
\DEF{signature} (see \cite{MT} for details).
The local rotation and signature define a \DEF{rotation system}.

A \emph{biconnected component tree decomposition} of a given graph $G$ consists of a tree-decomposition $({\cal T},R)$
such that for every $\{t,t'\} \in E({\cal T})$, $R_t \cap R'_t$ consists of a single vertex and for every $t \in {\cal T}$, $R_t$ consists of a 2-connected graph (i.e., a block).
${\cal T}$ is called a \emph{biconnected component tree}.



For a vertex $v$ in a graph $G$ we write
$N_G(v) = \{u\in V(G) : \{u,v\} \in E(G)\}$.
For some $X\subseteq V(G)$ we denote by $G[X]$ the subgraph of $G$ induced on $X$.
For $X, Y \in V(G)$ with $X \cap Y =\emptyset$, $E(X,Y)$ denotes the set of edges with one endpoint in $X$ and the other endpoint in $Y$.

\iffull
Let $G$ be a graph, $v\in V(G)$, and let ${\cal C}$ be the set of connected components of $G\setminus v$.
We say that the graph $G' = \bigsqcup_{C\in {\cal C}} G[C\cup v]$ is obtained by \emph{cutting} $G$ along $v$, where $\sqcup$ denotes disjoint union.
For a set $X=\{v_1,\ldots,v_k\}\subseteq V(G)$ we say that a graph $G'$ is obtained by cutting $G$ along $X$ if there exists a sequence of graph $G_0,\ldots,G_k$ with $G_0=G$, $G_k=G'$, and such that for each $i\in \{1,\ldots,k\}$ the graph $G_i$ is obtained by cutting $G_{i-1}$ along $v_i$.
Let $C$ be a maximal 2-connected component of $G$.
\fi

We recall the following result on the genus of the complete bipartite graph \cite{harary1994graph}.

\begin{lemma}\label{lem:K33}
For any $n,m\geq 1$, $\eg(K_{m,n}) = \left\lceil(m-2)(n-2)/4\right\rceil$.
\end{lemma}

\iffull

We recall the following result due to Richter \cite{DBLP:journals/jct/Richter87} (see also \cite{DBLP:journals/combinatorica/DeckerGH85}, \cite{stahl1980permutation}).

\begin{lemma}\label{lem:2sum}
Let $G_1,\ldots,G_k$ be non-planar graphs, and for each $i\in \{1,\ldots,k\}$ let $\{s_i,t_i\}\in E(G_i)$.
Let $G$ be the graph obtained by identifying every $s_i$ into a vertex $s$, every $t_i$ into a vertex $t$, and every edge $\{s_i,t_i\}$ into an edge $\{s,t\}$.
Then, $\eg(G) = \Omega\left( \sum_{i=1}^k \eg(G_i) \right)$.
\end{lemma}

We shall also use the following elementary fact.

\begin{lemma}\label{lem:contract_U}
Let $G$ be a graph, and let $U\subseteq V(G)$.
Let $G'$ be the graph obtained by contracting $U$ into a single vertex.
Then, $\eg(G') \leq \eg(G)+|U|-1$ and $\genus(G')\leq \genus(G)+|U|-1$.
\end{lemma}

\begin{proof}
We will give only the proof for Euler genus, since the argument for orientable genus is identical.
Let $J$ be the graph obtained from $G$ by removing all edges with both endpoints in $U$.
We have $\eg(J) \leq \eg(G)$.
Let also $J'$ be the graph obtained from $J$ by adding the edges of a tree spanning $U$.
We have $\eg(J') \leq \eg(J) + |U|-1$.
Since $G'$ is a minor of $J'$, we have $\eg(G')\leq \eg(J') \leq \eg(G)+|U|-1$, as required.
\end{proof}

\fi

\section{The algorithm}\label{sec:algo}

In this section we present our algorithm for approximating the Euler genus of a graph.
We begin by stating formally the reduction from the general problem to Sub-problems 1
and 2,
due to  \cite{DBLP:conf/focs/ChekuriS13}.
For a graph $G$ and a minor $\Gamma$ of $G$ we say that a mapping $\mu:V(\Gamma)\to 2^{V(G)}$ is a \emph{minor mapping} for $\Gamma$ if for each $u\in V(\Gamma)$ the graph $G[\mu(u)]$ is connected, and by contracting $G[\mu(u)]$ into a single vertex for each $u\in V(\Gamma)$, we obtain $\Gamma$.
The reduction can now be stated as follows.

\begin{lemma}[Chekuri and Sidiropoulos \cite{DBLP:conf/focs/ChekuriS13}]\label{lem:CS_summary}
Suppose that the following conditions hold:
\begin{description}
\item{(1)}
There exists a polynomial-time algorithm  which given a $n$-vertex graph $G$ of treewidth $t$ and an integer $g\geq 1$, either correctly decides that $\eg(G)>g$, or  outputs a flat subgraph $G'\subset G$, such that $X$ contains a $\left(\Omega(r)\times \Omega(r)\right)$-grid minor $M$, for some $r=r(n,g,t)$.  Moreover, in the latter case, the algorithm also outputs a minor mapping for $M$.
\item{(2)}
Given an $n$-vertex graph $G$, an integer $g$, and some $X\subset V(G)$ such that $G\setminus X$ is planar, either correctly decides that $\eg(G) > g$, or it outputs a drawing of $G$ into a surface of Euler genus at most $\gamma$, for some $\gamma=\gamma(n, g, |X|)$.
\end{description}
Then there exists a polynomial-time algorithm which given a $n$-vertex graph $G$ and an integer $g\geq 1$, either correctly decides that $\eg(G)>g$ or outputs an embedding of $G$ into a surface of Euler genus
at most $\gamma(n, g, k) + k$, for some $k=O(t' g  \log^{3/2}n)$, where $t'$ is some integer satisfying
$r(n,g,t')=O(g)$.
\end{lemma}

\iffull
For the sake of completeness, we include an overview of the proof of Lemma \ref{lem:CS_summary} in Section \ref{sec:CS_summary}.
\fi

The next lemma states our result for computing a flat grid minor.

\begin{lemma}[Computing a flat grid minor]\label{lem:flat_grid}
There exists a polynomial-time algorithm which
  given a graph $G$ of treewidth $t$, and
  an integer $g\geq 1$,
  either correctly decides that $\eg(G)>g$, or it outputs a flat
  subgraph $G'\subset G$, such that $G'$ contains a $\left(\Omega(r)
    \times \Omega(r)\right)$-grid minor $M$, for some
  $r=\Omega\left(\frac{t^{1/2}}{g^{4} \log^{15/4}
      n}\right)$.  In the latter case, the algorithm also
  outputs a minor mapping for $M$.
\end{lemma}

The next theorem gives our approximation algorithm for embedding $k$-apex graphs.
This is the main technical result of our paper.
The proof is discussed in subsequent sections.

\begin{theorem}\label{lem:vertex_insertion}
Let $G$ be a graph of Euler genus $g$ and let $X\subseteq V(G)$ such that $H=G\setminus X$ is planar.
Then there exists a polynomial-time algorithm which given $G$, $g$, and $X$, outputs an embedding of $G$ into a surface of Euler genus $O(g^{25} \cdot |X|^{21})$.
\end{theorem}

Given the above results, we are now ready to prove our main Theorem.


\begin{proof}[Proof of Theorem \ref{thm:main1}]
The algorithm given by Lemma \ref{lem:flat_grid} satisfies condition (1) of Lemma \ref{lem:CS_summary} with $r(n,g,t) = \Omega\left(\frac{t^{1/2}}{g^{4} \log^{15/4}
      n}\right)$.
The algorithm given by Lemma \ref{lem:vertex_insertion} satisfies condition (2) of Lemma \ref{lem:CS_summary}  with $\gamma(n,g,|X|) = O(g^{25} \cdot |X|^{21})$.
Let $t'$ be some integer satisfying
$r(n,g,t')=O(g)$.
We have $t'= O(g^{10}  \log^{15/2} n)$.
It now follows by Lemma \ref{lem:CS_summary} that there exists a polynomial-time algorithm that either correctly decides that $\eg(G)>g$ or outputs an embedding of $G$ into a surface of Euler genus
at most $g'=\gamma(n, g, k) + k$, for some $k=O(t' g  \log^{3/2}n) = O(g^{11} \log^{9}n)$.
Therefore, $g'=O(g^{25} \cdot k^{21}) = O(g^{25} \cdot (g^{11} \log^{9}n)^{21}) =  O(g^{256} \log^{189}n)$, completing the proof.
\end{proof}

\section{Vertex splitting}\label{sec:splitting}

We now formalize the notion of \emph{vertex splitting} that was discussed in the Introduction.

\begin{definition}[Splitting]
Let $H$ be a planar graph and let $\phi$ be a drawing of $H$ into the plane.
Let $v\in V(H)$.
Let $E^*$ be the set of edges incident to $v$.
Consider the circular ordering $\tau$ of $E^*$ induced by $\phi$.
Partition $E^*$ into $E^*=E_1 \cup E_2$ where for each $i\in \{1,2\}$ the edges in $E_i$ form a contiguous subsequence in $\tau$.
Let $H'$ be the graph obtained from $H$ by removing $v$ and introducing two new vertices $v_1,v_2$.
For each $e=\{v,w\}\in E_i$ we add the edge $\{v_i,w\}$ in $H'$.
This operation is called a \emph{$(H,\phi)$-splitting}.
We also say that the splitting operation is \emph{performed on $v$ with partition} $\{E_1,E_2\}$.
If $H$ is a subgraph of some graph $G$, then when performing a splitting operation on $v\in V(H)$ we also remove all edges in $E(G)$ between $v$ and $V(G)\setminus V(H)$ (see Figure \ref{fig:splitting_embedding} for an example).
\end{definition}


\begin{definition}[Splitting sequence]
Let $H$ be planar graph and let $\phi$ be a drawing of $H$ into the plane.
A sequence $\sigma=p_1,\ldots,p_t$ where each $p_i$ is a $(H,\phi)$-splitting is called a \emph{$(H,\phi)$-splitting sequence}.
The result of performing $\sigma$ on $H$ is a graph $H'$ defined as follows.
For every $v\in V(H)$ let $P_v$ be the set of all splittings in $\sigma$ that are performed on $v$.
Let $E_v$ be the set of edges incident to $v$ in $H$.
For each $p\in P_v$ let $\{E_{v,1}^p, E_{v,2}^p\}$ be the partition of $E_v$ induced by the splitting $p$.
Let $\{E_{v,1},\ldots,E_{v,r_v}\}$ be the common refinement of all these partitions of $E_v$.
Note that $r_v\leq 2 |P_v|$.
We remove $v$ and we add vertices $v_1,\ldots,v_{r_v}$ in $H'$.
For every $e=\{v,u\} \in E_{v,i}$ we add the edge $\{v_i,u\}$ in $H'$.
Repeating this process for all $v\in V(H)$ concludes the definition of $H'$.
\end{definition}

\begin{definition}[Fragment]
Let $H$ be a planar graph and let $\phi$ be a planar drawing of $H$.
Let $\sigma$ be a $(H,\phi)$-splitting sequence.
Let $H'$ be the graph obtained by performing $\sigma$ on $H$.
Let ${\cal C}$ be the set of connected components of $H$ that are also connected components of $H'$.
In other words, ${\cal C}$ contains all connected components of $H$ on which $\sigma$ does not perform any splittings.
Let also ${\cal C}'$ be the remaining connected components of $H$.
For any $C\in {\cal C}'$ let ${\cal F}_C$ be the collection of connected components in $H'$ that are obtained from $C$ after performing $\sigma$.
Let
${\cal D} = \bigcup_{C\in {\cal C}'} {\cal F}_C$.
Then we refer to the elements in ${\cal C}\cup {\cal D}$ as the \emph{fragments} of $H'$ (w.r.t.~$\sigma$).
Note that if $\sigma$ has length $k$, then there are at most $k+2$ fragments of $H'$.
\end{definition}

\iffull
\begin{definition}[Crossing partitions]
Let $U$ be a set and let $\{A,B\}$, $\{A',B'\}$ be partitions of $U$.
We say that the partitions are \emph{crossing} if their common refinement consists of four non-empty sets.
Otherwise we say that they are \emph{non-crossing}.
\end{definition}

\begin{definition}[Monotone splitting sequence]
Let $H$ be a planar graph and let $\phi$ be a drawing of $H$ into the plane.
A $(H,\phi)$-splitting sequence $\sigma$ is called \emph{monotone} if the following condition is satisfied.
Let $v \in V(H)$ be an arbitrary vertex.
Let $E_v$ be the set of edges incident to $v$ in $H$.
Then, for any pair $p,q\in \sigma$ of splittings performed   on $v$ the partitions of $E_v$ induced by $p$ and $q$ are non-crossing.
\end{definition}

Having defined monotonicity, we need the following result.

\begin{lemma}[Enforcing monotonicity on a splitting sequence]\label{lem:monotone_splitting}
Let $H$ be a planar graph and let $\phi$ be a drawing of $H$ into the plane.
Let $\sigma$ be a $(H,\phi)$-splitting sequence of length $m$.
Let $H'$ be the graph obtained by performing the $(H,\phi)$-splitting sequence $\sigma$.
Then, there exists a monotone $(H,\phi)$-splitting sequence $\sigma'$ of length at most $m$, such that the graph resulting by performing $\sigma'$ on $H$ is $H'$.
\end{lemma}
\begin{proof}
Let $v\in V(H)$, and let $E_v\subset E(H)$ be
the set of edges in $H$ that are incident to $v$.
Let $\xi$ be the circular ordering of $E_v$ around $v$ induced $\phi$.
Suppose that at most $t$ splittings in $\sigma$ are performed in $v$.
Then, it follows by induction on $t$ that $v$ corresponds to some $S_v\subseteq V(H')$, with $|S_v|\leq t$.
Moreover, for each $v'\in S_v$, the edges incident to $v'$ in $H'$ correspond to a subset $E_{v,v'}\subseteq E_v$, where $E_{v,v'}$ forms a contiguous segment in $\xi$.
We may therefore add for each $v'\in S_v$ a splitting in $\sigma'$ with partition $\{E_{v,v'}, E_v\setminus E_{v,v'}\}$.
It is immediate that all these splittings are monotone.
Repeating the same process for all $v\in V(H)$, we obtain the desired splitting sequence $\sigma$.
It is also immediate that for every $v\in V(H)$, the number of splittings in $\sigma'$ that are performed on $v$ is at most the number of such splittings in $\sigma$.
Therefore, the length of $\sigma'$ is at most $k$.
This concludes the proof.
\end{proof}

The following lemma shows that performing a small number of splitting operations does not increase the genus significantly.

\begin{lemma}\label{lem:genus_splittings}
Let $G$ be a graph and let $X\subseteq V(G)$ such that $H=G\setminus X$ is planar.
Let $\phi$ be a planar drawing of $H$.
Let $G'$ be the graph obtained from $G$ after performing a $(H,\phi)$-splitting sequence of length $k$ on $H$.
Then $\genus(G') \leq \genus(G) + O(k \cdot |X|)$, and $\eg(G') \leq \eg(G) + O(g\cdot |X|)$.
Moreover, given a drawing of $G'$ into a surface of orientable (resp.~non-orientable) genus $\gamma$, we can compute a drawing of $G$ into a surface of orientable (resp.~non-orientable) genus $\gamma + O(k\cdot |X|)$ in polynomial time.
\end{lemma}

\begin{proof}
By Lemma \ref{lem:monotone_splitting} we may assume that $\sigma$ is monotone.
Let $H'$ be the graph obtained by performing the splitting sequence $\sigma$ on $H$.
Each $v\in V(H)$ corresponds to some $S_v\subseteq V(H')$.
For every $v\in V(G)$, let $E_v$ be the set of edges in $H$ that are incident to $v$.
Let $\xi$ be the circular ordering of $E_v$ around $v$ induced by $\phi$.
Since $\sigma$ is monotone, it follows that each $s\in S_v$ is incident in $H'$ to some set of edges $E_{v,s} \subseteq E_v$, that appear in a contiguous subsegment of the circular ordering $\xi$.
By monotonicity, and by induction on the length of the splitting sequence $\sigma$, it follows that when we perform the $i$-splitting in $\sigma$, some vertex $v'_i$ in the current graph that corresponds to some $v_i$ is split into two vertices $v'_{i,1}$ and $v'_{i,2}$.
We set $v'_{i,1}$ to be the parent of $v'_{i,2}$.
By monotonicity of $\sigma$, it follows that this parent relationship defines a tree in $H'$.
Starting from $H'$, we add all the edges of all these trees and for all vertices $v\in V(H)$.
Let $H''$ be the resulting graph obtained from $H'$.
Let also $G''$ be the corresponding graph obtained from $G'$.
It is immediate that $H''$ is obtained by adding at most $k$ edges in $H'$.
Therefore, $\eg(G'') \leq \eg(G') + k$.
Let $G'''$ be the graph obtained by contracting each one of these trees into a single vertex.
Finally, for every spitting operation performed on some $v\in V(H)$, we might remove at most $|X|$ edges between $v$ and $X$.
Therefore, adding at most $k\cdot |X|$ edges to $G'''$, we obtain the graph $G$.
Therefore,
$\eg(G) \leq \eg(G''') + k\cdot |X| \leq \eg(G'') + k\cdot |X| \leq \eg(G') + (k+1)\cdot |X|$.
Finally, let $\psi$ be a drawing of $G$ into a surface ${\cal S}$ of Euler genus $\gamma$.
By adding at most $O(k\cdot |X|)$ handles in ${\sigma'}$, one for every new edge in $G'''$, we can extend $\psi$ into a drawing $\psi''$ of $G''$ into a surface of Euler genus $\gamma+O(k\cdot |X|)$, concluding the proof.
\end{proof}

\fi

The following lemma shows how to compute an embedding of a graph given an embedding of the graph obtained after performing a sequence of splitting operations.

\begin{lemma}[Glueing fragmented embeddings]\label{lem:glueing_fragmented}
Let $G$ be a graph and let $X\subseteq V(G)$ be an independent set such that $H=G\setminus X$ is planar (and possibly disconnected).
Let $\phi$ be a planar drawing of $H$.
Let $\sigma$ be a $(H,\phi)$-splitting sequence of length $k$.
Let $H'$ be the graph obtained by performing $\sigma$ on $H$.
Let ${\cal F}$ be the set of fragments of $H'$.
Suppose that for any $C\in {\cal F}$ there exits an embedding $\psi_C$ of $G'[V(C)\cup X]$ into a surface ${\cal S}_C$ of Euler genus $\gamma_C$.
Then, there exists an embedding $\psi$ of $G$ into a surface of Euler genus at most $O(k\cdot |X|) + \sum_{C\in {\cal F}} \gamma_C$.
Moreover, there exists a polynomial-time algorithm which given $G$, $X$, $\phi$, $\sigma$, and $\{\psi_C\}_{C\in {\cal F}}$ outputs $\psi$ (see Figure \ref{fig:splitting_embedding}).
\end{lemma}

\iffull

\begin{proof}
Let $G'$ be the graph obtained from $G$ after applying $\sigma$ on $H$.
For any $C,C'\in {\cal F}$, $N_{G'}(C)\cap N_{G'}(C')\subseteq X$.
Therefore, we can extend $\psi_C$ to $C'$ by adding at most $|X|$ cylinders, each connecting a puncture in the surface ${\cal S}_C$ (that accommodates $C$) to a puncture in the surface ${\cal S}_{C'}$ (that accommodates $C'$).
The resulting surface has Euler genus at most $\eg({\cal S}_C) + \eg({\cal S}_{C'}) + |X|$.
Using the same procedure we can inductively extend this embedding to  all the fragments of $H'$.
For every fragment $C''$ we increase the genus of the embedding by at most $\eg({\cal S}_{C''}) + |X|$.
Since $\sigma$ has length $k$, the number of fragments of $H'$ is at most $k+2$.
Thus we obtain an embedding $\psi'$ of $G'$ into a surface of Euler genus at most $\gamma'=\sum_{C\in {\cal F}} (\eg({\cal S}_C) + |X|)  = O(k\cdot |X|) + \sum_{C\in {\cal F}} \gamma_C$.

Finally, using Lemma \ref{lem:genus_splittings} we can compute an embedding of $G$ into a surface of Euler genus at most $\gamma' + O(k\cdot |X|) = O(k\cdot |X|) + \sum_{C\in {\cal F}} \gamma_C$, concluding the proof.
\end{proof}

\fi


\subsection{Extremities}

We now introduce some machinery that allows us to handle some issues that arise from splitting operations.
More specifically, when performing a spitting operation, we might create a certain number of pieces that are easy to embed, called \emph{extremities}. We formalize this notion next.

\begin{definition}[Extremity]
Let $G$ be a graph and let $X\subseteq V(G)$ such that $H=G\setminus X$ is planar.
Let ${\cal C}$ be a collection of maximal 2-connected components of $H$ such that $C=\bigcup_{A\in {\cal C}} A$ is connected.
Suppose further that there exist some 1-separator $v$ of $H$ such that all edges between $C$ and $H\setminus C$ are incident to $v$.
Moreover, suppose that there exists $x\in X$ such that $N_G[V(C) \setminus \{v\}] \cap X \subseteq \{x\}$.
Finally, assume that that $G[C\cup \{x\}]$ admits a planar drawing such that $x$ and $v$ are in the same face.
Then we say that $C$ is an \emph{extremity} (w.r.t. $X$).
Figure \ref{fig:extremity} depicts an example of an extremity.
We refer to $v$ as the \emph{portal} of $C$.
The \emph{extremity number} of $G$ is defined to be the minimum integer $M$ such that any family of pairwise edge-disjoint maximal extremities of $G$ has size at most $M$.
\end{definition}

\begin{figure}
\begin{center}
\scalebox{0.55}{\includegraphics{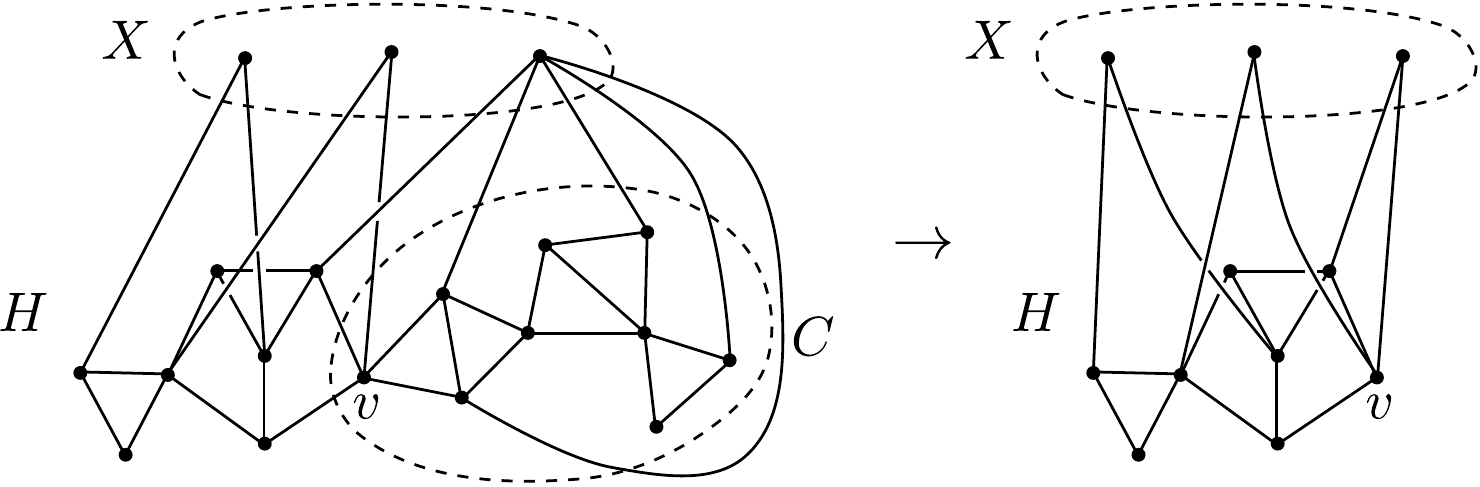}}
\caption{Example of an extremity $C$ w.r.to some set $X$, with portal $v$ (left) and the graph obtained by contracting $C$ into its portal (right).}
\label{fig:extremity}
\end{center}
\end{figure}

The next lemma allows us to compute embeddings of graphs, while ignoring their extremities.
\iffull
In other words,
if we obtain an embedding of a resulting graph (after contracting extremities), then we can extend to the embedding of the original graph
in the same surface.
\fi

\begin{lemma}[Contracting extremities]\label{lem:contracting_extremities}
Let $G$ be a graph and let $X\subseteq V(G)$ such that $H=G\setminus X$ is planar.
Let $C_1,\ldots,C_t$ be a collection of pairwise edge-disjoint extremities.
Let $G'$ be the graph obtained by contracting each $C_i$ into a single vertex $v_i$ and removing parallel edges (see Figure \ref{fig:extremity}).
Then there exists a polynomial-time algorithm which given an embedding of $G'$ into a surface of Euler genus $\gamma$ outputs an embedding of $G$ into a surface of Euler genus $\gamma$.
\end{lemma}

\iffull
\begin{proof}
Suppose that $\phi'$ is an embedding of $G'$ into some surface ${\cal S}'$.
Consider some extremity $C_i$.
If $C_i=H$, then the assertion is immediate, so suppose that $C_i\neq H$.
Let $v$ be the portal of $C_i$.
We have either $N(C_i\setminus \{v\})\cap X=\emptyset$ or $\{u\}=N(C_i\setminus \{v\})\cap X$.
In the later case, there exists at most one edge $e$ in $G'$ between $u$ and $v$ that corresponds to the edges between $u$ and $V(C_v)\setminus \{v\}$ in $G$.
Let ${\cal D}$ be a disk in ${\cal S}'$ that intersects $\phi'(G')$ only on $\phi'(v)$.
Moreover, if $e$ is defined, then we can ensure that ${\cal D}$ lies inside a face that contains $e$.
By choosing the disks for distinct extremities to be sufficiently small, we can further ensure that they are pairwise disjoint.
By the definition of an extremity it follows that there exists a planar drawing $\psi_{C_i}$ of $G[C_i\cup \{u\}]$ in which $u$ and $v$ are in the same face $F$.
We extend $\phi'$ to $H[C_i]$ by mapping $H[C_i]$ into ${\cal D}$ according to $\psi_{C_i}$, so that the face bounded by $F$ contains the boundary of ${\cal D}$.
We can also extend the embedding to the edges in $G$ between $u$ and $C_i$ by mapping them along curves that are contained in a sufficiently small neighborhood of $\phi'(e)\cup {\cal D}$.
Repeating the same process for all extremities $C_i$ results in the desired embedding $\phi$.
\end{proof}
\fi


Finally, we argue that splitting operations cannot create a significant number of extremities.

\iffull
\begin{lemma}[Creating extremities via a splitting]\label{lem:creating_extremities_splitting}
Let $G$ be a graph and let $X\subseteq V(G)$ such that $H=G\setminus X$ is planar.
Suppose that the extremity number of $G$ if $M$.
Let $\phi$ be a planar drawing of $\phi$.
Let $H'$ be the graph obtained by performing a splitting operation on some 1-separator of $H$, and let $G'$ be the corresponding graph obtained from $G$.
Then the extremity number of $G'$ is at most $M+2\cdot |X|$.
\end{lemma}

\begin{proof}
Suppose we perform a splitting operation on some 1-separator $v$ of $H$.
Then $v$ corresponds to exactly copies $v_1$ and $v_2$ in $H'$.
Any maximal extremity in $G'$ that is not an extremity in $G$ must contain either $v_1$ or $v_2$.
By the definition of an extremity, it follows any two new maximal and edge-disjoint extremities in $H'$ that contain $v_1$ must be incident to distinct vertices in $X$ (otherwise their union must also be an extremity, contradicting their maximality).
Therefore, there can be at most $|X|$ distinct new maximal and pairwise edge-disjoint extremities that contain each $v_i$, $i\in \{1,2\}$.
Thus, there can be at most $M+2\cdot |X|$ extremities in $G'$, concluding the proof.
\end{proof}
\fi

\begin{lemma}[Creating extremities via a splitting sequence]\label{lem:creating_extremities_splitting_sequence}
Let $G$ be a graph and let $X\subseteq V(G)$ such that $H=G\setminus X$ is planar.
Suppose that the extremity number of $G$ is $M$.
Let $\phi$ be a planar drawing of $\phi$.
Let $H'$ be the graph obtained by performing a splitting sequence $\sigma$ of length $\ell$ on $H$, where each splitting operation in $\sigma$ is performed on some 1-separator of $H$,
and let $G'$ be the corresponding graph obtained from $G$.
Then the extremity number of $G'$ is at most $M+2\ell\cdot |X|$.
\end{lemma}

\iffull
\begin{proof}
If follows immediately from Lemma \ref{lem:creating_extremities_splitting} and induction on $\ell$.
\end{proof}
\fi

\section{Embedding $k$-apex graphs}
\label{sec:k-apex}

In this section we present our algorithm for embedding $k$-apex graphs.
\iffull
This uses several other algorithms as sub-routines, that we discuss in subsequent sections.
\fi
We first state a preliminary result that allows us to assume that every vertex of the planar piece is incident to at most two apices.

\begin{lemma}\label{lem:all3}
Let $G$ be a graph of genus $g$ and let $X\subset V(G)$ such that $H=G\setminus X$ is planar.
Let $x_1,x_2,x_3\in X$ be distinct vertices.
Let
$U = V(H) \cap N_G(x_1) \cap N_G(x_2) \cap N_G(x_3)$.
Then $|U|=O(g)$.
\end{lemma}

The following lemma allows us to reduce the $k$-apex case to two sub-cases: (1) 2-apex graphs, and (2) $k$-apex graphs where every 1-separator of the planar piece is incident to at most one apex.

\begin{lemma}\label{lem:2apices_or_simple1separators}
Let $G$ be a graph of Euler genus $g$ and let $X\subset V(G)$ be an independent set such that $H=G\setminus X$ is planar.
Let $\phi$ be a planar drawing of $H$.
Suppose that every vertex $v\in V(H)$ is incident to at most two vertices in $X$, that is $|N(v) \cap X|\leq 2$.
Then there exists a $(H,\phi)$-splitting sequence $\sigma$ of length $O(g\cdot |X|^3)$ such that if we let $H'$ be the graph obtained by performing $\sigma$ on $H$,
then for every connected component $C$ of $H'$ at least one of the following conditions is satisfied:

\iffull
\begin{description}
\item{(1)}
$C$ is incident to at most two vertices in $X$, that is
$|N_G(C)\cap X| \leq 2$.

\item{(2)}
Every $1$-separator $v$ of $C$ is incident to at most one vertex in $X$, that is $|N(v)\cap X| \leq 1$.
\end{description}
\fi

\end{lemma}

The next two lemmas summarize our embedding algorithms for the above two cases.

\begin{lemma}[Embedding 2-apex graphs]\label{lem:embedding_2-apex}
Let $G$ be a planar graph
and let $X\subseteq V(G)$, with $|X|=2$, such that $H=G\setminus X$ is planar.
Then there exists a polynomial-time algorithm which given $G$ and $X$, outputs a drawing of $G$ into a surface of Euler genus $O((\eg(G))^{15})$.
\end{lemma}

\begin{lemma}[Embedding graphs with simple 1-separators]\label{lem:embedding_1-apex_1-separators}
Let $G$ be a graph of genus $g$ and let $X\subseteq V(G)$ such that $H=G\setminus X$ is planar.
Suppose that every 1-separator $v$ of $H$ is incident to at most one vertex in $X$, that is $|N_G(v)\cap X|\leq 1$.
Suppose further that the extremity number of $G$ is $M$.
Then there exists a polynomial-time algorithm which given $G$, $g$, $M$, and $X$ outputs an embedding of $G$ into a surface of Euler genus $O(g^{24} \cdot |X|^{18} + g^{22} \cdot |X|^{14} \cdot M)$.
\end{lemma}

Given all the above results, we are  ready to present our algorithm for embedding $k$-apex graphs.

\begin{proof}[Proof of Theorem \ref{lem:vertex_insertion}]
The algorithm consists of the following steps.

\textbf{Step 1: Deleting edges between vertices in $X$.}
There can be at most $O(|X|^2)$ edges between the vertices in $X$. We remove all such vertices, and we extend the drawing at the end to include these edges by adding at most $O(|X|^2)$ additional handles.
We can therefore assume for the remainder that there are no edges between the vertices in $|X|$.

\textbf{Step 2: Contracting extremities.}
We compute a maximal collection ${\cal E}$ of pairwise edge-disjoint maximal extremities in $G$.
For every extremity $C\in {\cal E}$ let $v_C$ be its portal.
There exits $x_C\in X$ such that $N_G(V(C)\setminus \{v_C\}) \cap X \subseteq \{x_C\}$.
We contract $C$ into $v_C$, thus replacing $G[C\cup \{x_C\}]$ by a single edge $e_C$ between $x_C$ and $v_C$.
Repeating for all $C\in {\cal E}$, we obtain a graph $G'$ with extremity number 0.
We will show next how to compute an embedding for $G'$ into a surface of Euler genus $\gamma$.
Given this embedding we can extend it to $G$ using Lemma \ref{lem:contracting_extremities}, obtaining an embedding of $G$ into a surface of Euler genus $\gamma$.
We may therefore assume for the remainder of the algorithm that the extremity number of $G$ is 0.

\textbf{Step 3: Bounding the number of apices that are incident to any vertex in $H$.}
By Lemma \ref{lem:all3} for any three distinct vertices $x_1,x_2,x_3\in X$ there can be at most $O(g)$ vertices in $V(H)$ that are incident to all of $x_1,x_2,x_3$.
Therefore, there can be at most $O(g\cdot |X|^3)$ vertices in $V(H)$ that are incident to at least three vertices in $X$.
For every such a vertex $v$ we remove all except for two edges between $v$ and  $X$.
The total number of edges we remove is  $O(g\cdot |X|^4)$.
Since we leave at least two edges between $v$ and $X$, it follows that the extremity number of $G$ remains 0.
We next show how to compute an embedding for the resulting graph.
We will extend the drawing to these edges at the end of the algorithm, by adding at most $O(g\cdot |X|^4)$ additional handles.
We may therefore assume for the remainder that every $v\in V(H)$ is incident to at most two vertices in $X$.

\textbf{Step 4: Splitting into pieces that are either 2-apex or contain only simple 1-separators.}
Fix any planar drawing $\phi$ of $H$.
By Lemma \ref{lem:2apices_or_simple1separators} we can find a $(H,\phi)$-splitting sequence $\sigma$ of length $k=O(g\cdot |X|^3)$, such that for every fragment $C$ of the resulting graph $H'$, either $C$ is incident to at most two vertices in $X$, or every 1-separator in $C$ is connected to at most one apex.
By Lemma \ref{lem:creating_extremities_splitting_sequence}, the extremity number of the resulting graph is at most $M=O(g\cdot |X|^4)$.

\textbf{Step 5: Embedding the fragments.}
Let ${\cal F}$ be the set of fragments of $H'$.
We have $|{\cal F}|=O(k)=O(g\cdot |X|^3)$.
Consider some $C\in {\cal F}$.
If $C$ is incident to at most two vertices in $X$, then by Lemma \ref{lem:embedding_2-apex} we can compute an embedding of $G[V(C)\cup X]$ into a surface of Euler genus $O(g^{15})$.
Otherwise, every 1-separator in $C$ is incident to at most one vertex in $X$.
Thus, by Lemma \ref{lem:embedding_1-apex_1-separators} we can compute an embedding of $G[V(C)\cup X]$ into a surface of Euler genus $O(g^{24}\cdot |X|^{18} + g^{22} \cdot |X|^{14} \cdot M) = O(g^{24}\cdot |X|^{18})$.
Thus in either case, for each $C\in {\cal F}$ we can compute an embedding $\phi_C$ of $G[V(C)\cup X]$ into a surface of Euler genus $O(g^{24}\cdot |X|^{18})$.

\textbf{Step 6: Combining the embeddings of all fragments into a single embedding.}
By Lemma \ref{lem:glueing_fragmented} we can combine all the embeddings $\{\phi_C\}_{C\in {\cal F}}$ to obtain an embedding $\phi$ of $G$ into a surface of Euler genus $O(k\cdot |X|) + \sum_{C\in {\cal F}} \eg(\phi_C) = O(g\cdot |X|^4) + |{\cal F}| \cdot O(g^{24} \cdot |X|^{18}) = O(g^{25} \cdot |X|^{21})$.

This concludes the proof.
\end{proof}


\iffull

\section{Embedding graphs with a 2-connected planar piece}\label{sec:2-connected}


In this section we present the algorithm for embedding $k$-apex graphs with a 2-connected planar piece.
Our approach is to decompose the graph into a small number of graphs with a constant number of apices.
The key ingredient is the notion of \emph{coupled} edge sets, which we now introduce.

\begin{definition}[Coupled subsets of edges]
Let $G$ be a graph and let $X\subseteq V(G)$ and $H=G\setminus X$.
Let $x_1,x_2\in X$ be distinct vertices.
Let $P$ be a path in $H$.
We say that $E'$ is \emph{$(P, \{x_1,x_2\})$-coupled} (or just $P$-coupled when $x_1$ and $x_2$ are clear from the context) if  the following conditions are satisfied:
\begin{description}
\item{(1)}
$E' \subseteq E(V(P), \{x_1,x_2\})$.
\item{(2)}
For every internal vertex $v$ of $P$ $E(\{v\},\{x_1,x_2\}) \subseteq E'$.
\end{description}
Figure \ref{fig:coupled} depicts an example of a coupled set of edges.
\end{definition}

\begin{figure}
\begin{center}
\scalebox{0.65}{\includegraphics{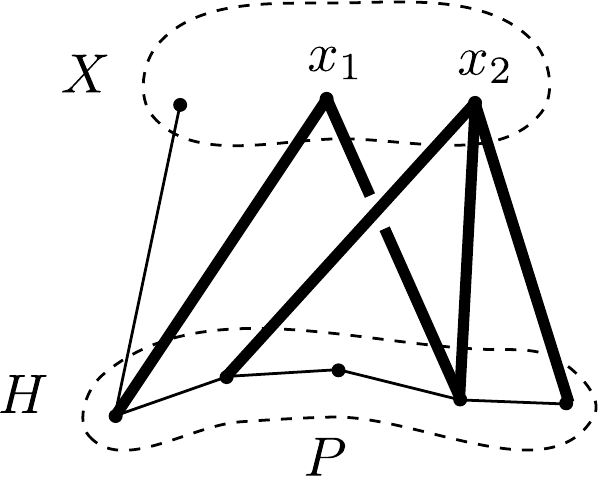}}
\caption{A $(P, \{x_1,x_2\})$-coupled set of edges (in bold).}
\label{fig:coupled}
\end{center}
\end{figure}

Intuitively, we will give a decomposition of the edges between the planar piece and the apex set into a small number of coupled sets.
In order to do that, we need the notion of \emph{interleaving number}.
This allows us to argue that when the decomposition fails, we can find a large $K_{3,r}$ minor, which contradicts the fact that the Euler genus is small.

\begin{definition}[Interleaving number]
Let $G$ be a graph and let $X\subseteq V(G)$ be such that $H=G\setminus X$ is planar.
Let $x_1,x_2,x_3\in X$ be distinct vertices.
Let $P$ be a path in $H$.
Let $E' = E(V(P), \{x_1,x_2,x_3\})$.
The $(x_1,x_2,x_3)$-interleaving number of $P$ is defined to be the minimum integer $k$ such that the following holds:
There exists a collection of edge-disjoint subpaths $P_1,\ldots,P_k$ of $P$ and a decomposition $E'=E_1\cup\ldots\cup E_k$ such that for each $i\in \{1,\ldots,k\}$, $E_i$ is $P_i$-coupled (in particular, it follows that there exists $y_i \in \{x_1, x_2, x_3\}$ such that no edge in $E_i$ is incident to $y_i$).
See Figure \ref{fig:interleaving} for an example.
\end{definition}

\begin{figure}
\begin{center}
\scalebox{0.65}{\includegraphics{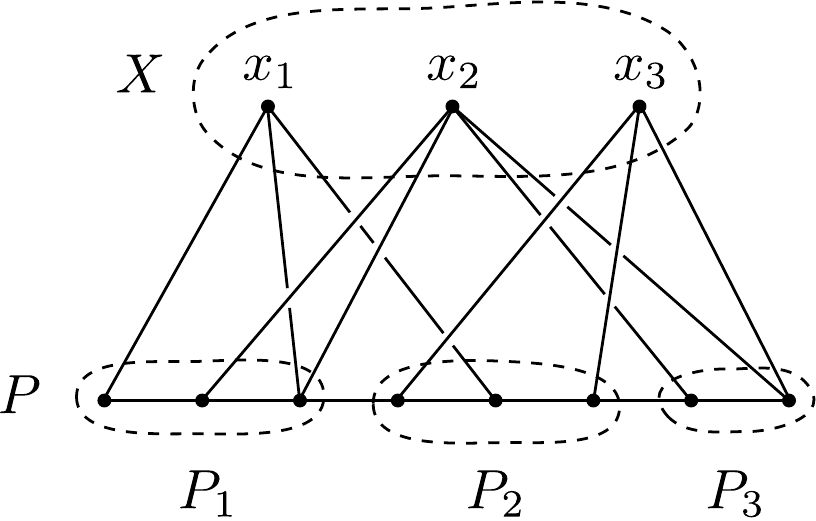}}
\caption{Example of a path with interleaving number 3, and the corresponding collection of subpaths paths $P_1,P_2,P_3$.
We remark that for the sake of clarity we have chosen an example where the paths $P_1,P_2,P_3$ are vertex-disjoint. In general, these paths are only required to  be edge-disjoint.}
\label{fig:interleaving}
\end{center}
\end{figure}

Having defined the above concepts, we now prove the following two lemmas that
are used in the subsequent decomposition.

\begin{lemma}\label{lem:interleaving_number}
Let $G$ be a graph of Euler genus $g$ and let $X\subseteq V(G)$ be such that $H=G\setminus X$ is planar.
Then, for any triple $x_1,x_2,x_3\in X$ of distinct vertices and for any path $P$ in $H$ the $\{x_1,x_2,x_3\}$-interleaving number of $P$ is at most $O(g)$.
\end{lemma}

\begin{proof}
Suppose that the $\{x_1,x_2,x_3\}$-interleaving number of $P$ is $k$.
By the definition of interleaving number there exists a collection of edge-disjoint subpaths $P_1,\ldots,P_k$ of $P$ and a decomposition $E'=E_1\cup\ldots\cup E_k$ such that for each $i\in \{1,\ldots,k\}$ $E_i$ is $P_i$-coupled, and
moreover, for any $k'<k$ there exists no such collections.
This means that for any $i\in \{1,\ldots,k-1\}$
$E_i \cup E_{i+1}$ contains edges incident to all three vertices in $\{x_1,x_2,x_3\}$.

By the choose of $k$, there can be at most two paths in $P_1,\ldots,P_k$ that share the same vertex $v\in V(P)$.
It follows that for any $i\in \{1,\ldots,\lfloor 3/4 \rfloor\}$, if we let $Q_i=P_{3i+1}\circ P_{3i+2}$, then $Q_i$ is incident to all vertices in $\{x_1,x_2,x_3\}$, and for any $i\neq j\in \{1,\ldots,\lfloor k/3 \rfloor\}$ $V(Q_i)\cap V(Q_j) = \emptyset$.
We can now construct a $K_{3,\lfloor k/3 \rfloor}$ minor in $G$ with the left side $\{x_1,x_2,x_3\}$, and for any $i\in \{1,\ldots,\lfloor k/3 \rfloor\}$ the right side contains a vertex obtained by contracting $Q_i$.
It follows by Lemma \ref{lem:K33} that $k=O(g)$, as required.
\end{proof}

We also need the following decomposition.

\begin{lemma}\label{lem:interleaving_number_X}
Let $G$ be a graph of Euler genus $g$ and let $X\subseteq V(G)$ be such that $H=G\setminus X$ is planar.
Let $P$ be a path in $H$.
Let $E'=E(V(P), X)$.
There exists a multi-set\footnote{We consider multi-sets of edge-disjoint paths so that we may allow single-vertex paths to be repeated.} of edge-disjoint subpaths $P_1,\ldots,P_k\subseteq P$, and a decomposition $E'=E'_1\cup \ldots E'_k$, for some $k=O(g\cdot |X|^3)$, such that for any $i\in \{1,\ldots,k\}$ the set $E'_i$ is $P_i$-coupled.
Moreover, there exists a polynomial-time algorithm which given $G$, $g$, $X$, and $P$,  outputs $P_1,\ldots,P_k$ and $E_1',\ldots,E_k'$.
\end{lemma}

\begin{proof}
For each triple $\tau = \{x_1,x_2,x_3\}\in {X \choose 3}$ let ${\cal P}_\tau = \{P_{\tau,i}\}_i^{k_{\tau}}$ be the collection of edge-disjoint paths of $P$ and $E'=E_{\tau,1}\cup\ldots\cup E_{\tau,k_{\tau}}$ the decomposition given by Lemma \ref{lem:interleaving_number}, for some $k_{\tau}=O(g)$.

We construct a collection ${\cal P}=P_1,\ldots,P_k$ of subpaths  of $P$ and a decomposition $E=E_1\cup \ldots\cup E_k$ as follows.
Consider some $\{x_1,x_2\}\in {X\choose 2}$.
For any $x_3\in X\setminus \{x_1,x_2\}$ the paths in ${\cal P}_{\{x_1,x_2,x_3\}}$ have at most $O(g)$ endpoints by Lemma \ref{lem:interleaving_number}.
Therefore, for all possible choices for $x_3$, there are at most $O(g\cdot |X|)$ such endpoints.
It follows that there exists a collection ${\cal P}_{\{x_1,x_2\}}$ of edge-disjoint subpaths of $P$ of $O(g\cdot |X|)$,
such that each edge $e\in E(\{x_1,x_2\}, V(P))$ has an endpoint in some $Q\in {\cal P}_{\{x_1,x_2\}}$, and there is no edge with one endpoint in $X\setminus \{x_1,x_2\}$ and another endpoint in an interior vertex of some $Q\in {\cal P}_{\{x_1,x_2\}}$.
We add all the paths in ${\cal P}_{\{x_1,x_2\}}$ to ${\cal P}$ and the corresponding subsets of $E(\{x_1,x_2\}, V(P))$ to the decomposition of $E'$.
In total, we add at most $O(g\cdot |X|)$ paths to ${\cal P}$ that correspond to the pair $\{x_1,x_2\}$.
Repeating over all possible choices for $\{x_1,x_2\}\in {X\choose 2}$, we obtain the desired decomposition with $k=O(g\cdot |X|^3)$.
\end{proof}

\subsection{Kissing decomposition}

We now give our first decomposition which is needed in a later proof.
Let us begin with some definition.
Intuitively, we would like to argue that in a $k$-apex graph of small genus, the edges between the apices and the planar piece can be partitioned into coupled sets that ``interact'' in a simple fashion; that is, a pair of coupled sets cannot intersect in an arbitrary way.
The following definition formalizes this intuition.

\begin{definition}[Kissing edge-sets]
Let $G$ be a graph and let $X\subseteq V(G)$ be such that $H=G\setminus X$ is planar and 2-connected.
Fix a planar drawing $\phi$ of $H$.
Let $E_1,E_2\subseteq E(X,V(H))$ with $E_1\cap E_2 = \emptyset$.
Let $F_1,F_2$ be distinct faces in $\phi$.
For any $i\in \{1,2\}$ let $P_i$ be a subpath of $F_i$ such that  $E_i$ is $P_i$-coupled.
Then we say that $(E_1,E_2)$ is \emph{$(P_1,P_2)$-kissing} (w.r.t.~$\phi$) if at least one of the following conditions is satisfied:
\begin{description}
\item{(1)}
$V(P_1) \cap V(P_2) = \emptyset$.

\item{(2)}
$V(P_1) \cap V(P_2) = \{v\}$ where $v$ is an endpoint of both $P_1$ and $P_2$.

\item{(3)}
The paths $P_1$ and $P_2$ share both of their endpoints.
Moreover, the following holds.
The closed curve $\phi(P_1)\cup \phi(P_2)$ bounds a collection of zero or more disks ${\cal D}_1,\ldots,{\cal D}_k\subset \mathbb{R}^2$.
Let $H_1,\ldots,H_k$ be the maximal 2-connected components of $H$ with $\phi(H_i) \subset {\cal D}_i$.
Then, we have
\[
E\left(X, V(P_1) \cup V(P_2) \cup \left(\bigcup_{i=1}^k V(H_i)\right)\right) = E_1 \cup E_2.
\]
In other words, there is no edge in $E(X,V(H))$ with an endpoint in some interior vertex of some $H_i$.
See Figure \ref{fig:kissing} for an example.
\ken{Picture needed}
\end{description}
\end{definition}

\begin{figure}
\begin{center}
\scalebox{0.65}{\includegraphics{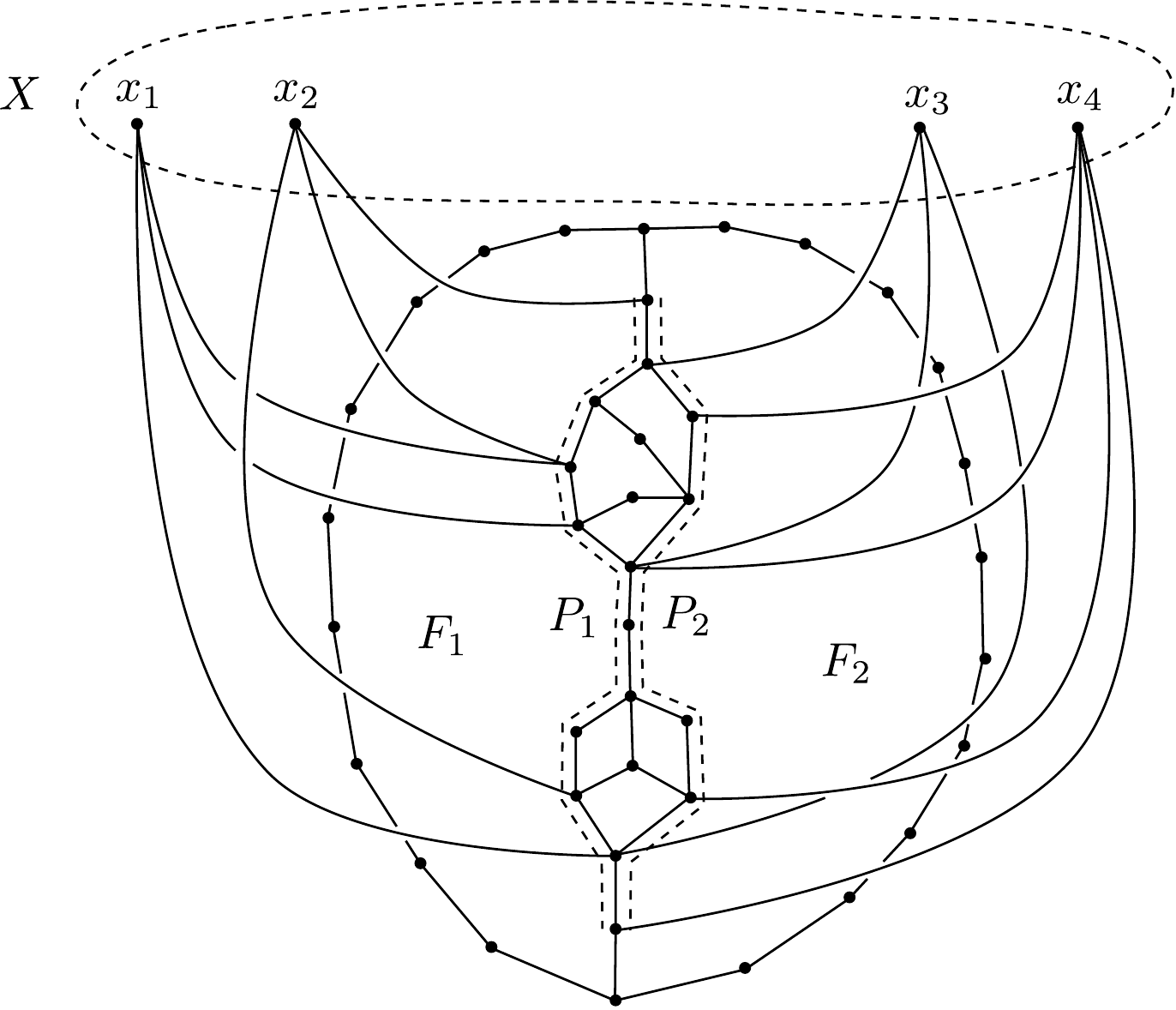}}
\caption{Example of a pair of $(P_1,P_2)$-kissing sets of edges $E_1$ and $E_2$.
$E_1$ is $(P_1, \{x_1,x_2\})$-coupled and  $E_2$ is $(P_2, \{x_3,x_4\})$-coupled. The paths $P_1$ and $P_2$ are marked by dashed curves.}
\label{fig:kissing}
\end{center}
\end{figure}

Using the above definition, the main decomposition result of this subsection can now be stated as follows.

\begin{lemma}[Kissing decomposition]\label{lem:kissing_decomposition}
Let $G$ be a graph of Euler genus $g$ and let $X\subseteq V(G)$ be such that $H=G\setminus X$ is planar and 2-connected.
Let $E'=E(X,V(H))$.
There exists a planar drawing $\phi$ of $H$ and a collection ${\cal F}=\{F_i\}_{i=1}^k$ of faces in $\phi$, for some $k=O(g^2 + |X|^2)$, such that the following conditions are satisfied:
\begin{description}
\item{(1)}
For each $i\in \{1,\ldots,k\}$ there exists a multi-set of pairwise edge-disjoint subpaths $P_{i,1},\ldots,P_{i,k_i}$ of $F_i$, for some $k_i=O(g^6\cdot |X|^6 + g^2 \cdot |X|^{10})$, and a decomposition
\[
E'=\bigcup_{i=1}^k \bigcup_{j=1}^{k_i} E_{i,j},
\]
such that for each $i\in \{1,\ldots,k\}$ and for each $j\in \{1,\ldots,k_i\}$ $E_{i,j}$ is $P_{i,j}$-coupled.

\item{(2)}
For any $i<i'\in \{1,\ldots,k\}$, $j\in \{1,\ldots,k_i\}$, $j'\in \{1,\ldots,k_{i'}\}$, and for any $v\in V(P_{i,j}) \cap V(P_{i',j'})$, $E(v,X) \subseteq E_{i,j}$.
In other words, the edges in $E_G(X,V(H))$ are assigned to the faces in ${\cal F}$ in a greedy fashion, giving priority to the faces $F_i\in {\cal F}$ with the smallest index $i$.

\item{(3)}
For any $i,i'\in \{1,\ldots,k\}$ and for any $j\in \{1,\ldots,k_i\}$, $j'\in \{1,\ldots,k_{i'}\}$, with $(i,j)\neq (i',j')$, $(E_{i,j}, E_{i',j'})$ is $(P_{i,j},P_{i',j'})$-kissing.
\end{description}
Moreover, there exists a polynomial-time algorithm which given $G$, $g$, $X$, and ${\cal F}$, outputs $\{P_{i,j}\}_{i,j}$ and $\{E_{i,j}\}_{i,j}$.
\end{lemma}

\begin{proof}
Let $G'$ be the graph obtained by identifying $X$ into a single vertex $x^*$.
By Lemma \ref{lem:contract_U} $\eg(G')\leq g+|X|-1$.
By Lemma \ref{lem:face_cover} there exists a planar drawing $\phi$ of $H$ and a $\phi$-face cover ${\cal F}$ of $N_{G'}(x^*)$ with $|{\cal F}| = O(g^2+|X|^2)$.
Note that ${\cal F}$ is also a $\phi$-face cover of $N_G(X)$.

Let ${\cal F}=\{F_1,\ldots,F_k\}$.
Let $E' = E(X, V(H))$.
We define a partition $E'=E_1 \cup \ldots \cup E_k$ as follows.
Let $E_1 = E_G(X, V(F_1))$, and for any $i\in \{2,\ldots,k\}$ let
\[
E_i = E(X, V(F_i)) \setminus \left( \bigcup_{j=1}^{i-1} E_j \right).
\]
We have at most $k=O(g^2+|X|^2)$ different subsets $E_i$.

Let $i\in \{1,\ldots,k\}$.
Since $H$ is 2-connected, it follows that $F_i$ is a cycle.
Since we may assume that $G$ does not contain parallel edges, it follows that $|V(F_i)|\geq 3$.
We may therefore decompose $F_i$ into three edge-disjoint subpaths $Q_{i,1}, Q_{i,2}, Q_{i,3}$, and we can also compute a decomposition $E_i = R_{i,1}\cup R_{i,2}\cup R_{i,3}$ such for each $i,j$ we have $R_{i,j}\subseteq E_i \cap E(X,Q_{i,j})$.
For each $i,j$, by Lemma \ref{lem:interleaving_number_X} we can compute a a collection $Q_{i,j,1},\ldots,Q_{i,j,k_{i,j}}$ of edge-disjoint subpaths of $Q_{i,j}$ and a decomposition $E_{i,j}=\bigcup_{r=1}^{k_{i,j}} E_{i,j,r}$, for some $k_{i,j} = O(g\cdot |X|^3)$ such that each $E_{i,j,r}$ is $Q_{i,j,r}$-coupled.
For each $F\in {\cal F}$ we obtain a total number of at most $O(g\cdot |X|^3)$ subsets of $E'$.

It remains to enforce the kissing condition (3).
There are at most $|{\cal F}| \cdot g\cdot|X|^3=O((g^2+|X|^2)\cdot g\cdot |X|^3) = O(g^3\cdot |X|^3 + g\cdot |X|^5)$ total subsets $E_{i,j,r}$ and corresponding paths $Q_{i,j,r}$.
Let $U\subseteq V(H)$ be the set of all the endpoints of the paths $Q_{i,j,r}$.
So $|U| = O(g^3\cdot |X|^3 + g\cdot |X|^5)$.
For every path $Q_{i,j,r}$, we partition $Q_{i,j,r}$ into subpaths $Q_{i,j,r,1},\ldots,Q_{i,j,r,t}$, for some $t\leq |U|+1$ by cutting $Q_{i,j,r,1}$ along all vertices in $U$.
We also partition each $E_{i,j,r}$ into subsets $E_{i,j,r,1},\ldots,E_{i,j,r,t}$ such that each $E_{i,j,r,t}$ is $Q_{i,j,r,t}$-coupled.
Let ${\cal E}$ be the resulting collection of subsets $E_{i,j,r,t}$ of $E'$ and let ${\cal Q}$ be the resulting collection of paths $Q_{i,j,r,t}$.
Then $|{\cal E}| = |{\cal Q}| = O(g^6 \cdot |X|^6 + g^2 \cdot |X|^{10})$.
In particular, for every $F\in {\cal F}$ at most $O(g\cdot |X|^3 \cdot (g^3\cdot |X|^3 + g\cdot |X|^5)) = O(g^4\cdot |X|^6 + g^2 \cdot |X|^8)$ subsets of $E'$ in ${\cal E}$.

Consider now some $R\in {\cal E}$ and the corresponding path $Q\in {\cal Q}$ such that $R$ is $Q$-coupled.
The path $Q$ was obtained as a subpath of some face $F\in {\cal F}$ in $\phi$.
Let $F'\in {\cal F}$ with $F'\neq F$.
Since $H$ is 2-connected, there exists a minimal disk ${\cal D}\subset \mathbb{R}^2$ containing $\phi(F) \cup \phi(F')$.
Let $L$ be the maximal subpath of $F$ such that the interior of the curve $\phi(L)$ is completely contained in the interior of ${\cal D}$.
We cut $Q$ along all endpoints of $L$ that are in the interior of $Q$.
The result consists of at most three edge-disjoint subpaths $Q_1,Q_2,Q_3$ of $Q$.
We accordingly partition $R$ into at most three subsets $R_1,R_2,R_3$ such that each $R_i$ is $Q_i$-coupled.
Repeating the above process for all $F'\neq F \in {\cal F}$ we partition $Q$ into at most $3\cdot (|{\cal F}|-1)=O(g^2+|X|^2)$ subpaths, and $R$ into at most $3\cdot (|{\cal F}|-1)=O(g^2+|X|^2)$ subsets.
Repeating the same process for all $R\in {\cal E}$ and all corresponding $Q\in {\cal Q}$, we obtain a partition ${\cal E'}$ of $E'$ with $|{\cal E}'| = O((g^2+|X|^2)\cdot (g^6 \cdot |X|^6 + g^2 \cdot |X|^{10})) = O(g^8\cdot |X|^6 + g^2\cdot |X|^{12})$, and a corresponding collection ${\cal Q}$ of paths such that each $R\in {\cal E}'$ is $Q$-coupled for some $Q$ in ${\cal Q}'$.
Note that for each $F\in {\cal F}$ we have at most $O((g^2+|X|^2)\cdot (g^4\cdot |X|^6 + g^2 \cdot |X|^8)) = O(g^6\cdot |X|^6 + g^2 \cdot |X|^{10})$ subsets of $E'$ in ${\cal E}'$.
It is immediate from the definition that for any $R,R'\in {\cal E}'$, $Q,Q'\in {\cal Q}'$ such that $R$ is $Q$-coupled and $R'$ is $Q'$-coupled we have that $(R,R')$ is $(Q,Q')$-kissing,
concluding the proof.
\end{proof}

\subsection{Centipedes and butterflies}

Having obtained the ``kissing decomposition'' of the previous subsection, we now proceed to define another decomposition into structures that we call ``centipedes'' and ``butterflies''.
Intuitively, our goal is to enforce additional structures on pairs of coupled edges sets, that will subsequently allow us to compute the desired embedding into a surface of small Euler genus.
Let us now formally define these structures.

\begin{definition}[Centipede]
Let $G$ be a graph and let $X\subseteq V(G)$ such that $H=G\setminus X$ is planar.
Let $\phi$ be a planar drawing of $H$.
Let $x_1,x_2\in X$.
Let $F$ be a face in $\phi$ and let $C$ be a subpath of $F$.
Let $R\subseteq E(X, V(C))$ such that $R$ is $(C, \{x_1,x_2\})$-coupled.
We say that $(C,R)$ is a \emph{$(\phi, \{x_1,x_2\})$-centipede} (w.r.t.~$X$), or just $\phi$-centipede when $x_1$ and $x_2$ are clear from the context.
\end{definition}

\begin{definition}[Butterfly]
Let $G$ be a graph and let $X\subseteq V(G)$ such that $H=G\setminus X$ is planar and connected.
Let $s,t\in V(H)$, $x_1,x_2\in X$.
Let $H'$ be the graph obtained by cutting $H$ along $s$ and $t$.
Let $C$ be a component of $H'$ and let $R\subseteq E_G(\{x_1,x_2\}, V(C))$.
We say that $(C,R)$ is an \emph{$\{x_1,x_2\}$-butterfly} (w.r.t.~$X$), or just butterfly when $x_1$ and $x_2$ are clear from the context, if the following conditions are satisfied.
\begin{description}
\item{(1)}
For any $v\in V(C) \setminus \{s,t\}$, we have $E(v,X) = E(v,\{x_1,x_2\}) \subseteq R$.
\item{(2)}
There exists a planar drawing $\phi$ of $C$ such that $s$, $t$, and all the endpoints of edges in $R$ lie on the boundary face of $\phi$.
\end{description}
We refer to $s$ and $t$ as the \emph{endpoints} of $C$ (see Figure \ref{fig:butterfly} for an example).
\end{definition}

\begin{figure}
\begin{center}
\scalebox{0.65}{\includegraphics{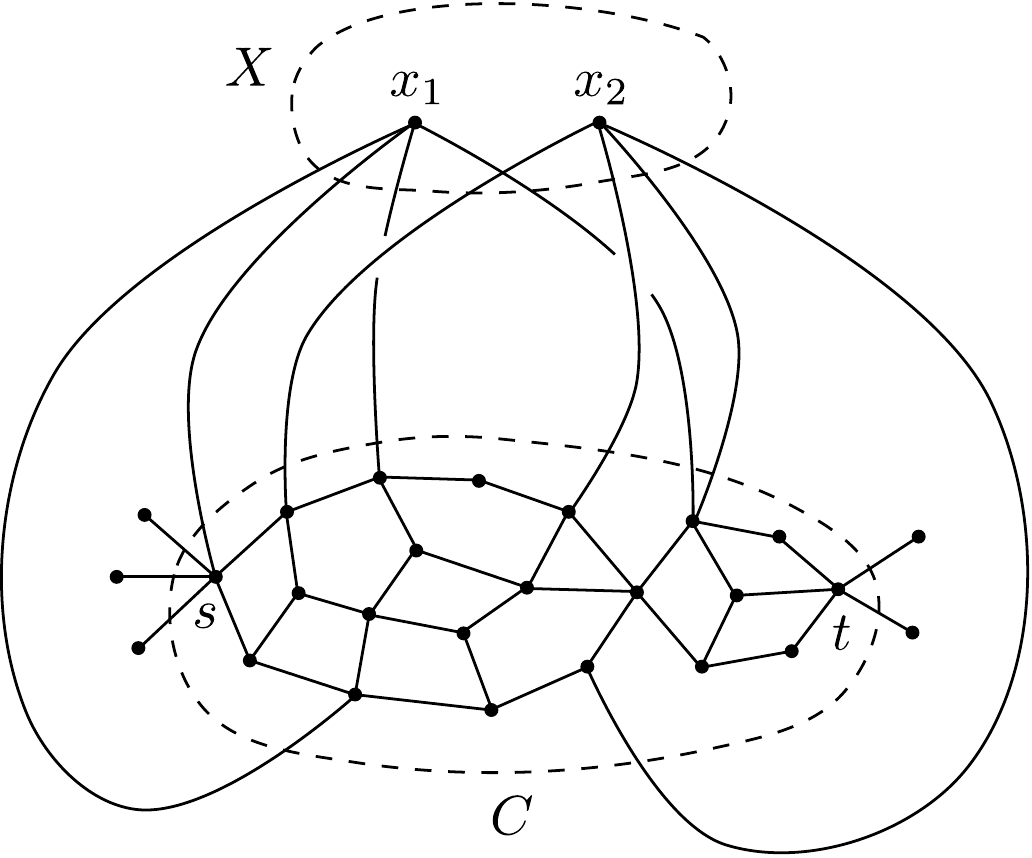}}
\caption{Example of $\{x_1,x_2\}$-butterfly $(C,R)$ with endpoints $s$ and $t$.}
\label{fig:butterfly}
\end{center}
\end{figure}

The following is the main lemma of this subsection. It starts with the kissing decomposition from the previous subsection, and modifies the decomposition into centipedes and butterflies.

\begin{lemma}[Centipede-butterfly decomposition]\label{lem:centipede-butterly-decomposition}
Let $G$ be a graph of Euler genus $g$ and let $X\subseteq V(G)$ such that $H=G\setminus X$ is planar and 2-connected.
Let $E'=E_G(X,V(H))$.
Then there exists a planar drawing $\phi$ of $H$ and a collection ${\cal A}=\{(C_i,R_i)\}_{i=1}^a$, for some $a=O(g^9 \cdot |X|^6 + g^3\cdot |X|^{12})$, such that the following conditions are satisfied:
\begin{description}
\item{(1)}
For any $i\in \{1,\ldots,a\}$, $(C_i, R_i)$ is either a $\phi$-centipede or a butterfly (w.r.t.~$X$).

\item{(2)}
For any $i\neq j\in \{1,\ldots,a\}$, $C_i$ and $C_j$ can intersect only on their endpoints.

\item{(3)}
$E' = \bigcup_{i=1}^a R_i$ and for any $i\neq j\in \{1,\ldots,t\}$, $R_i\cap R_j = \emptyset$.
\end{description}
Moreover, there exists a polynomial-time algorithm which given $G$, $g$, and $X$,  outputs $\phi$ and ${\cal A}$.
\end{lemma}

\begin{proof}
Let $\phi$ be the planar drawing of $H$, and ${\cal F}=\{F_i\}_{i=1}^k$ the collection of faces in $\phi$, for some $k=O(g + |X|)$, satisfying the conditions of Lemma \ref{lem:kissing_decomposition}.
Let ${\cal Q}=\{P_{i,j}\}_{i,j}$ be the multi-set of paths in $H$, let ${\cal E}=\{E_{i,j}\}_{i,j}$ be the decomposition of $E'$, and let $\phi$ be the planar drawing of $H$ given by Lemma \ref{lem:kissing_decomposition}.

We construct ${\cal A}$ inductively.
Initially we set ${\cal A}=\emptyset$.
We consider all paths in ${\cal Q}$ as being  unmarked and proceed to examine all unmarked paths in an arbitrary order until all paths have been marked.
Let $P_1 \in {\cal Q}$ be an unmarked path.
Let $E_1 \in {\cal E}$ be the corresponding  $P_1$-coupled set.
If $|V(P_1)|=1$, then clearly $(P_1,E_1)$ is a butterfly.
We add to $(P_1,E_1)$ to ${\cal A}$ and we mark $P_1$.
Otherwise, suppose that $|V(P_1)|>1$, and therefore $P_1$ contains at least one edge.
Suppose first that $P_1$ does not share any edges with the remaining unmarked paths.
In that case, $(P_1,E_1)$ is a centipede.
We add to $(P_1,E_1)$ to ${\cal A}$ and we mark $P_1$.

Finally, suppose that $P_1$ shares some edge with at least one other unmarked path.
By Lemma \ref{lem:kissing_decomposition} it follows that there exists at most one unmarked path $P_2\neq P_1\in {\cal Q}_{i}$ that shares at least one edge with $P_1$.
Let $E_2 \in {\cal E}$ be the corresponding  $P_2$-coupled set.
Then $(E_1,E_2)$ is $(P_1,P_2)$-kissing, and the paths $P_1$ and $P_2$ share their endpoints.
Moreover, there exist faces $F_1,F_2$ of $\phi$ such that for each $i\in \{1,2\}$, $P_i$ is a subpath of $F_i$.
Also, there exist (not necessarily distinct) $x_1,x_2,x_3,x_4\in X$ such that $E_1$ is $(P_1,\{x_1,x_2\})$-coupled and $E_2$ is $(P_2,\{x_3,x_4\})$-coupled.
By condition (2) of Lemma \ref{lem:kissing_decomposition} we may assume w.l.o.g.~that $E(X, V(P_1)\cap V(P_2)) \subseteq E_1\cup E_2\subseteq E(X, V(P_1))$.
Let $s,t$ be the common endpoints of $P_1$ and $P_2$.
Let $H'$ be the graph obtained by cutting $H$ along $s$ and $t$.
Let $C$ be the component of $H'$ that contains $P_1\cup P_2$.
By the definition of kissing paths
the closed curve $\phi(P_1)\cup \phi(P_2)$ bounds a collection of zero or more disks ${\cal D}_1,\ldots,{\cal D}_{\ell}\subset \mathbb{R}^2$.
Let $H_1,\ldots,H_{\ell}$ be the maximal 2-connected components of $H$ with $\phi(H_i) \subset {\cal D}_i$.
Then
\[
E\left(X, V(P_1) \cup V(P_2) \cup \left(\bigcup_{i=1}^k V(H_i)\right)\right) = E_1 \cup E_2.
\]
In other words, there is no edge in $E(X,V(H))$ with an endpoint in some interior vertex of some $H_i$.
Let $C'$ be the graph obtained from $C$ by contracting every $H_i$ into a single vertex $h_i$.
Let also $P'$ and $G'$ be the corresponding minors of $H$ and $G$.
Note that $P'$ is a path.
By Lemma \ref{lem:interleaving_number_X}
 there exists a multi-set of edge-disjoint subpaths $P_1',\ldots,P_{k'}'\subseteq P'$, and a decomposition $E''=E''_1\cup \ldots E''_{k'}$, where $E''=E_{G'}(X, V(P'))$, for some $k'=O(g\cdot |\{x_1,x_2,x_3,x_4\}|^3)=O(g)$, such that for any $i\in \{1,\ldots,k'\}$ the set $E''_i$ is $P'_i$-coupled.
We consider all the paths $P'_{\iota}$ in the above collection in an arbitrary order.
We distinguish between the following two cases for each $P'_{\iota}$:

\begin{description}
\item{Case 1:}
Suppose that $P'_{\iota}$ does not contain any of the vertices $h_1,\ldots,h_{\ell}$.
Then $P'_{\iota}$ is a subpath of $H$.
It follows that $(P'_{\iota}, E''_{\iota})$ is a $\{x_1,x_2\}$-butterfly (in $G$, w.r.t.~$X$).
We add $(P'_{\iota}, E''_{\iota})$ to ${\cal A}$.

\item{Case 2:}
Suppose that $P'_{\iota}$ contains some vertex $h_{\zeta}$.
Since $P'_{\iota}$ is a subpath of $P'$ it follows that all the vertices in $\{h_1,\ldots,h_{\ell}\}$ that are contained in  $P'_{\iota}$ must span a continuous interval of indices, say, $h_{\sigma},\ldots,h_{\tau}$.
For each $r\in \{\sigma,\ldots,\tau\}$, the boundary of $H_r$ in $\phi$ consists of a closed curve $\phi(P_{\iota,r,1}) \cup \phi(P_{\iota,r,2})$, where $P_{\iota,r,1}$ is a subpath of $P_1$ and $P_{\iota,r,2}$ is a subpath of $P_2$.
Moreover, $E_{G'}(\{x_3,x_4\}, V(P')) \subseteq E_{G'}(\{x_3,x_4\}, \{h_{\tau},\ldots,h_{\sigma}\})$.
Thus $E_(\{x_3,x_4\}, V(P')) \subseteq E(\{x_3,x_4\}, \{V(P_{\iota,\sigma,2}),\ldots,V(P_{\iota,\tau,2})\})$.
We have the following sub-cases:
\begin{description}
\item{Case 2.1:}
If $\{s,t\} \cap \{h_{\sigma},h_{\tau}\}=\emptyset$ then let $R$ be the set of edges in $E(G)$ corresponding to the set $E''_{\iota}$.
Let also $J$ be the subgraph of $C$ corresponding to $P'_{\iota}$.
Then it follows that $(J, R)$ is a $\{x_1,x_2\}$-butterfly.
We add $(P'_{\iota}, E''_{\iota})$ to ${\cal A}$.

\item{Case 2.2:}
Suppose that $|\{s,t\} \cap \{h_{\sigma},h_{\tau}\}|=1$, and assume w.l.o.g.~that $s=h_{\sigma}$.
Let $W_1 = E_1 \cap E(X,V(P_{\iota,\sigma,1}))$.
Then $(P_{\iota,\sigma,1}, W_1)$ is a $(\phi, \{x_1,x_2\})$-centipede.
Note that $W_1$ might share edges with some $E''_{\iota'}$ for some $\iota'\neq \iota$.
In that case, $(P_{\iota,\sigma,1}, W_1)$ might have already been added to ${\cal A}$ while considering $P'_{\iota'}$.
If this is not the case, then we add $(P_{\iota,\sigma,1}, W_1)$ to ${\cal A}$.
Similarly, let $W_2 = E_2 \cap E(X,V(P_{\iota,\sigma,2}))$.
Then $(P_{\iota,\sigma,2}, W_2)$ is a $(\phi, \{x_3,x_4\})$-centipede.
If $(P_{\iota,\sigma,2}, W_2)$ is not already in ${\cal A}$ then we add it to ${\cal A}$.
Finally, let $R$ be the set of edges in $E(G)$ corresponding to the set $E''_{\iota}$, and let $J$ be the subgraph of $C$ corresponding to $P'_{\iota}$.
Then $(J\setminus V(H_{\sigma}), R\setminus (W_1\cup W_2))$ is a $\{x_1,x_2\}$-butterfly.
We add $(J\setminus V(H_{\sigma}), R\setminus (W_1\cup W_2))$ to ${\cal A}$.

\item{Case 2.3:}
Suppose that $|\{s,t\} \cap \{h_{\sigma},h_{\tau}\}|=2$, and assume w.l.o.g.~that $s=h_{\sigma}$ and $t=h_{\tau}$.
As in case 2.2, we let
$W_1 = E_1 \cap E(X,V(P_{\iota,\sigma,1}))$ and $W_2 = E_2 \cap E(X,V(P_{\iota,\sigma,2}))$.
Then $(P_{\iota,\sigma,1}, W_1)$ is a $(\phi, \{x_1,x_2\})$-centipede and $(P_{\iota,\sigma,2}, W_2)$ is a $(\phi, \{x_3,x_4\})$-centipede.
Similarly,
let $W_1' = E_1 \cap E(X,V(P_{\iota,\tau,1}))$ and $W_2' = E_2 \cap E(X,V(P_{\iota,\tau,2}))$.
Then $(P_{\iota,\tau,1}, W_1')$ is a $(\phi, \{x_1,x_2\})$-centipede and $(P_{\iota,\tau,2}, W_2')$ is a $(\phi, \{x_3,x_4\})$-centipede.
We add to ${\cal A}$ any of the above centipedes that might have not been already added to ${\cal A}$.
Finally, let $R$ be the set of edges in $E(G)$ corresponding to the set $E''_{\iota}$, and let $J$ be the subgraph of $C$ corresponding to $P'_{\iota}$.
Then $(J\setminus (V(H_{\sigma}\cup V(H_{\tau}))), R\setminus (W_1\cup W_2\cup W_1'\cup W_2'))$ is a $\{x_1,x_2\}$-butterfly; we add it to ${\cal A}$.
\end{description}
\end{description}
This completes the description of the collection ${\cal A}$.
It is immediate by the construction that conditions (1), (2), and (3) are satisfied.
Moreover, for every path in ${\cal Q}$ we add at most $O(g)$ elements in ${\cal A}$.
Thus $|{\cal A}| = O(g \cdot |{\cal Q}|) = O(g^9 \cdot |X|^6 + g^3\cdot |X|^{12})$, concluding the proof.
\end{proof}

\subsection{Algorithms for the 2-connected case}

We now use the decomposition into centipedes and butterflies (Lemma~\ref{lem:centipede-butterly-decomposition}) to
obtain an embedding for $k$-apex graphs with a 2-connected planar piece.
We first need to show how to embed centipedes.
This is done in the following Lemma.

\begin{lemma}[Embedding centipedes]
\label{lem:embedding_centipedes}
Let $G$ be a graph of Euler genus $g$ and let $X\subseteq V(G)$ such that $H=G\setminus X$ is  planar and $2$-connected.
Let $E'=E(X, V(H))$.
Let $\phi$ be a planar drawing of $H$.
Let ${\cal C}=\{C_i, R_i\}_{i=1}^k$ be a collection of $\phi$-centipedes such that
for any $i\neq j\in \{1,\ldots,k\}$ the paths $C_i$ and $C_j$ can intersect only on their endpoints.
Suppose further that $E'=\bigcup_{i=1}^k R_i$ and for any $i\neq j\in \{1,\ldots,k\}$, $R_i\cap R_j = \emptyset$.
Then there exists a polynomial time algorithm which given $G, g, X, \phi$, and ${\cal C}$,  outputs a drawing of $G$ into a surface of genus $O(k g^2)$.
\end{lemma}

\begin{proof}
For any $i\in \{1,\ldots,k\}$, $(C_i,R_i)$ is a $\phi$-centipede.
Let $s_i$, $t_i$ be the endpoints of the path $C_i$.
Therefore, there exist distinct $x_i,y_i\in X$ such that $R_i$ is $(C_i,\{x_i,y_i\})$-coupled.
Let ${\cal C}'$ be the subset of ${\cal C}$ containing all $(C_i,R_i)$ such that
all edges in $R_i$ are incident to at most one vertex in $X$, that is
either $|E(x_i,V(C_i))\cap R_i)|=0$ or $|E(y_i,V(C_i))\cap R_i)|=0$.
Let also ${\cal C}''={\cal C}\setminus {\cal C}'$.

We may assume w.l.o.g.~that every centipede $(C_i,R_i)\in {\cal C}''$ satisfies $|E(x_i,V(C_i))\cap R_i)|\geq 2$ and $|E(y_i,V(C_i))\cap R_i)|\geq 2$.
This is because if, say, $|E(x_i,V(C_i))\cap R_i)|=1$ then we may remove one edge from $G$ and move $(C_i,R_i)$ to ${\cal C}'$.
In total we remove at most $k$ edges.
We can then compute a drawing for the resulting graph and extend it to the removed edges by adding at most $k$ additional handles.

We may assume w.l.o.g.~that for every $(C_i,R_i)\in {\cal C}''$, $\{s_i,x_i\}\in R_i$ and $\{t_i,x_i\}\in R_i$.
Otherwise, let $C_{i,1}$ be the minimal subpath of $C_i$ containing all the endpoints of edges in $R_i$ that are incident to $x_i$.
Similarly, let $C_{i,2}$ be the minimal subpath of $C_i$ containing all the endpoints of edges in $R_i$ that are incident to $y_i$.
Then, we may remove $(C_i, R_i)$ from ${\cal C}$ and add the centipede $(C_{i,1}, R_i \cap E(\{x_i,y_i\}, V(C_{i,1}))$ in ${\cal C}$.
If $C_i\setminus C_{i,1}$ consists of a single path, then we also add the centipede $(C_i\setminus C_{i,1}, R_i \cap E(\{x_i,y_i\}, V(C_i\setminus C_{i,1}))$ in ${\cal C}$.
Otherwise, if $C_i\setminus C_{i,1}$ consists of two paths $W$ and $W'$, then we add the centipedes $(W, E_i\cap E(\{x_i,y_i\}, V(W))$ and $(W', E_i\cap E(\{x_i,y_i\}, V(W'))$ to ${\cal C}$.
This only increases the total number of centipedes by a constant factor, so it does not affect the assertion.

Similarly, we may further assume w.l.o.g.~that for any $(C_i,R_i)\in {\cal C}'$, $\{s_i,x_i\}\in R_i$ and $\{t_i,x_i\}\in R_i$, by replacing $C_i$ with the minimal subpath of $C_i$ containing all the endpoints of the edges in $R_i$.

We next define a sequence of 2-connected planar graphs $\{H_i\}_{i=0}^k$ and corresponding planar drawings $\{\phi_i\}_{i=1}^k$, with $H_0=H$ and $\phi_0=\phi$.
We will maintain the following inductive invariants:
\begin{align*}
\text{(I1)}~~& H_i\text{ is 2-connected}\\
\text{(I2)}~~& \text{for any } j\in \{1,\ldots,i\}\text{, there exists a } \phi_i\text{-face cover } E_j\cap N_{H_i}(y_j) \text{ of size }O(g^2),\\
\text{(I3)}~~& \text{for any } j\in \{i+1,\ldots,k\}, (C_j,R_j)\text{ is a } \phi_i\text{-centipede},
\end{align*}
which are both true for $i=0$.

Given $H_i$ and $\phi_i$ for some $i<k$, we proceed to define $H_{i+1}$ and $\phi_{i+1}$.
Let $H_{i+1}$ be the graph obtained by adding to $H_i$ the vertex $x_{i+1}$ and all edges in $E_{i+1}$ that are incident to $x_{i+1}$.
That is, $V(H_{i+1}) = V(H_i) \cup \{x_{i+1}\}$ and $E(H_{i+1}) = E(H_i) \cup (E_{i+1}\cup E(x_{i+1}, V(C_{i+1})))$.
Note that invariant (I1) is maintained since there are at least two edges in $E_i$ that are incident to $x_i$.
By invariant (I3) $(C_{i+1}, R_{i+1})$ is a $\phi_i$-centipede.
It follows that $C_{i+1}$ is a subpath of a face $W_i$ of $\phi_i$, and thus $H_{i+1}$ is planar.
By Lemma \ref{lem:face_cover} there exists a planar drawing $\psi_{i+1}$ of $H_{i+1}$ and a $\psi_{i+1}$-face cover ${\cal F}_{i+1}$ of $N_{H_{i+1}}(y_{i+1})$ of size $O(g^2)$.
Let ${\cal Z}$ be the collection of connected components of $H_{i+1}\setminus W_i$.
We define the embedding $\phi_{i+1}$ of $H_{i+1}$ as follows.
For every component $Z\in {\cal Z}$ we set $\phi_{i+1}|_Z$ to be an embedding combinatorially equivalent to $\phi_{i}|_Z$.
By invariant (I1) $W_i$ is a cycle in $H_{i+1}$, and thus $\gamma=\psi_{i+1}(W_i)$ is a simple loop in the plane.
Therefore, every $Z\in {\cal Z}$ is embedded inside one of the two connected components of $\mathbb{R}^2\setminus \gamma$.
We extend $\phi_{i+1}$ to each component $Z\in {\cal Z}$ by embedding it inside the same component of $\mathbb{R}^2\setminus \gamma$ as in $\psi_{i+1}$.
This can be done by setting $\phi_{i+1}|_{W_i}$ to be combinatorially equivalent to $\psi_{i+1}|_{W_i}$.
Since the paths of any two centipedes can intersect only on their endpoints, it follows that for any $j\in \{1,\ldots,i\}$, the size of the minimum $\phi_{i+1}$-face cover of $N_{H_{i+1}}(y_j)$ is equal to the size of the minimum $\phi_{i}$-face cover of $N_{H_{i}}(y_j)$, and thus invariant (I2) is maintained.
Similarly, since the paths of any two centipedes may intersect only on their endpoints, it follows that invariant (I3) is also maintained.
This completes the construction of the planar graphs $\{H_i\}_{i=0}^k$ and the planar drawings $\{\phi_i\}_{i=0}^k$.

By invariant (I2) for any $j\in \{1,\ldots,k\}$ there exists a $\phi_k$-face cover ${\cal F}_j$ of $E_j\cap N_G(y_j)$ of size $O(g^2)$.
Moreover, since the paths of any two centipedes in ${\cal C}$ can only intersect on their endpoints, and for any $i\in \{1,\ldots,k\}$, $\{x_i,s_i\}\in E_i$ and $\{x_i,t_i\}\in E_i$, it follows that we can choose the collection of face covers $\{{\cal F}_i\}_i$ such that for any $j\neq j'\in \{1,\ldots,k\}$, we have ${\cal F}_j \cap {\cal F}_{j'} = \emptyset$.
Therefore, we can extend $\phi_k$ to $G$ by adding for each $i\in \{1,\ldots,k\}$ at most $|{\cal F}_i| = O(g^2)$ handles.
Therefore, we obtain an embedding of $G$ into a surface of Euler genus $O(k g^2)$, concluding the proof.
\end{proof}

Using the above ingredients, we are now ready to obtain our algorithm for embedding $k$-apex graphs with a 2-connected planar piece.

\begin{lemma}\label{lem:genus_2-connected}
Let $G$ be a planar graph of genus $g$ and let $X\subset V(G)$ such that $H=G\setminus X$ is planar.
Suppose that $H$ is 2-connected.
Then there exists a polynomial-time algorithm that given $G$, $g$, and $X$,  outputs a drawing of $G$ into a surface of Euler genus $O(g^{11}\cdot |X|^6 + g^5 \cdot |X|^{12})$.
\end{lemma}

\begin{proof}
By Lemma \ref{lem:centipede-butterly-decomposition} we can compute a planar drawing $\phi$ of $H$ and a collection ${\cal A} = \{(C_i,R_i)\}_{i=1}^a$, for some
$a=O(g^9 \cdot |X|^6 + g^3\cdot |X|^{12})$, such that:
(1)
for any $i\in \{1,\ldots,a\}$, $(C_i, R_i)$ is either a $\phi$-centipede or a butterfly (w.r.t.~$X$),
(2) for any $i\neq j\in \{1,\ldots,a\}$, $C_i$ and $C_j$ can intersect only on their endpoints,
and
(3) $E' = \bigcup_{i=1}^a R_i$ and for any $i\neq j\in \{1,\ldots,a\}$, $R_i\cap R_j = \emptyset$.

For any $i\in \{1,\ldots,a\}$, let $s_i$, $t_i$ be the endpoints of $C_i$.
Let $H'$ be the graph obtained from $H$ by replacing
$C_i$ by an edge $\{s_i,t_i\}$ and removing all edges in $R_i$, for each $i\in \{1,\ldots,a\}$ such that $(C_i,R_i)$ is a $\phi$-butterfly.
Let also $G'$ be the corresponding graph obtained from $G$ after performing the above operation.
Clearly, $G'$ is a minor of $G$, and therefore $\eg(G') \leq \eg(G)$.
Moreover, since $H$ is 2-connected, the graph $H'$ is also 2-connected.

Let ${\cal C}$ be the set containing all $(C_i,R_i)\in {\cal A}$ such that $(C_i,R_i)$ is a $\phi$-centipede.
By Lemma \ref{lem:embedding_centipedes} we can compute an embedding $\psi$ of $G'$ into a surface of Euler genus at most
$O(g^2\cdot |{\cal C}|) = O(g^2\cdot |{\cal A}|) = O(g^2 a) = O(g^{11} \cdot |X|^6 + g^5\cdot |X|^{12})$.

Consider some $\{x_1,x_2\}$-butterfly $B_i=(C_i,R_i)\in {\cal A} \setminus {\cal C}$,
for some $x_1,x_2\in X$.
Let $J_i$ be the graph with $V(J_i) = V(C_i) \cup \{x_1\}$ and $E(J_i) = E(C_i) \cup (E_i\cap E(x_1,V(C_i)))$.
By the definition of butterfly it follows that $C_i$ is a planar graph.
Moreover, there exists a planar drawing $\zeta_i$ of $C_i$ such that $s_i$, $t_i$, and all the endpoints of edges in $R_i$ in $C_i$ lie on the outer face.
It follows that $J_i$ is planar.
By Lemma \ref{lem:face_cover} we can compute a planar embedding $\zeta_i'$ of $J_i$ and a $\zeta_i'$-face cover ${\cal F}_i$ of $R_i \cap E_G(x_2, V(C_i))$ of size $O(g^2)$.
It follows that we can extend $\zeta_i'$ to $x_2$ by adding at most $O(g^2)$ handles.
This results to an embedding $\xi_i$ of the graph $J_i'$ where $V(J_i')=V(C_i)\cup \{x_1,x_2\}$ and $E(J_i') = E(C_i) \cup E_i$ into a surface of Euler genus $O(g^2)$.

We can now compute an embedding of $G$ as follows.
We begin with the embedding $\psi$ of $G'$.
For every butterfly $B_i=(C_i,R_i)$ there exists an edge $e_i=\{s_i,t_i\}$ in $G'$.
Let $J_i'$ be the graph constructed above, and the corresponding embedding $\xi_i$ of $J_i'$ into a surface of Euler genus $O(g^2)$.
We remove $e_i$ and we identify the copies of $s_i$ and $t_i$ in $J_i'$ with their copies in $G'$, by increasing the Euler genus of the underlying surface by at most $O(g^2)$.
Repeating the same process for all butterflies in ${\cal A}$ we obtain an embedding of $G$ into a surface of Euler genus
$O(g^{11} \cdot |X|^6 + g^5\cdot |X|^{12} + |{\cal A}| \cdot g^2)  = O(g^{11} \cdot |X|^6 + g^5\cdot |X|^{12} + (g^9 \cdot |X|^6 + g^3\cdot |X|^{12}) \cdot g^2) = O(g^{11}\cdot |X|^6 + g^5 \cdot |X|^{12})$.
\end{proof}


\section{Dealing with 1-separators}
\label{sec:1-separators}

In this section we present some of the tools necessary to handle $k$-apex graphs when the planar piece is not $2$-connected.

We begin by giving the proof for Lemma \ref{lem:all3}.

\begin{proof}[Proof of Lemma \ref{lem:all3}]
The graph $G$ contains a $K_{3,r}$ minor, with $r=|U|$.
By Lemma \ref{lem:K33} it follows that $|U|=O(g)$.
\end{proof}

We need the following two definitions.

\begin{definition}[Petal]
Let $H$ be a planar graph.
Let $v\in V(H)$ be a 1-separator of $H$.
Let $C$ be a component of $H\setminus \{v\}$.
The subgraph $H[C\cup \{v\}]$ is called a \emph{$v$-petal}.
\end{definition}

\begin{definition}[Propeller]
Let $H$ be a planar graph and let $v\in V(H)$ be a 1-separator of $H$.
Let $C_1,\ldots,C_t$ be $v$-petals.
Then, the subgraph $C_1\cup \ldots \cup C_t$ of $H$ is called a \emph{$v$-propeller}.
\end{definition}


\begin{figure}
\begin{center}
\scalebox{0.65}{\includegraphics{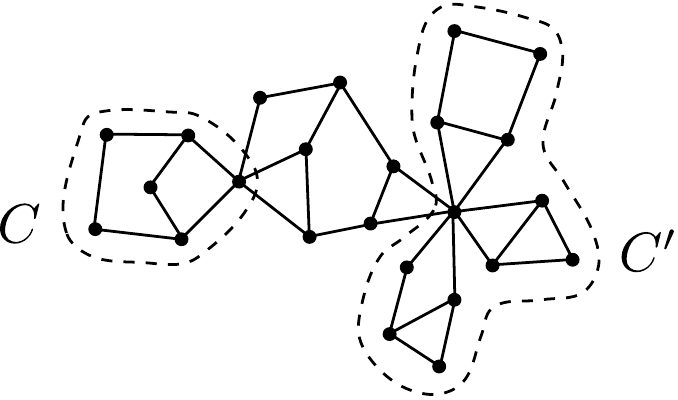}}
\caption{Example of a petal $C$ and a propeller $C'$.}
\label{fig:petal_propeller}
\end{center}
\end{figure}

The following two lemmas are needed for
the main algorithm of this section.

\begin{lemma}\label{lem:x1x2x3}
Let $G$ be a graph of Euler genus $g$ and let $X\subset V(G)$ such that $H=G\setminus X$ is planar.
Let $\phi$ a planar drawing of $H$.
Suppose that no vertex in $V(H)$ is incident to at least three vertices in $X$.
Let $x_1,x_2,x_3\in X$ be distinct vertices.
Let $W$ be the set of 1-separators of $H$.
Let
$W' = W \cap N(x_1) \cap N(x_2)$.
Let ${\cal L}$ be the set of connected components of the graph obtained by cutting $H$ along $W'$.
Then the following two conditions hold.
\begin{description}
\item{(1)}
Let ${\cal L}'$ be the set of all components in ${\cal L}$ that do not contain any of the leaves of ${\cal T}$ and that contain at least one vertex incident to $x_3$.
Then we have $|{\cal L}'| = O(g)$.
\item{(2)}
Let ${\cal L}''$ be the set of components in ${\cal L}\setminus {\cal L}'$ that contain at least one vertex incident to $x_3$.
Then, there exist distinct $w_1,\ldots,w_k\in W'$, with $k=O(g)$, and for each $i\in \{1,\ldots,k\}$ a $w_i$-propeller $H_i$, such that
\[
\bigcup_{C\in {\cal L}''} C = \bigcup_{i=1}^k H_i.
\]
\item{(3)}
There exists a $(H,\phi)$-splitting sequence of length $O(g)$ such that in the resulting graph every cluster in ${\cal L}'\cup \{V(H_1),\ldots,V(H_k)\}$ is contained in a distinct connected component.
\end{description}
\end{lemma}
\begin{proof}
(1)
Let $r=|{\cal L}'|$.
Pick an arbitrary $C^*\in {\cal C}$ and consider ${\cal T}$ being rooted at $C^*$.
Consider some $C\in {\cal L}'$.
Pick some 1-separator $u_C\in C$ that connects $C$ to some descendant component in ${\cal T}$.
Since all components in ${\cal L}'$ correspond to non-leaves of ${\cal T}$ it follows that for any $C\neq C'\in {\cal L}'$ $u_C \neq u_{C'}$.
For any $C\in {\cal L}'$ there exists a vertex $z_C\in V(C)$ with $\{z_C,x_3\}\in E(G)$.
Since $C^*$ is connected, it follows that there exists a path $P_{C^*} \subset C^*$ between $u_C$ and $z_C$.
For any $C\neq C^*$ let $p_C\in V(C)$ be the vertex minimizing $d_G(p_C, C^*)$.
Since $C$ is 2-connected it follows that $C\setminus p_C$ is connected.
Therefore we can pick a path $P_C$ between $z_C$ and $u_C$ that is contained in $C\setminus p_C$.
It follows that for any $C\neq C'\in {\cal L}'$ we have $V(P_C) \cap V(P_{C'}) = \emptyset$.
We can now obtain a $K_{3,r}$ minor in $G$ by setting the left side to be $\{x_1,x_2,x_3\}$ and the right side containing a vertex for every $P_C$ with $C\in {\cal L}'$.
By Lemma \ref{lem:K33} it follows that $|U|=O(g)$, concluding the proof.

(2)
We greedily choose a sequence of pairs $w_i,H_i$, where $H_i$ is  $w_i$-propeller in $H$ until all components in ${\cal L}''$ are covered, as follows.
Suppose we have inductively chosen $w_1,\ldots,w_{i-1}$.
We pick such an uncovered component $C\in {\cal L}''$.
By construction $C$ contains some 1-separator $w_i \in W'$.
We set $H_i$ to be the union of all uncovered components in ${\cal L}''$ that contain $w_i$.
This completes the construction of $w_1,\ldots,w_k$ and $H_1,\ldots,H_k$.
Clearly we have $\bigcup_{C\in {\cal L}''} C = \bigcup_{i=1}^k H_i$.
It remains to obtain an upper bound on $k$.
We can construct a $K_{3,k}$ minor in $G$ as follows.
We set the left side to be $\{x_1,x_2,x_3\}$.
For every $i\in \{1,\ldots,k\}$ let $U_i\subseteq H_i$ be an arbitrary $w_i$-petal.
For every $i\in \{1,\ldots,k\}$ the right side of the $K_{3,k}$ minor contains a vertex obtained by contracting $U_i\setminus \{w_i\}$.
This completes the construction of the $K_{3,k}$ minor.
It follows by Lemma \ref{lem:K33} that $k=O(g)$.

(3)
Since the components in ${\cal L}'$ correspond to subtrees of ${\cal T}$ (the biconnected component tree decomposition),
it follows that there exist a $(H,\phi)$-splitting sequence of length $|{\cal L}'|-1$ that separates every pair of components in ${\cal L}'$.
Moreover, for every propeller $H_i$ there exists a single $(H,\phi)$-splitting that separates $H_i$ from the rest of $H$.
It follows that there exists a $(H,\phi)$-splitting sequence of length $|{\cal L'}|-1 + k = O(g)$ that separates every pair of clusters in ${\cal L}'\cup \{V(H_1),\ldots,V(H_k)\}$, concluding the proof.
\end{proof}

We are now ready to give a proof of Lemma \ref{lem:2apices_or_simple1separators}.

\begin{proof}[Proof of Lemma \ref{lem:2apices_or_simple1separators}]
If $|X|\leq 2$ then there is nothing to show.
We may therefore assume that $|X|\geq 3$.
We consider all possible triples $(x_1,x_2,x_3)\in X^3$.

Let $W$ be the set of 1-separators of $H$.
Let
$W' = W \cap N_{G}(x_1) \cap N_{G}(x_2)$.
Let ${\cal L}$ be the set of connected components of the graph obtained by cutting $H$ along $W'$.
Let ${\cal L}'$ be the set of all components in ${\cal L}$ that do not contain any of the leaves of ${\cal T}$ (from the biconnected component tree decomposition) and that contain at least one vertex incident to $x_3$.
Let ${\cal L}''$ be the set of components in ${\cal L}\setminus {\cal L}'$ that contain at least once vertex incident to $x_3$.

Then, by Lemma \ref{lem:x1x2x3} we have $|{\cal L}'| = O(\genus(G)) = O(g)$.
Moreover, there exist distinct $w_1,\ldots,w_k\in W'$, with $k=O(g)$, and for each $i\in \{1,\ldots,k\}$ a $w_i$-propeller $H_i$, such that
\[
\bigcup_{C\in {\cal L}''} C = \bigcup_{i=1}^k H_i.
\]
Also, there exists a $(\phi, H)$-splitting sequence $\sigma_I$ of length $O(g)$ that separates every pair of clusters in ${\cal L}'\cup \{V(H_1),\ldots,V(H_k)\}$.

Setting $\sigma$ be the $(H,\phi)$-splitting sequence obtained by concatenating all sequences $\sigma_I$, $I\in X^3$, we obtain a sequence of length $O(g\cdot |X|^3)$ satisfying all conditions, concluding the proof.
\end{proof}

\subsection{Bounding the number of 2-connected components with at least three apices}

We next show that there is only a bounded number of 2-connected components that are incident to at least three apices.

\begin{lemma}\label{lem:few_2connected_3apices}
Let $G$ be a graph of Euler genus $g$ and let $X\subseteq V(G)$ such that $H=G\setminus X$ is planar.
Suppose that any $1$-separator $u$ in $H$ is incident to at most one vertex in $X$, that is $|N(u)\cap X| \leq 1$.
Let ${\cal L}$ be the set of maximal $2$-connected components of $H$.
Then there are at most $O(|X|^3 \cdot g^{2})$ components in ${\cal L}$ that are incident to at least three vertices in $X$.
\end{lemma}

\begin{proof}
Let $x_1,x_2,x_3\in X$ be distinct.
Let ${\cal L}'$ be the set of all components in ${\cal L}$ that are incident to all three vertices $x_1,x_2,x_3$.
We will first upper bound $|{\cal L}'|$.
The assertion will then follow by considering all possible triples of vertices in $X$.

Let $t$ be the maximum number of components in ${\cal L}'$ that intersect the same component in ${\cal L}'$,
that is $t=\max_{C'\in {\cal L}'}\left|\{C\in {\cal L}' : C\cap C' \neq \emptyset\}\right|$.
Let $C'\in {\cal L}'$ intersect $t$ components in ${\cal L}'$.
Let ${\cal W} = \{C\in {\cal L}' : C\cap C' \neq \emptyset\}$.
For any $C\in {\cal W}$ the subgraph $C\cap C'$ consists of exactly one 1-separator in $H$.
For any $j\in \{1,2,3\}$ let ${\cal W}_j$ be the set of components $C\in {\cal W}$ such that the 1-separator in $C\cap C'$ is incident to $x_j$, that is
\[
{\cal W}_j = \{C\in {\cal W} : C\cap C' = \{v\} \text{ for some } v \text{ with } \{v,x_j\}\in E(G)\}.
\]
Let also ${\cal W}_0 = {\cal W} \setminus \left(\bigcup_{\ell=1}^3 {\cal W}_{\ell} \right)$.
We can construct a $K_{3,\Omega(|{\cal W}_{0}|)}$ minor in $G$ as follows.
The left side consists of $\{x_1,x_2,x_3\}$ and the right side contains for each $C\in {\cal W}_0$ a vertex obtained by contracting $C\setminus C'$.
We can also construct a $K_{3,\Omega(|{\cal W}_{1}|)}$ minor in $G$ as follows.
The left side consists of $\{x_1',x_2,x_3\}$, where $x_1'$ is the vertex obtained by contracting $x_1\cup C'$.
The right side contains for each $C\in {\cal W}_0$ a vertex obtained by contracting $C\setminus C'$.
Similarly, we can construct a $K_{3,\Omega(|{\cal W}_{2}|)}$ minor and a $K_{3,\Omega(|{\cal W}_{3}|)}$ minor in $G$.
By Lemma \ref{lem:K33} we conclude that for each $\ell\in \{0,1,2,3\}$ we have $|{\cal W}_{\ell}| = O(g)$.
Therefore, $t = |{\cal W}| \leq \sum_{\ell=0}^3 |{\cal W}_{\ell}| = O(g)$.
Thus, every component in ${\cal L}'$ intersects at most $O(g)$ other components in ${\cal L}'$.
It follows that there exists a collection ${\cal Q}$ of $t'=\Omega(|{\cal L}'|/g)$ components in ${\cal L}'$ that are pairwise vertex-disjoint.
We can thus construct a $K_{3,t'}$ minor in $G$ where the left side is $\{x_1,x_2,x_3\}$ and for any $C\in {\cal Q}$ the right side contains a vertex obtained by contracting $C$.
By Lemma \ref{lem:K33} it follows that $t'=O(g)$ and therefore $|{\cal L}'| = O(t \cdot t') = O(g^2)$.

It follows that $|{\cal L}| \leq {n \choose 3} \cdot |{\cal L}'| = O(|X|^3 \cdot g^2)$, concluding the proof.
\end{proof}

\subsection{Isolating a connected collection of 2-connected components}
We also need the following lemma that will be used in the main
algorithm of this section, Lemma~\ref{lem:extending_isolated}.
It allows us to extend an embedding of one 2-connected piece to an embedding of a component with a planar piece that consists of the connected union of maximal 2-connected components.

\begin{lemma}[Isolating a 2-connected component]\label{lem:isolating_2-connected}
Let $G$ be a graph of Euler genus $g$ and let $X\subseteq V(G)$ such that $H=G\setminus X$ is planar.
Let ${\cal Y}$ be a set of maximal 2-connected components of $H$ such that $C=\bigcup_{Y\in {\cal Y}} Y$ is connected.
Let ${\cal C}$ be the set of maximal connected subgraphs of $H$ that do not contain any edges in $E(C)$.
Then the following conditions are satisfied:
\begin{description}
\item{(1)}
Let
${\cal C}_1$ be the set of components $C' \in {\cal C}$ such that $|N_G(C' \setminus C) \cap X| \geq 2$.
Then, $|{\cal C}_1| = O(g \cdot |X|^2)$.
\item{(2)}
Let
${\cal C}_2$ be the set of components $C' \in {\cal C}$ such that $|N_G(C' \setminus C) \cap X| = 1$ and $G[C'\cup X]$ is non-planar.
Then, $|{\cal C}_2| = O(g\cdot |X|)$.
\item{(3)}
Let ${\cal C}_3$ be the set of components $C'\in {\cal C}$ such that $|N_G(C' \setminus C) \cap X| = 1$, $G[C'\cup X]$ is planar, and if we let $\{v\}=V(C)\cap V(C')$ and $\{u\}=N_G(C'\setminus C)\cap X$, then there exists no planar drawing of $G[C'\cup X]$ in which $v$ and $u$ are in the same face.
Then, $|{\cal C}_3| = O(g \cdot |X|)$.
\item{(4)}
Let ${\cal C}_4 = {\cal C} \setminus ({\cal C}_1 \cup {\cal C}_2 \cup {\cal C}_3)$.
Then each $C'\in {\cal C}_4$ is an extremity (w.r.t. $X$).
\end{description}
\end{lemma}

\begin{proof}
(1)
By averaging, it follows that there exists a pair of distinct vertices $x_1,x_2\in X$ and a subset ${\cal C}_1'\subseteq {\cal C}_1$ with $|{\cal C}'_1| \geq |{\cal C}_1|/{|X| \choose 2}$ such that for every $C'\in {\cal C}_1$ we have $\{x_1,x_2\}\subseteq N(C'\setminus C)$.
We can construct a $K_{3,|{\cal C}_1'|}$ minor in $G$ as follows.
The left side consists of $\{x_1,x_2,x_3\}$ where $x_3$ is obtained by contracting $C$ into a single vertex.
It follows that for every $C'\in {\cal C}_1'$, the graph $C'\setminus C$ is connected.
For every $C'\in {\cal C}'_1$, the right side of the minor contains a vertex obtained by contracting $C'\setminus C$.
It follows by Lemma \ref{lem:K33} that $|{\cal C}'_1| = O(g)$, and therefore $|{\cal C}_1| = O(g \cdot |X|^2)$.

(2)
Let $G'$ be the graph obtained from $G$ by contracting $C$ into a single vertex $v_C$.
By averaging, it follows that there exists $x\in X$ and ${\cal C}'_2 \subseteq {\cal C}_2$ with $|{\cal C}_2'| \geq |{\cal C}_2|/|X|$ such that for any $C'\in {\cal C}'_2$ $N(C'\setminus C) \cap X = \{x\}$ and $G'[C'\cup \{x\}]$ is non-planar.
By Lemma \ref{lem:2sum} we have $|{\cal C}'_2| = O(\genus(G')) = O(\genus(G)) = O(g)$.
Therefore, $|{\cal C}_2| = O(g\cdot |X|)$.

(3)
Let $G'$ be the graph obtained from $G$ by contracting $C$ into a single vertex $v_C$.
By averaging, it follows that there exists $x\in X$ and ${\cal C}'_3 \subseteq {\cal C}_3$ with $|{\cal C}_3'| \geq |{\cal C}_3|/|X|$ such that for any pair of distinct $C',C''\in {\cal C}_3'$ the graph $G'[C'\cup C''\cup \{x\}]$ is non-planar.
By Lemma \ref{lem:2sum} it follows that $|{\cal C}'_3| = O(g)$.
Therefore, $|{\cal C}_3| = O(g\cdot |X|)$.

(4)
Follows immediately by the definition of ${\cal C}_1$, ${\cal C}_2$, ${\cal C}_3$, and ${\cal C}_4$.
\end{proof}

We next show how the above result can be used in conjunction with the tools for dealing with extremities.

\begin{lemma}[Extending the embedding of a 2-connected component]\label{lem:extending_isolated}
Let $G$ be a graph and let $X\subseteq V(G)$ such that $H=G\setminus X$ is planar.
Let ${\cal Y}$ be a set of maximal 2-connected components of $H$ such that $C=\bigcup_{Y\in {\cal Y}} Y$ is connected.
Let ${\cal C}_4$ be the set of subgraphs of $H$ given by Lemma \ref{lem:isolating_2-connected}.
Let $G'$ be the graph obtained by contracting each $C'\in {\cal C}_4$ into a single vertex $v_{C'}$.
Then, given an embedding $\phi'$ of $G'[C\cup X]$ into a surface of Euler genus $\gamma$ we can compute in polynomial time an embedding of $G\left[C\cup X \cup \bigcup_{C'\in {\cal C}_4} C'\right]$ into a surface of Euler genus $\gamma$.
\end{lemma}

\begin{proof}
Follows directly by Lemma \ref{lem:contracting_extremities}.
\end{proof}


\section{Embedding $2$-apex graphs}\label{sec:2-apex}

This section deals with the the case of $2$-apex graphs.
We first show that any splitting of the graph into non-planar ``parts'', has a small number of components.





\begin{lemma}\label{lem:2-apex-components}
Let $G$ be a planar graph of Euler genus $g$ and let $X$, with $|X|\leq 2$, such that $H=G\setminus X$ is planar.
Let $({\cal T},{\cal B})$ be
a biconnected component tree decomposition of $H$
(that is, each $C\in {\cal B}$ is the union of maximal 2-connected components of $H$ and for each $C, C' \in H$ we have either $C\cap C'=\emptyset$ or $C\cap C'=\{v\}$, for some 1-separator $v$ of $H$).
Suppose that for each $C\in {\cal B}$, $G[V(C) \cup X]$ is not planar.
Then $|{\cal B}| = O(g^4)$.
\end{lemma}

\begin{proof}
We can assume that $|X|=2$, since otherwise we can add an isolated vertex to $X$.
Let $X=\{x_1,x_2\}$.
Let
\[
{\cal B}_1 = \{C \in {\cal B} : G[C\cup \{x_1\}] \text{ is non-planar}\},
\]
\[
{\cal B}_2 = \{C \in {\cal B} : G[C\cup \{x_2\}] \text{ is non-planar}\},
\]
and ${\cal B}_3 = {\cal B} \setminus ({\cal B}_1 \cup {\cal B}_2)$.
We will bound the size of each ${\cal B}_i$ independently.

Let us first argue that $|{\cal B}_1|=O(g^2)$.
Let $|{\cal B}_1| = \ell$.
Let ${\cal B}_1 = \{C_1,\ldots,C_\ell\}$.
For any $C \in {\cal B}_1$, since $G[C\cup \{x_1\}]$ is non-planar, it follows that for any planar drawing $\psi_C$ of $C$, the size of the minimum $\phi_C$-face cover of $N_G(x_1)\cap V(C)$ is at least two.
For any $i\in \{1,\ldots,\ell\}$ let $H_i = \bigcup_{j=1}^{\ell} C_j$.
We argue by induction on $i$ that for any $i\in \{1,\ldots,\ell\}$, for any planar drawing $\phi_i$ of $H_i$, the size of the minimum $\phi_i$-face cover of $N_G(x_1)\cap V(H_i)$ is at least $i+1$.
The base case follows by the fact that in any planar drawing $\phi_1$ of $H_1$ the size of the minimum $\phi_1$-face cover of $N_G(\{x_1\}, V(H_1)) = N_G(\{x_1\}, V(C_1))$ is at least two.
Suppose now that the assertion holds for some $i\in \{1,\ldots,\ell-1\}$.
Consider some planar drawing $\phi_{i+1}$ of $H_{i+1}$.
Let also ${\cal F}_{i+1}$ be a minimum $\phi_{i+1}$-face cover of $N_G(\{x_1\})\cap V(H_{i+1})$.
The restriction of $\phi_{i+1}$ on $C_{i+1}$ is a planar drawing $\phi'_{i+1}$.
Moreover, the $\phi_{i+1}$-face cover ${\cal F}_i$ induces a $\phi'_{i+1}$-face over ${\cal F}_i'$ of $N_G(x_1)\cap V(C_{i+1})$.
Similarly, the restriction of $\phi_{i+1}$ on $H_i$ is a planar drawing $\phi_{i+1}''$, and the $\phi_{i+1}$-face cover ${\cal F}_{i+1}$ induces a $\phi_{i+1}''$-face cover of $N_G(\{x_1\}) \cap V(H_i)$.
By the inductive hypothesis, we have $|{\cal F}_{i+1}''|\geq i+1$.
Since $C\in {\cal B}_1$ we obtain by arguing as above that $|{\cal F}_{i+1}'|\geq 2$.
Therefore, there exists a least some face in ${\cal F}_{i+1}'$ that is not the outer face of $\phi_{i+1}'$.
It follows that
$|{\cal F}_{i+1}| \geq |{\cal F}_{i+1}'| + |{\cal F}_{i+1}''| - 1 \geq i + 1 + 2 - 1 = i+2$,
as required.
We conclude that for any planar drawing $\psi$ of $H_\ell$ the minimum face cover of $N_G(x_1)\cap V(H_\ell)$ is at least $\ell+1$.
By Lemma \ref{lem:face_cover} we obtain $\ell=O(g^2)$, as required.

Similarly we can show that $|{\cal B}_2| = O(g^2)$.

It remains to bound $|{\cal B}_3|$.
Suppose that $|{\cal B}_3| = t$.
We may assume w.l.o.g.~that there exists some ${\cal B}_3'\subseteq {\cal B}_3$, with $|{\cal B}_3'|=t'\geq t/2$, and such that for any $C\in {\cal B}_3$ the graph $G[C\cup \{x_1\}]$ is planar.
Let ${\cal B}_3''$ be a maximal subset of pairwise vertex-disjoint components in ${\cal B}_3'$.
Let $|{\cal B}_3''| = t''$.
Let ${\cal B}_3''= \{X_1,\ldots,X_{t''}\}$.
For any $i\in \{1,\ldots,t''\}$ let $Z_i=G[X_i\cup \{x_1\}]$, and recall that $Z_i$ is a planar graph.
Since two graphs $Z_i$, $Z_j$ with $i\neq j$ can intersect only on $x_1$, it follows that the graph $Z=\bigcup_{i=1}^{t''} Z_i$ is planar.
Since for any $i\in \{1,\ldots,t''\}$, $G[C\cup \{x_1,x_2\}]=G[Z_i\cup \{x_2\}]$ is non-planar, it follows that for any planar drawing $\psi$ of $Z_i$, the size of the minimum $\psi$-face cover of $N_G(x_2)\cap V(Z_i)$ is at least two.
Arguing by induction, we can conclude that for any planar drawing $\psi$ of $Z$, the size of a minimum $\psi$-face cover of $N_G(x_2)\cap V(Z)$ is at least $t''+1$.
By Lemma \ref{lem:face_cover} we obtain that $t''=O(g^2)$.
On the other hand, there exists a collection ${\cal P}\subseteq {\cal B}_3$ with $|{\cal P}|=p\geq \Omega(|{\cal B}_3'| / |{\cal B}_3''|) = \Omega(t'/t'') = \Omega(t/t'')$, and some $v\in V(H)$ such that for any $C\in {\cal P}$, $v\in V(C)$.
Let ${\cal P}'\subseteq {\cal P}$ be the set containing all $C\in {\cal P}$ such that $G[C\cup \{x_1\}]$ admits a planar drawing where $x_1$ and $v$ are in the same face.
Let $|{\cal P}'|=p'$.
Let also ${\cal P}'' = {\cal P} \setminus {\cal P}'$ and let $|{\cal P}''| = p''$.
Then, the graph $W=G\left[\{x_1\} \cup \bigcup_{C\in {\cal P}'} V(C)\right]$ is planar.
We can therefore argue by induction as above that for any planar drawing $\xi$ of $W$, the size of the minimum $\xi$-face cover of $N_G(x_2)\cap V(W)$ is at least $p''+1$.
By Lemma \ref{lem:face_cover}, we obtain $p''=O(g^2)$.
Similarly, let $U=\bigcup_{C\in {\cal P}''} C$.
The graph $U$ is planar since it is a subgraph of $H$.
We can argue by induction as above that for any planar drawing $\zeta$ of $U$, the size of the minimum $\zeta$-face cover of $N_G(x_1)\cap V(U)$ is at least $p''+1$.
By Lemma \ref{lem:face_cover} we obtain $p''=O(g^2)$.
Therefore, $p = p'+p'' = O(g^2)$.
Thus, $t = O(p\cdot t'') = O(g^4)$.

We thus obtain $|{\cal B}| \leq |{\cal B}_1| + |{\cal B}_2| + |{\cal B}_3| = O(g^4)$, concluding the proof.
\end{proof}

We now use the previous lemma to give a decomposition
theorem for 2-apex graphs.

\begin{lemma}\label{lem:2-apex-decomposition}
Let $G$ be a planar graph of Euler genus $g$ and let $X\subseteq V(G)$, with $|X|=2$, such that $H=G\setminus X$ is planar.
Let ${\cal C}$ be the set of maximal 2-connected components of $H$.
Then there exists a decomposition ${\cal C}={\cal C}_0\cup \ldots\cup {\cal C}_{2t}$, for some $t=O(g^4)$, such that the following conditions are satisfied.
For any $i\in \{0,\ldots,2t\}$ let $V_i = \bigcup_{C\in {\cal C}_i} V(C)$ and $H_i = H[V_i]$. Then:
\begin{description}
\item{(1)}
For any connected component $Z$ of $H_0$, we have that $G[V(Z)\cup X]$ is planar.
\item{(2)}
For any $i\in \{1,\ldots,2t\}$, $H_i$ is connected.
Moreover, $G[V_i\cup X]$ is planar or $|{\cal C}_i|=1$.
\end{description}
Moreover, there exists a polynomial-time algorithm which given $G$, $g$, and $X$, outputs ${\cal C}_0,\ldots,{\cal C}_{2t}$.
\end{lemma}

\begin{proof}
Let ${\cal L}$ be the set of maximal 2-connected components of $H$.
We inductively define the sequence of pairwise disjoint sets ${\cal C}_1,\ldots,{\cal C}_{2t}$ of ${\cal C}$; the set ${\cal C}_0$ will be defined at the end of the induction.
Consider some integer $i\geq 1$, and suppose that we have already defined ${\cal C}_1,\ldots,{\cal C}_{2i}$.
We proceed to define ${\cal C}_{2i+1}$ and ${\cal C}_{2i+2}$.
Let $\Gamma_i=H\left[{\cal C} \setminus \bigcup_{j=1}^{2i} C\right]$.
If for each connected component $Z$ of $\Gamma_i$, the graph $G[V(Z)\cup X]$ is planar, then we set $t=i$, we terminate the sequence ${\cal C}_0,\ldots,{\cal C}_{2t}$; for each connected component $Z$ of $\Gamma_i$, we add a set ${\cal C}_Z$ to ${\cal C}_0$, where ${\cal C}_Z$ consists of all the components in ${\cal C}$ that are contained in $Z$.
Otherwise, let $Z$ be a connected component of $\Gamma_i$ such that $G[V(Z)\cup X]$ is non-planar.
Let ${\cal M}_i$ be a maximal (possibly empty) subset of ${\cal C} \setminus \bigcup_{j=1}^{2i} {\cal C}_j$ such that the graph $G\left[X\cup \left(\bigcup_{C\in {\cal M}_i} C\right)\right]$ is planar.
By the maximality of ${\cal M}_i$ it follows that there exists come $C_i\in \left({\cal C}\setminus \bigcup_{j=1}^{2i} {\cal C}_i\right)\setminus {\cal M}_i$, such that $C_i$ intersects some $C\in {\cal M}_i$, and such that the graph $G\left[X\cup C_i \cup \left(\bigcup_{C\in {\cal M}_i} C\right)\right]$ is non-planar.
We set ${\cal C}_{2i+1}={\cal M}_i$, and ${\cal C}_{2i+2} = \{C_i\}$.
This completes the definition of  ${\cal C}_0,\ldots,{\cal C}_{2t}$.

It is immediate by the construction of ${\cal C}_0$ that for any connected component $Z$ of $H_0$ we have that $G[V(Z)\cup X]$ is planar, and thus assertion (1) is satisfied.
For any $i\in \{1,\ldots,t\}$, we have that $G[V_{2i+1}\cup X]$ is planar and $|{\cal C}_{2i+2}|=1$.
Thus, assertion (2) is also satisfied.
It remains to bound $t$.
For any $i\in \{1,\ldots,t\}$ let $H_i' = H_{2i+1}\cup H_{2i+2}$.
The collection of subgraphs ${\cal B}=\{H_1',\ldots,H_t'\}$ satisfies the conditions of Lemma \ref{lem:2-apex-components}.
It follows by Lemma \ref{lem:2-apex-components} that $t=|{\cal B}| = O(g^4)$, concluding the proof.
\end{proof}

Putting everything together, we obtain the following algorithm for 2-apex graphs.

\begin{proof}[Proof of Lemma \ref{lem:embedding_2-apex}]
Let $g=\eg(G)$.
We define a sequence of subgraphs $G_0\subset G_1 \subset \ldots \subset G_k$, for some
$k=O(g^4)$
to be specified, with $G_k = G$.
Let ${\cal L}$ be the set of maximal 2-connected components of $H$.
Let ${\cal C}={\cal C}_0,\ldots,{\cal C}_{2t}$ be the sequence of subsets of ${\cal L}$ computed by Lemma \ref{lem:2-apex-decomposition}, for some
$t=O(g^4)$.
Let $k=2t$.
For any $i\in \{0,\ldots,k\}$, let $H_i = \bigcup_{C\in {\cal C}_i} C$, $V_i=V(H_i)$, and
\[
G_i = G[V_0 \cup \ldots \cup V_i \cup X].
\]

We now proceed to inductively compute a sequence $\phi_0,\ldots,\phi_k$, where each $\phi_i$ is an embedding of $G_i$ into some surface ${\cal S}_i$.
Let us first compute an embedding for $G_0$.
By the definition of ${\cal C}_0$ in Lemma \ref{lem:2-apex-decomposition} for any connected component $J$ of $H_0$, the graph $G[V(J)\cup X]$ is planar.
Let $x_1\in X$.
Clearly, for any connected component $J$ of $H_0$, the graph $G[V(J)\cup \{x_1\}]$ is planar.
Therefore the graph $G[H\cup \{x_1\}]$ is obtained via 1-sums of planar graphs, and therefore it is planar.
Thus the graph $G[H\cup X]$ is 1-apex.
By Corollary \ref{cor:1-apex} we can compute an embedding $\phi_0$ of $G_0$ into a surface ${\cal S}_0$ of Euler genus $O(g^2)$.

Suppose now that $\phi_i$ has been computed for some $i\in \{0,\ldots,k-1\}$.
We proceed to compute $\phi_{i+1}$.
Let ${\cal L}$ be the set of maximal 2-connected components $H$.
Note that since $i+1\geq 1$, by Lemma \ref{lem:2-apex-decomposition} we have that $H_{i+1}$ is connected.
Let ${\cal C}$ be the set of all maximal connected subgraphs of $H$ that do not contain any edges in $E[H_{i+1}]$.
Let ${\cal C}={\cal C}_1 \cup \ldots \cup {\cal C}_4$ be the decomposition computed by Lemma \ref{lem:isolating_2-connected}.
We have $|{\cal C}_1|+|{\cal C}_2| + |{\cal C}_3| = O(g\cdot |X|^2) = O(g)$, since $|X|=2$.
Let $G_{i+1}'$ be the graph obtained by contracting each $C'\in {\cal C}_4$ into a single vertex $v_{C'}$.
We shall first compute an embedding $\psi_{i+1}$ of $G'_{i+1}$ into a surface of small genus.
Subsequently, using Lemma \ref{lem:extending_isolated} we will compute the embedding $\phi_{i+1}$ given $\psi_{i+1}$.

Let us first compute an embedding $\xi_{i+1}$ of $G'[V_{i+1} \cup X]$ into a surface of small genus.
By condition (2) of Lemma \ref{lem:2-apex-decomposition} we have that $G[V_{i+1}\cup X]$ is planar or $|{\cal C}_{i+1}|=1$.
In the former case, by the definition of ${\cal C}_4$ in Lemma \ref{lem:isolating_2-connected} the graph $G'[V_{i+1}\cup X]$ must also be planar; we set $\xi_{i+1}$ to be any planar drawing of $G'[V_{i+1}\cup X]$.
In the latter case, the graph $G[V_{i+1}\cup X]$ has a $2$-connected planar piece.
Since $G'$ is obtained from $G$ by contracting each cluster in ${\cal C}_4$ into a single vertex, it follows that the graph $G'[V_{i+1}\cup X]$ also has a 2-connected planar piece.
Therefore, we can compute using Lemma \ref{lem:genus_2-connected} an embedding $\xi_{i+1}$ of $G'[V_{i+1}\cup X]$ into a surface of Euler genus
$O(g^{11}\cdot |X|^6 + g^5 \cdot |X|^{12}) = O(g^{11})$.
In either case, we obtain an embedding $\xi_{i+1}$ of $G'[V_{i+1}\cup X]$ into a surface ${\cal M}_{i+1}$ of Euler genus at most
$O(g^{11})$.

Let ${\cal C}' = {\cal C}_1 \cup {\cal C}_2 \cup {\cal C}_3$.
By the inductive hypothesis, $\phi_i$ is an embedding of $G[V_0\cup\ldots\cup V_i \cup X]$ into a surface ${\cal S}_i$ of Euler genus
$O(i\cdot g^{11})$.
Since $G'$ is a minor of $G$ we can compute an embedding $\phi_i'$ of $G'\left[X\cup \bigcup_{K\in {\cal C}'} K \right]$ into a surface ${\cal S}'_i$ of Euler genus
$O(i\cdot g^{11})$.

For every $K\in {\cal C}'$, the graphs $K$ and $H_{i+1}$ share at most one vertex.
Thus, the graphs $G'\left[X\cup \bigcup_{K\in {\cal C}'} K\right]$ and $G'[V_{i+1}\cup X]$ share at most $|X|+|{\cal C}'|= O(g)$ vertices (since $|X|=2$).
Therefore, we can extend $\xi_{i+1}$ to $G'[V_0\cup \ldots \cup V_{i+1} \cup X]$ by adding at most $O(g)$ cylinders, each connecting a puncture in ${\cal M}_{i+1}$ to a puncture in ${\cal S}'_{i}$.
It follows that the resulting embedding $\psi_{i+1}$ of $G'[V_0\cup \ldots \cup V_{i+1}\cup X]$ is into a surface of Euler genus at most $\eg({\cal S}_i') + \eg({\cal M}_{i+1}) + O(g) = \eg({\cal S}'_i) + O(g^3)$.
By Lemma \ref{lem:extending_isolated} we can now compute an embedding $\phi_{i+1}$ of $G_{i+1}$ into a surface of Euler genus at most $\eg(\psi_{i+1}) + O(g^{11}) = O((i+1) \cdot g^{11})$, as required.

Since $G=G_k$, it follows that we can compute an embedding of $G$ into a surface of Euler genus $O(k\cdot g^{11}) = O(g^{15})$, concluding the proof.
\end{proof}

\section{Embedding locally 2-apex graphs and their generalizations}\label{sec:2-apex_generalizations}

This section deals with the following case:
$G$ is a graph and let $X\subseteq V(G)$ and $H=G\setminus X$.
For any 2-connected maximal component $C$, $C$ has at most two neighbors
in $X$. We call such a graph \DEF{locally 2-apex}.

Formally, we have the following notation.

\begin{definition}[Locally 2-apex]
Let $G$ be a graph and let $X\subseteq V(G)$ such that $H=G\setminus X$ is planar.
We say that $G$ is \emph{locally 2-apex} (w.r.t.~$X$) if for every maximal 2-connected component $C$ of $H$, $C$ is connected to at most two vertices in $X$, that is $|N(C)\cap X|\leq 2$.
\end{definition}

We next derive a decomposition of locally 2-apex graphs.

\begin{lemma}[Tree-decompositions of  locally 2-apex graphs into 1-sums of 2-apex graphs]\label{lem:locally-2-apex_decomposition}
Let $G$ be a graph of Euler genus $g$ and let $X\subseteq V(G)$, such that $H=G\setminus X$ is planar.
Suppose that $G$ is locally 2-apex (w.r.t.~$X$) and that every 1-separator $v$ of $H$ is incident to at most one vertex in $X$, that is $|N_G(v)\cap X| \leq 1$.
Suppose further that the extremity number of $G$ is $M$.
Then there exists a tree decomposition $({\cal T}, {\cal C})$ of $H$ and some ${\cal P}\subseteq {\cal C}$, such that the following conditions are satisfied:
\begin{description}
\item{(1)}
The intersection of any two distinct clusters in ${\cal C}$ is either empty or consists of a 1-separator of $H$.

\item{(2)}
Any $C\in {\cal C}\setminus {\cal P}$ is incident to at most two vertices in $X$, that is $|N_G(V(C))\cap X| \leq 2$ (thus $G[C\cup X]$ is 2-apex).

\item{(3)}
$|V({\cal T})| = O(g^4 \cdot |X|^8 + g^3 \cdot |X|^6 \cdot M + g^2 \cdot |X|^4 \cdot M^2)$.



\end{description}
Moreover, there exists a polynomial-time algorithm which given $G$, $g$, and $X$, outputs $({\cal T}, {\cal C})$ and ${\cal P}$.
\end{lemma}

\begin{proof}
We compute the tree decomposition $({\cal T}, {\cal C})$ greedily as follows.
We begin with the tree ${\cal T}$ being empty and we inductively add vertices to ${\cal T}$.
Let ${\cal L}$ be the set of maximal 2-connected components of $H$.
Pick some maximal ${\cal Y}\subseteq {\cal L}$ such that the graph $C=H\left[\bigcup_{Y\in {\cal Y}} V(Y)\right]$ is connected and $|N_G(V(C)) \cap X| \leq 2$.
We add $C$ to $V({\cal T})$ and we consider ${\cal T}$ to be the root of ${\cal T}$.

Suppose now that we have added some $C$ to $V({\cal T})$.
Let ${\cal N}$ be the set of maximal connected subgraphs of $H$ that do not have any edges in the part of $H$ that has been added to the tree decomposition $({\cal T}, {\cal C})$ and intersect $C$; in particular, every such subgraph intersects $C$ at some 1-separator of $H$.
Let ${\cal N}={\cal C}_1\cup {\cal C}_2\cup {\cal C}_3\cup {\cal C}_4$ be the decomposition of ${\cal N}$ given by Lemma \ref{lem:isolating_2-connected}.
We have $|{\cal C}_1| + |{\cal C}_2| + |{\cal C}_3| = O(g\cdot |X|^2)$.
For each $H'\in {\cal C}_1\cup {\cal C}_2\cup {\cal C}_3$, we recursively find a tree decomposition $({\cal T}',{\cal C}')$ for $H'$, rooted at some subgraph $C'\subset H$ that shares a 1-separator with $C$.
We add ${\cal T}'$ to ${\cal T}$ by making $C'$ a child of $C$.
Finally, we deal with the components in ${\cal C}_4$.
Every component in ${\cal C}_4$ is a maximal extremity, and any two distinct $H',H''\in {\cal C}_4$ are edge-disjoint.
Therefore, $|{\cal C}_4| \leq M$ by the assumption.
We add every $H'\in {\cal C}_4$ to $V({\cal T})$ by making it a child of $C$.
This completes the definition of $({\cal T}, {\cal C})$.
It is immediate by the construction that every node in ${\cal T}$ has at most $\delta=|{\cal C}_1|+\ldots+|{\cal C}_4| = O(M+g\cdot |X|^2)$ children.


Let $h$ be the height of ${\cal T}$.
Let $C_1,\ldots,C_{h+1}$ be a branch of ${\cal T}$ where $C_1$ is the root and $C_{h+1}$ is a leaf.
It follows that there exist
\[
y_1<\ldots < y_{2t}\in \{1,\ldots,h+1\}
\]
for some $t\geq \frac{h}{2{|X|\choose 2}} = \Omega(h/\delta)$ such that for any $i\in \{1,\ldots, t\}$,
\[
N_G(C_{y_1})\cap X = N_G(C_{y_{2i-1}})\cap X,
\]
\[
N_G(C_{y_2})\cap X = N_G(C_{y_{2i}})\cap X,
\]
and
\[
|N_G(C_{y_1}) \cup N_G(C_{y_2}) \cap X| \geq 3.
\]
We can therefore obtain a $K_{3,\lfloor t/2 \rfloor}$ minor in $G$ where the left side contains any three distinct vertices in $N_G(C_{y_1}) \cup N_G(C_{y_2}) \cap X$, and for each $i\in \{1,\ldots,\lfloor t/2 \rfloor\}$ the right side contains a vertex obtained by contracting $C_{y_{4t-3}} \cup C_{y_{4t-2}}$.
It follows by Lemma \ref{lem:K33} that
$t = O(g)$
and thus
\begin{align}
h &= O(t \cdot \delta) = O(g\cdot M + g^2\cdot |X|^2). \label{eq:h1}
\end{align}

Let $k=|V({\cal T})$.
For each $i\in \{0,\ldots,h\}$, let $X_i$ be the set of vertices in $T$ that are at distance exactly $i$ from the root.
By averaging, there exists some $d\in \{1,\ldots,h\}$ such that $|X_d|\geq k/(h+1)$.
Since each vertex in $T$ has at most $\delta$ children, it follows that there exist some $Y_{d-1}\subseteq X_{d-1}$
and some $Y_d \subseteq X_d$,
with $|Y_{d-1}|=|Y_d| \geq |X|_d/\delta = \Omega\left(\frac{k}{h\cdot \delta}\right)$,
such that each $C\in Y_{d-1}$ has a unique child in $Y_d$ and each $C'\in Y_d$ has a unique parent in $Y_{d-1}$.
Furthermore, there must exist
$Z_{d-1}\subseteq Y_{d-1}$ and $Z_d\subseteq Y_d$, with $|Z_{d-1}|=|Z_d|\geq \left(\frac{|Y_d|}{|X|^4}\right)=\Omega\left(\frac{k}{h\cdot \delta \cdot |X|^4}\right)$, such that each $C\in Z_{d-1}$ has a unique child in $Z_{d}$, each $C'\in Z_d$ has a unique parent in $Z_{d-1}$, for any $C,C'\in Z_{d-1}$,
\[
N_G(C)\cap X = N_G(C') \cap X,
\]
and
for any $C,C'\in Z_{d}$,
\[
N_G(C)\cap X = N_G(C') \cap X.
\]
We can now obtain a $K_{3,r}$ minor in $G$, for some $r = |Z_d| = \Omega\left(\frac{k}{h\cdot \delta \cdot |X|^4}\right)$ as follows.
For each $C_v\in Z_{d-1}$ with unique child $C_u\in Z_d$, the right size of $K_{3,r}$ contains a vertex obtained by contracting $C_v\cup C_u$.
Let $C_{v'}\in Z_{d-1}$ with $C_{v'}\not=C_v$, and let $C_{u'}$ be the unique child of $C_v$ in $Z_d$.
Let $U=N_G(V(C_{v'})) \cap N_G(V(C_{u'})) \cap X$.
Then, $|U|\geq 3$ by the maximality of $C_{v'}$ and $C_{u'}$ (by the construction of ${\cal T}$).
We set the left side of $K_{3,r}$ to be any three distinct vertices in $U$.
By Lemma \ref{lem:K33} we obtain $r=O(g)$.
Thus
\begin{align}
\frac{k}{h\cdot \delta \cdot |X|^4} &= O(g). \label{eq:h2}
\end{align}
By \eqref{eq:h1} and \eqref{eq:h2} we get
$k = O(g\cdot h \cdot \delta \cdot |X|^4) = O(g \cdot (g\cdot M + g^2\cdot |X|^2) \cdot (M+g\cdot |X|^2) \cdot |X|^4) = O(g^4 \cdot |X|^8 + g^3 \cdot |X|^6 \cdot M + g^2 \cdot |X|^4 \cdot M^2)$, as required.
\end{proof}

Putting everything together, we obtain the following algorithm
for locally 2-apex graphs.

\begin{lemma}[Embedding locally 2-apex graphs]\label{lem:embedding_locally_2-apex}
Let $G$ be a planar graph of Euler genus $g$ and let $X\subseteq V(G)$, such that $H=G\setminus X$ is planar.
Suppose that $G$ is locally 2-apex (w.r.t.~$X$) and that every 1-separator $v$ of $H$ is incident to at most one vertex in $X$, that is $|N_G(v)\cap X| \leq 1$.
Suppose further that the extremity number of $G$ is $M$.
Then there exists a polynomial-time algorithm which given $G$, $g$, $M$, and $X$, outputs a drawing of $G$ into a surface of Euler genus $O(g^{19} \cdot |X|^8 + g^{18} \cdot |X|^6 \cdot M + g^{17} \cdot |X|^4 \cdot M^2)$.
\end{lemma}

\begin{proof}
Let $({\cal T}, {\cal C})$ be the tree decomposition of $H$ given by Lemma \ref{lem:locally-2-apex_decomposition}.
Let $\phi$ be a planar drawing of $H$.
There exists a $(H, \phi)$-splitting sequence $\sigma$ of length $|V({\cal T})|-1$ such that in the graph $H'$ obtained by performing $\sigma$ on $H$ the set of fragments is exactly ${\cal C}$.
Let also $G'$ be the corresponding graph obtained from $G$.
For each $C\in {\cal}$ the graph $G'_C=G'[C\cup X]$ is 2-apex.
By Lemma \ref{lem:embedding_2-apex} we can therefore compute an embedding of $G'_C$ into a surface of Euler genus $O(g^{15})$.
Since for every $C,C'\in {\cal C}$ the graphs $G'_C$ and $G'_{C'}$ share at most two vertices, it follows that we can combine all these embeddings into an embedding of $G'=\bigcup_{C\in {\cal C}} G'_C$ into a surface of Euler genus $O(|{\cal C}|\cdot (2+g^{15})) = O(g^{15}\cdot (g^4 \cdot |X|^8 + g^3 \cdot |X|^6 \cdot M + g^2 \cdot |X|^4 \cdot M^2)) = O(g^{19} \cdot |X|^8 + g^{18} \cdot |X|^6 \cdot M + g^{17} \cdot |X|^4 \cdot M^2)$, as required.
\end{proof}

\subsection{Generalizations of locally 2-apex graphs}

In the previous section, we give an algorithm for locally 2-apex graphs with a small number of extremities.
It turns out that the techniques presented in the previous subsection
can be extended to a larger family, which we call \DEF{nearly locally 2-apex graphs}.
This is necessary for our algorithm for embedding $k$-apex graphs.

\begin{definition}[Nearly locally 2-apex]
Let $G$ be a graph and let $X\subseteq V(G)$ such that $H=G\setminus X$ is planar.
We say that $G$ is \emph{nearly locally 2-apex} (w.r.t.~$X$) if there exists at most one maximal 2-connected component of $H$ that is incident to at least three vertices in $X$.
\end{definition}

\begin{lemma}[Embedding nearly locally 2-apex graphs]\label{lem:embedding_nearly_locally_2-apex}
Let $G$ be a graph of Euler genus $g$ and let $X\subseteq V(G)$, such that $H=G\setminus X$ is planar.
Suppose that $G$ is nearly locally 2-apex.
Suppose further that the extremity number of $G$ is $M$.
Then there exists a polynomial-time algorithm which given $G$, $g$, $M$, and $X$, outputs a drawing of $G$ into a surface of Euler genus $O(g^{20}\cdot |X|^{12} + g^{20} \cdot |X|^{11} \cdot M + g^{19} \cdot |X|^8 \cdot M^2)$.
\end{lemma}

\begin{proof}
If $G$ is locally 2-apex, then we can compute the desired embedding using Lemma \ref{lem:embedding_locally_2-apex}.
We may therefore assume that $G$ is not locally 2-apex.

Let $C$ be the single maximal 2-connected component of $H$ that is incident to at least three vertices in $X$.
Let ${\cal C}$ be the set of all other components of $H$ that intersect $C$.
Let ${\cal C}=\bigcup_{i=1}^4 {\cal C}_i$ be the decomposition given by Lemma \ref{lem:isolating_2-connected}.
We have $|{\cal C}_1| + |{\cal C}_2| + |{\cal C}_3| = O(g\cdot |X|^2)$.
Moreover we can separate $C$ from every component in ${\cal C}_1\cup {\cal C}_2\cup {\cal C}_3$ via a $(H,\phi)$-splitting sequence of length $O(g\cdot |X|^2)$.
Let $G'$ be the resulting graph after applying the splitting sequence.
By Lemma \ref{lem:creating_extremities_splitting_sequence} the extremity number of $G'$ is at most $M'=M+O(g\cdot |X|^3)$.

We contract each component $C'\in {\cal C}_4$ into a single vertex.
Let $G''$ be the resulting graph.
In this graph the planar piece $H''$ is 2-connected.
We can therefore compute an embedding $\xi''$ for $G''[C\cup X]$ using Lemma \ref{lem:genus_2-connected} into a surface of genus
$O(g^{11} \cdot |X|^6 + g^5 \cdot |X|^{12})$.
Given this embedding $\xi''$ of $G''[C\cup X]$, we can compute using Lemma \ref{lem:extending_isolated} an embedding $\xi'$ of $G'\left[C\cup X \cup \bigcup_{C'\in {\cal C}_4} C'\right]$ into a surface of Euler genus
$O(g^{11} \cdot |X|^6 + g^5 \cdot |X|^{12})$.

Next, consider some component $C'\in {\cal C}_1\cup {\cal C}_2 \cup {\cal C}_3$.
The graph $G[V(C')\cup X]$ is locally 2-apex and has extremity number at most $M'$.
Thus, using Lemma \ref{lem:embedding_locally_2-apex}, we can compute an embedding $\xi_{C'}$ of $G[V(C')\cup X]$ into a surface of Euler genus
$O(g^{19} \cdot |X|^8 + g^{18}\cdot |X|^5 \cdot (M') + g^{17} \cdot |X|^4 \cdot (M')^2) = O(g^{19} \cdot |X|^{10} + g^{19} \cdot |X|^9 \cdot M + g^{18} \cdot |X|^6 \cdot M^2)$.
We compute such an embedding for each $C'\in {\cal C}_1\cup {\cal C}_2\cup {\cal C}_3$.

Recall that the splitting sequence $\sigma$ has length $k=O(g\cdot |X|^2)$.
Thus using Lemma \ref{lem:glueing_fragmented} we can combine the embeddings $\{\xi'\} \cup \{\xi_{C'}\}_{C'\in {\cal C}_1\cup {\cal C}_2 \cup {\cal C}_3}$ into an embedding $\xi$ of $G$ into a surface of Euler genus at most $O(k\cdot |X|) + \eg(\xi') + \sum_{C'\in {\cal C}_1\cup {\cal C}_2 \cup {\cal C}_3} \eg(\xi_{C'}) = O(g^{20}\cdot |X|^{12} + g^{20} \cdot |X|^{11} \cdot M + g^{19} \cdot |X|^8 \cdot M^2)$, concluding the proof.
\end{proof}

Next we give an embedding for graphs with ``simple'' 1-separators.
We argue that such graphs are in some sense similar to ``almost'' locally 2-apex graphs.

\begin{proof}[Proof of Lemma \ref{lem:embedding_1-apex_1-separators}]
By Lemma \ref{lem:few_2connected_3apices},
there are at most $O(g^{2} \cdot |X|^{3})$ maximal 2-connected components of $H$ that are incident to at least three vertices in $X$.
Therefore, there exists a planar drawing $\phi$ of $H$ and some $(\phi,H)$-splitting sequence $\sigma$ of length $k=O(g^2 \cdot |X|^3)$ that separates every pair of such components.
Let $H'$ be the graph obtained by performing $\sigma$ on $H$, and let $G'$ be the corresponding graph obtained from $G$.
It follows that for each fragment $C$ of $H$, $G[V(C)\cup X]$ is nearly locally 2-apex.
By Lemma \ref{lem:creating_extremities_splitting_sequence} the extremity number of $G'$ is at most $M'=M+O(g^2 \cdot |X|^4)$.
Thus, by Lemma \ref{lem:embedding_nearly_locally_2-apex} we can compute an embedding $\phi_C$ of $G[V(C)\cup X]$ into a surface of Euler genus
$O(g^{20}\cdot |X|^{12} + g^{20} \cdot |X|^{11} \cdot M' + g^{19} \cdot |X|^8 \cdot (M')^2) = O(g^{22} \cdot |X|^{15} + g^{20} \cdot |X|^{11} \cdot M)$.
Let ${\cal F}$ be the set of fragments of $H'$.
By Lemma \ref{lem:glueing_fragmented} we can combine all the embeddings $\{\phi_C\}_{C\in {\cal F}}$ to obtain an embedding $\phi$ of $G$ into a surface of Euler genus $O(k\cdot |X|) + \sum_{C\in {\cal F}} \eg(\phi_C) = O(g^2\cdot |X|^4) + O(g^2 \cdot |X|^3 \cdot (g^{22} \cdot |X|^{15} + g^{20} \cdot |X|^{11} \cdot M)) = O(g^{24} \cdot |X|^{18} + g^{22} \cdot |X|^{14} \cdot M)$, concluding the proof.
\end{proof}

\section{Face covers and embedding 1-apex graphs}
\label{sec:1-apex}

In this section we present our algorithm for embedding 1-apex graphs.
Towards this end, we extend Mohar's theorem on face covers for 1-apex graphs.
We begin by recalling the notion of a face cover.

\begin{definition}[Face cover]
Let $H$ be a planar graph, and let $\phi$ be a planar drawing of $H$.
Let $U\subseteq V(H)$.
A \emph{$\phi$-face cover} for $U$ is a collection ${\cal F}=\{F_1,\ldots,F_t\}$ of faces of $\phi$ such that $U\subseteq \bigcup_{i=1}^t V(F_i)$.
\end{definition}

An EPTAS for minimum face cover on planar graphs has been obtained by Frederickson \cite{DBLP:journals/jacm/Frederickson91}, and also follows by the more general bi-dimensionality framework of Demaine and Hajiaghayi  \cite{DBLP:conf/soda/DemaineH05}.
Moreover, the result of \cite{DBLP:conf/soda/DemaineH05} implies a EPTAS for the more general problem where one seeks to find a minimum face cover, over all possible planar drawings of the input graph.
We remark that Frederickson gave a constant-factor approximation for this more general setting, which is also sufficient for our application.
Below is a formal statement of the result we will use.

\begin{theorem}[Demaine and Hajiaghayi \cite{DBLP:conf/soda/DemaineH05}]\label{lem:face_cover_approx}
For any $\eps>0$, there exists an algorithm with running time $f^{O(1/\eps)} n^{O(1)}$, for some function $f$,
which given a planar graph $H$ and some $U\subseteq V(H)$, outputs a planar drawing $\phi$ of $H$ and a $\phi$-face cover of $U$ of size at most $(1+\eps) k^*$, where $k^*$ is the minimum size of a $\phi^*$-face cover of $U$, taken over all possible planar drawings $\phi^*$ of $H$.
\end{theorem}

\begin{lemma}[Mohar \cite{DBLP:journals/jct/Mohar01}]\label{lem:face_cover_mohar}
Let $G$ be a graph of orientable genus $g$.
Let $a\in V(G)$ such that the graph $H=G\setminus \{a\}$ is planar.
Then, there exists a planar drawing $\phi$ of $H$ and a $\phi$-face cover of $N_G(a)$ of size $O(g)$.
\end{lemma}



Lemma \ref{lem:face_cover_mohar} and Theorem \ref{lem:face_cover_approx} together imply a $O(1)$-approximation for the orientable genus of 1-apex graphs.
We shall obtain a similar algorithm for the non-orientable case.
To that end, we shall next obtain a generalization of Lemma \ref{lem:face_cover_mohar} for non-orientable embeddings.
The following is implicit in the work of Mohar \cite{DBLP:journals/jct/Mohar01}.
Mohar uses this argument to find a large collection of bouquets in a 3-connected planar graph, such that no two cycles in a bouquet have more than one vertex in common.


Let ${\cal F}$ be a collection of cycles in some graph $H$ and let $x\in V(H)$.
Suppose that the intersection of any two distinct cycles in ${\cal F}$ is either $x$ or an edge incident to $x$.
Then we say that ${\cal F}$ is a \emph{bouquet} with \emph{center} $x$.

\begin{lemma}[Mohar \cite{DBLP:journals/jct/Mohar01}]\label{lem:mohar_face_cover_implicit}
Let $H$ be a planar graph and let $\phi$ be a planar drawing of $H$.
Let ${\cal F}$ be a collection of faces in $\phi$ such that each $F\in {\cal F}$ is an induced cycle.
Suppose that no two faces in ${\cal F}$ have more than one vertex in common.
Then there exists ${\cal F}'\subseteq {\cal F}$, with $|{\cal F}'| \geq |{\cal F}|/10$, such that the faces in ${\cal F}'$ form a collection of bouquets in which no two cycles intersect more than in a vertex.
\end{lemma}



\begin{lemma}\label{lem:SPQR_apex}
Let $G$ be a graph genus $g$, and let $a\in V(G)$ such that $H=G\setminus \{a\}$ is planar and 2-connected.
Let $\phi$ be a planar drawing of $H$.
Let ${\cal F}$ be a minimal collection of faces in $\phi$ that cover $N_G(a)$.
Then there exists ${\cal F}'\subseteq {\cal F}$, with $|{\cal F}'| = \Omega(|{\cal F}|/g)$, such that no two faces in ${\cal F}'$ have more than one vertex in common.
\end{lemma}

\begin{proof}
We first show that there can be at most $O(g)$ faces in ${\cal F}$ that all have at least two vertices in common.
Let $u,v\in V(H)$, and let ${\cal J}\subseteq {\cal F}$ such that any $F\in {\cal J}$ contains both $u$ and $v$.
By the minimality of ${\cal F}$ it follows that for any $F\in {\cal J}$ there exists $w_F\in V(F)$ with $a\in N_G(a)$, and such that $w_F$ is not contained in any other $F'\in {\cal J}$.
We can therefore construct a $K_{3,|{\cal J}|}$ minor in $G$ as follows.
The left side consists of $a$, $u$, and $v$, and the right side contains all the vertices $w_F$, $F\in {\cal J}$.
It follows by Lemma \ref{lem:K33} that $|{\cal J}| = O(g)$, as required.
This establishes that there are at most $O(g)$ faces in ${\cal F}$ that all share any two vertices.

Next, arguing as in \cite{DBLP:journals/jct/Mohar01}, it follows by the 4-color theorem that there exists ${\cal F}_0\subseteq {\cal F}$, with $|{\cal F}_0|\geq |{\cal F}|/4$, and such all cycles in ${\cal F}_0$ are pairwise edge-disjoint.

Let now $T$ be an SPQR tree of $H$ (for the definition of SPQR trees and further exposition we refer the reader to \cite{DBLP:conf/focs/BattistaT89}).
The tree $T$ corresponds to a tree decomposition of $H$, where every edge in $T$ corresponding to a 2-separator in $H$.
We consider $T$ a being rooted at some $r\in V(T)$.
For any $v\in V(T)$ let $H_v$ be the union of all bubbles in the subtree of $T$ rooted at $v$; in particular $H_r=H$.

We start with the collection of faces ${\cal F}_0$, which we initially consider to be active, and proceed to refine it until no two faces have more than one vertex in common.
Along the way, we delete some of the faces and \emph{charge} them to the faces that remain.
We examine all vertices in $T$ in a bottom-up fashion starting from the leaves.
Each leaf of $T$ is a $Q$-node; we ignore all such nodes.
Consider now some internal node $v\in V(T)$, and suppose that we have already examined all the children of $v$.
Let $u_1,\ldots,u_k$ be the children of $v$.
Let ${\cal F}_v$ be the set of all active faces that contain at least one edge in $E(H_v)$.
By planarity, for any $i\in \{1,\ldots,k\}$ there is a set of at most two faces ${\cal Q}_i\subseteq {\cal F}_v$ that intersect edges of $H_{u_i}$.
Similarly, if $v\neq r$, then there exists a subset of at most two faces ${\cal R}\subseteq {\cal F}_v$ that intersect edges in $E(H_v)\setminus E(H_{v'})$, where $v'$ is the parent of $v$ in $T$; if $v=r$ then we set ${\cal R}=\emptyset$.
For any $i\in \{1,\ldots,k\}$, if ${\cal Q}_i\cap {\cal R}=\emptyset$, and ${\cal Q}_i\neq \emptyset$, then we pick some $F\in {\cal Q}_i$, and we mark it as inactive.
If there exists another face $F'\in {\cal Q}_i$, then we also mark is as inactive, and we delete it form the collection of faces. We charge $F'$ to $F$.
If ${\cal Q}_i\cap {\cal R}=\{F\}$, then we consider the following two cases: (i) If ${\cal Q}_i=\{F\}$, then we do nothing. (ii) Otherwise, we have ${\cal Q}_i=\{F,F'\}$ for some $F'$; we mark both $F$ and $F'$, we delete $F$, and we charge $F$ to $F'$.
If ${\cal Q}_i\cap {\cal R}=\{F,F'\}$, then we do nothing.
Finally, we deal with the faces in ${\cal R}$.
If ${\cal R}=\emptyset$, then we do nothing.
Otherwise, we have ${\cal R}\neq \emptyset$, which implies $v\neq r$.
Let $\{x,y\}$ be the 2-separator corresponding to the edge between $v$ and its parent in $T$.
Let ${\cal Z}$ be the set of all faces in ${\cal F}_v\setminus {\cal R}$ that intersect both $x$ and $y$.
We consider the following two cases:
(i) If ${\cal Z}=\emptyset$, then we do nothing.
(ii) Otherwise, we pick some $F''\in {\cal Z}$, we mark all faces in ${\cal Z}$ as inactive (including the faces in ${\cal R}$), we delete all faces in ${\cal Z}\setminus \{F''\}$, and we charge all the deleted faces to $F''$.
This concludes the description of the refinement process.
Let ${\cal F}'$ be the resulting set of faces.

It remains to show that ${\cal F}'$ satisfies the assertion.
We first argue that no two faces in ${\cal F}'$ share more than one vertex.
Any two faces that share two vertices must both contain some 2-separator of $H$.
After examine a vertex $v$ of the tree $T$, we maintain the invariant that for any $i\in \{1,\ldots,k\}$, either there exists at most one face in ${\cal Q}_i$ that is not deleted, or all faces in ${\cal Q}_i$ are active.
Since all faces eventually become inactive, it follows that at the end of the refinement process, at most one of the faces in ${\cal Q}_i$ is not deleted.
This establishes that no two faces in ${\cal F}'$ share more than one vertex.

Finally, we argue that $|{\cal F}'| = \Omega(|{\cal F}|/g)$.
To that end, it suffices to show that $|{\cal F}'| = \Omega(|{\cal F}_0|/g)$.
During the refinement process, whenever we charge some face $F$ to some face $F'$, we maintain the property that $F'$ is not deleted.
Therefore, it suffices to show that we charge at most $O(g)$ faces to any face.
Recall that for any two distinct $u,v\in V(H)$ there exist at most $O(g)$ faces in ${\cal F}$ that contain both $u$ and $v$.
Consider some face $F'\in {\cal F}'$.
We can charge at most one face to $F'$ when we consider some set ${\cal Q}_i$; however, if this happens, then we don't consider $F'$ in any of the subsequent iterations.
Moreover, we may charge $O(g)$ faces to $F'$ when we consider ${\cal R}$.
Overall, we can charge at most $O(g)$ faces to $F'$, which concludes the proof.
\end{proof}

The next lemma relates the Euler genus of a 1-apex graph with a 2-connected planar piece to the size of a minimum face cover of the neighborhood of the apex, taken over all possible planar embeddings of the planar piece.

\begin{lemma}
\label{lem:face_cover_2-connected}
Let $G$ be a graph of Euler genus $g$.
Let $a\in V(G)$ such that the graph $H=G\setminus \{a\}$ is planar and 2-connected.
Let $\tau$ be the size of a minimum $\phi$-face cover of $N_G(a)$, taken over all possible planar drawings of $H$.
If $\tau\geq 2$, then $\tau=O(g^2)$.
\end{lemma}

\begin{proof}
Let $\psi$ be an embedding of $G$ into a surface of Euler genus $g$.
Let ${\cal F}^*$ be a minimum $\phi$-face cover of $N_G(a)$, taken over all possible planar drawings $\phi$ of $H$.
Let $k^* = |{\cal F}^*|$.
Following Mohar \cite{DBLP:journals/jct/Mohar01}, we let ${\cal F}_0$ be the set of faces in $\phi$ that are not faces in $\psi$.
We have that ${\cal F}_0$ is a $\phi$-face cover of $N_G(a)$.
Let ${\cal F}_1$ be a minimal subset of ${\cal F}_0$ that covers $N_G(a)$.
By Lemma \ref{lem:SPQR_apex} there exists ${\cal F}_2\subseteq {\cal F}_1$ with $|{\cal F}_2|=\Omega(|{\cal F}_1|/g)$ and such that any two distinct faces in ${\cal F}_2$ share at most one vertex.
By Lemma \ref{lem:mohar_face_cover_implicit} there exist ${\cal F}_3\subseteq {\cal F}_2$ with $|{\cal F}_3| \geq |{\cal F}_2|/10 = \Omega(|{\cal F}_1|/g)$ and such that the faces in ${\cal F}_3$ form a collection of bouquets in which any two cycles share at most one vertex.

Since $H$ is 2-connected, every face in ${\cal F}_3$ is an induced cycle.
We shall modify $\phi$, and the sets ${\cal F}_0,\ldots,{\cal F}_3$, until cycles in ${\cal F}_3$ become non-contractible, while preserving the minimality of the face cover ${\cal F}^*$.
Suppose that there exists some cycle $C\in {\cal F}_3$ that is contractible.
Let ${\cal D}$ be the disk bounded by $\psi(C)$.
Since $\tau\geq 2$, it follows that $\psi(a)$ is not inside ${\cal D}$.
Since $C$ is not a face in $\psi$, it follows that there exists some vertex $v\in V(H)\setminus V(C)$ such that $\psi(v)$ is in the interior of ${\cal D}$.
Moreover, there exists some $X\subset H$, with $x\in V(X)$, and a 2-separator $\{x,x'\}\subseteq V(C)\cap V(X)$, such that $\psi(X)\subset {\cal D}$.
The boundary of $X$ consists of two paths $X_1,X_2$ with endpoints $x$ and $x'$, where $X_1\subset C$.
There exists some face $C'$ in $\phi$ that contains $X_2$ as a subpath.
We consider the following two cases:
\begin{description}
\item{Case 1:}
Suppose that $V(C') \cap N_G(a)=\emptyset$.
Then, we delete $X\setminus C$ from $H$.
Let $\phi$ be the new drawing of the resulting graph $H$.
Since $C'$ does not intersect the neighborhood of $a$, it follows that ${\cal F}^*$ remains a minimum face cover.

\item{Case 2:}
Suppose that $V(C') \cap N_G(a)\neq \emptyset$.
By the minimality of ${\cal F}^*$ it follows that $C$ must have a vertex in $N_G(a)$ that is not in $X_1$.
That is, $N_G(a) \cap (V(C) \setminus V(X_1)) \neq \emptyset$.
Then we modify $\phi$ by performing a Whitney flip on $X$.
Thus, the replace the path $X_1$ by $X_2$ in $C$, and we replace the path $X_2$ by $X_1$ in $C'$.
Since both $C$ and $C'$ have vertices in $N_G(a)$ that are not in $X_1\cup X_2$, it follows that the new ${\cal F}^*$ remains a minimum face cover.
\end{description}

We repeat the above process until all cycles in ${\cal F}_3$ become contractible, every time recomputing the sets ${\cal F}_0, \ldots, {\cal F}_3$.
We argue that the process terminates after finitely many steps.
This is because in Case 1 above we delete at least one vertex from $H$, and in Case 2 we increase by at least two the number of vertices in $V(H)$ where the local rotations of $\phi$ and $\psi$ agree (fixing an appropriate signature).
When the process terminates we arrive at some planar embedding $\phi$, and corresponding sets ${\cal F}^*$, and ${\cal F}_0,\ldots, {\cal F}_3$, where ${\cal F}^*$ is a minimum face cover, and all cycles in ${\cal F}_3$ are non-contractible.

Since ${\cal F}_1$ is minimal, it follows that every $C\in {\cal F}_3\subseteq {\cal F}_1$ must contain at least one vertex incident to $a$ that is not in the center of the bouquet containing $C$.
Let $C\in {\cal F}_3$.
We argue that there can be at most $2$ cycles in ${\cal F}_3$ homotopic to $C$ in $\psi$.
We consider two cases:
\begin{description}
\item{Case 1:}
Suppose that the loop $\psi(C)$ is two-sided.
Let $C',C''\in ({\cal F}_3\setminus {\cal F}_4)\setminus \{C\}$ such that $\psi(C')$ and $\psi(C'')$ are homotopic to $\psi(C)$.
Any two of the loops in $\psi(C)$, $\psi(C')$, and $\psi(C'')$ either bound a cylinder, or a disk with two points of its boundary identified.
It follows that removing $\psi(C)$, $\psi(C')$, and $\psi(C'')$ from the surface we obtain three connected components.
Thus, there is no path in the surface between $\psi(a)$ and at least one of the $\psi(C)$, $\psi(C')$, and $\psi(C'')$, a contradiction.

\item{Case 2:}
Suppose next that the loop $\psi(C)$ is one-sided.
Arguing as above, let $C',C''\in {\cal F}'''\setminus \{C\}$ such that $\psi(C')$ and $\psi(C'')$ are homotopic to $\psi(C)$.
Any two one-sided homotopic loops must intersect.
Therefore there exists some bouquet with center $c$ that contains all of the cycles $C$, $C'$, and $C''$, and such that $\psi(C)$, $\psi(C')$, and $\psi(C'')$ intersect at $\psi(c)$.
Any two of these loops bound a disk with two points of its boundary identified.
Therefore we conclude as above that removing $\psi(C)$, $\psi(C')$, and $\psi(C'')$ from the surface we obtain three connected components, and thus there is no path in the surface between $\psi(a)$ and at least one of the $\psi(C)$, $\psi(C')$, and $\psi(C'')$, a contradiction.
\end{description}
We have therefore obtained that there at most two loops in ${\cal F}_3$ that are in the same homotopy class.

Let ${\cal Z}\subseteq {\cal F}_3$ be obtained by keeping at most one loop from every homotopy class. Clearly, $|{\cal Z}|\geq |{\cal F}_3|/2$.
It follows that $\{\psi(C)\}_{C\in {\cal Z}}$ can be partitioned into pairwise disjoint collections of loops, where all loops in each collection intersect at exactly one base-point.
It follows that $|{\cal Z}| = O(g)$, which implies $\tau = O(|{\cal F}_1|) = O(g\cdot |{\cal F}_3|) = O(g\cdot |{\cal Z}|) = O(g^2)$, concluding the proof.
\end{proof}

Next, we obtain a generalization of Lemma \ref{lem:face_cover_2-connected} for the case where the planar graph is not necessarily 2-connected.

\begin{lemma}
\label{lem:face_cover}
Let $G$ be a graph of Euler genus $g$.
Let $a\in V(G)$ such that the graph $H=G\setminus \{a\}$ is planar.
Let $\tau$ be the size of a minimum $\phi$-face cover of $N_G(a)$, taken over all possible planar drawings of $H$.
If $\tau\geq 2$, then $\tau=O(g^2)$.
\end{lemma}
\begin{proof}
We may assume that $g\geq 1$, since otherwise the assertion is immediate.
Let ${\cal C}$ be the set of maximal 2-connected components of $H$.
For each $C\in {\cal C}$, let $\tau_C$ be the minimum size of a $\phi_C$-face cover of $N_G(a)\cap V(C)$, taken over all planar drawings $\phi_C$ of $C$.

Let ${\cal C}_0 = \{C\in {\cal C} : N_G(C)\cap X=\emptyset\}$.
Let $G'$ be the graph obtained from $G$ by contracting each all $C\in {\cal C}_0$ into single vertices.
Let also $H'$ be the corresponding subgraph of $G$ obtained from $H$ after contracting each $C\in {\cal C}_0$.
Let ${\cal T}$ be the tree of the biconnected component tree-decomposition of $H'$.
We have $V({\cal T}) = {\cal C} \setminus {\cal C}_0$.
Pick some arbitrary $C_R\in V({\cal T})$.
Define the partition $V({\cal T})={\cal C}_1\cup {\cal C}_2$, where ${\cal C}_1$ (resp.~${\cal C}_2$) contains all $C\in V({\cal T})$ such that the distance between $C_R$ and $C$ in ${\cal T}$ is odd (resp.~even).
Let $i\in \{1,2\}$.
Let $G''_i$ be the graph obtained from $G'$ by contracting each cluster $C\in {\cal C}_{3-i}$ into a vertex $v_C$.
Let also $H''_i$ be the corresponding subgraph of $G''_i$ obtained form $H'_i$.
Every 1-separator in $H''_i$ is adjacent to $a$.
For each $C\in {\cal C}_i$ let $G''_{i,C}=G''_i[V(C)\cup \{a\}]$.
Therefore by Lemma \ref{lem:2sum} we get $\eg(G''_i) = \Omega\left(\sum_{C\in {\cal C}_i} \eg(G''_{i,C} \right)$.
For each $C\in {\cal C}_i$, let $\tau_C'$ be the minimum size of a $\phi_C$-face cover of $N_{G''_i}(\{a\}) \cap V(C)$, taken over all planar drawings $\phi_C$ of $C$.
Since $N_{G}(\{a\}) \cap V(C) \subseteq N_{G''_i}(\{a\}) \cap V(C)$, it follows that $\tau_C'\leq \tau_C$.
Let ${\cal C}_{i,1}$ be the set of all $C\in {\cal C}_i$ such that $G''_{i,C}$ is non-planar.
Let also ${\cal C}_{i,2}={\cal C}_i \setminus {\cal C}_{i,1}$.
By Lemma \ref{lem:face_cover_2-connected} we have that for each $C\in {\cal C}_{i,1}$, $\tau_C' = O((\eg(G''_{i,C}))^2)$.
Thus $\sum_{C\in {\cal C}_{i,1}} \tau_C \leq \sum_{C\in {\cal C}_{i,1}} \tau_C' = O\left(\sum_{C\in {\cal C}_{i,1}}(\eg(G''_{i,C}))^2\right) = O((\eg(G))^2)$.
Thus $\sum_{i=1}^2 \sum_{C\in {\cal C}_{i,1}} \tau_C = O((\eg(G))^2)$.

For any $C\in {\cal C}\setminus {\cal C}_0$ fix a planar drawing $\phi_C$ that minimizes the size of a $\phi_C$-cover ${\cal F}_C$ of $N_G(a)\cap V(C)$.
Moreover, for any $C\in {\cal C}_0$, let $\phi_C$ be a planar drawing of $C$ where all 1-separators of $H$ in $C$ are in the outer face.
We can combine all these drawings into a planar drawing $\phi$ of $H$ as follows.
Let ${\cal S}$ be the tree of biconnected component tree-decomposition of $H$.
We perform a traversal of ${\cal S}$ starting from some arbitrary $C_r\in V({\cal S})\setminus {\cal C}_0$.
We start by setting $\phi=\phi_{C_r}$.
We also maintain a face cover of $N_G(a)\cap \Gamma$, where $\Gamma$ is the current graph on which $\phi$ is defined.
The first time we visit some $C\in V({\cal S})$ other than $C_r$, we extend the current planar drawing $\phi$ to $C$ as follows.
Let be $v$ be a 1-separator of $H$ that is in $C$, and for which $\phi(v)$ is already defined (such a 1-separator always exists between $C$ and component $C'\in V({\cal S})$ we traversed before visiting $C$).
If there exists a face in the current $\phi$-face cover that contains $v$, then we pick such a face $F_{C}$; otherwise we set $F_{C}$ to be an arbitrary face of $\phi$ containing $v$.
Similarly, if there exists a face in $\phi_C$ that contains $v$, then we let $R_C$ be such a face; otherwise we let $R_C$ be an arbitrary face of $\phi_C$ containing $v$.
We extend $\phi$ to $C$ be placing $C$ inside $F_{C}$, and by placing the rest of the current graph $\Gamma$ inside $R_C$.
It is immediate by induction on the traversal that the resulting face cover has size at most $\sum_{C \in {\cal C}_{1,1}\cup {\cal C}_{2,1}} \tau_C = O((\eg(G))^2)$, concluding the proof.
\end{proof}

The following is the main result of this section, and follows immediately from Lemma \ref{lem:face_cover} and Theorem \ref{lem:face_cover_approx}.
We remark that even though the embedding obtained by Corollary \ref{cor:1-apex} is into some orientable surface ${\cal S}$, the genus of ${\cal S}$ is small compared to the Euler genus of $G$.
In particular, this implies that the orientable and non-orientable genus of any $1$-apex graph are polynomially related.

\begin{corollary}\label{cor:1-apex}
There exists a polynomial-time algorithm which given a 1-apex graph $G$ computes a drawing of $G$ into an orientable surface of orientable genus $O((\eg(G))^2)$.
\end{corollary}

\begin{proof}
Let $a\in V(G)$ be an apex of $G$, that is such that $H=G\setminus \{a\}$ is planar.
Let $U=N_G(\{a\})$.
Let $k^*$ be the size of the minimum $\phi^*$-face cover of $U$, taken over all possible planar drawings $\phi^*$ of $H$.
By Lemma \ref{lem:face_cover} we have $k^*=O((\eg(G))^2)$.
By Lemma \ref{lem:face_cover_approx} we can compute a planar drawing $\phi$ of $H$ and a $\phi$-face cover ${\cal F}$ of $U$ of size at most $O(k^*)=O((\eg(G))^2)$.
We can extend $\phi^*$ to $G$ by one handle  $C_F$ for every $F\in {\cal F}$ connecting a puncture in the interior of $\phi(F)$ to a puncture in a neighborhood of $\phi(a)$.
For each $v\in U$ covered by $F$ we map the edge $\{v,a\}$ to a path from $\phi(v)$ to $\phi(a)$ along $C_F$.
This results in an embedding of $G$ into an orientable surface of genus at most $O((\eg(G))^2)$, as required.
\end{proof}

\section{Computing flat grid minors}\label{sec:flat_grids}

In this section we present our algorithm for computing flat grid minors (Sub-problem 1).
We begin by recalling the following result from \cite{DBLP:conf/focs/ChekuriS13} for computing planarizing sets.

\begin{lemma}[Chekuri and Sidiropoulos \cite{DBLP:conf/focs/ChekuriS13}]\label{lem:large_tw_planar_piece}
  Let $G$ be a graph of Euler genus $g\geq 1$, and treewidth $t\geq
  1$. There is a polynomial time algorithm to compute a set
  $X\subseteq V(G)$, with $|X| = O(gt \log^{5/2} n)$, and a planar
  connected component $\Gamma$ of $G \setminus X$ containing the
  $(r'\times r')$-grid as a minor, with $r' = \Omega\left(
    \frac{t}{g^{3} \log^{5/2}n} \right)$.  (The algorithm does not
  require a drawing of $G$ as part of the input.)
\end{lemma}

We will also need the following $O(1)$-approximation algorithm for computing  grid minors in planar graphs.

\begin{lemma}[Seymour and Thomas \cite{DBLP:journals/combinatorica/SeymourT94}, Robertson, Seymour, and Thomas \cite{DBLP:journals/jct/RobertsonST94}]\label{lem:planar_grid_minor_approx}
  Let $r>0$, and let $G$ be a planar graph containing a $(r \times
  r)$-grid minor.  Then, on input $G$, we can compute in polynomial
  time a $(\Omega(r)\times \Omega(r))$-grid minor in $G$.
\end{lemma}

We first establish the following auxiliary lemma for finding  a large $K_{2,r}$ minor in a grid.

\begin{lemma}[$K_{2,r}$ minors in grids]\label{lem:K2r_grid}
Let $r\geq 1$ and let $H$ be the $(r\times r)$-grid.
Let $\partial H$ denote the boundary cycle of $H$.
Let $A\subseteq V(H) \setminus V(\partial H)$.
Then $H$ contains as a minor the graph $\Gamma=K_{2,\ell}$, for some $\ell\geq |A|/3$.
Furthermore, let $U_1$ and $U_2$ be the left and right sides of $\Gamma$ respectively, where $|U_1| = 2$, $|U_2|=\ell$.
Then there exists a minor mapping $\mu : V(\Gamma)\to 2^{V(H)}$, such that for each $v\in U_2$, $\mu(v)\cap A\neq \emptyset$.
\end{lemma}

\begin{proof}
For any $i,j\in \{1,\ldots,r\}$ let $v_{i,j}$ be the vertex in the $i$-row and $j$-th column of $H$.
For any $i\in \{1,\ldots,r\}$ let $R_i$ be the $i$-th row of $H$, that is $R_i=\bigcup_{j=1}^r \{v_{i,j}\}$, and let $C_i$ be the $i$-th column of $H$, that is $C_i = \bigcup_{j=1}^r \{v_{j,i}\}$.
For any $t\in \{0,1,2\}$ we define a pair of vertex-disjoint subtrees $T_t$, $T_t'$ of $H$.
Let
\[
T_t = H\left[R_1 \cup \left(\bigcup_{i=0}^{\lceil r/3 \rceil-1} (C_{3i+1+t}\setminus R_r) \right) \right]
\]
and
\[
T_t' = H\left[R_r \cup \left(\bigcup_{i=0}^{\lfloor r/3 \rfloor-1} (C_{3i+3+t}\setminus R_1) \right) \right].
\]
It is immediate that for any $t\in \{0,1,2\}$ the trees $T_t$ and $T_t'$ are vertex-disjoint.
Moreover, for any $v\in A$ there exists a unique $t\in \{0,1,2\}$ such that $v\notin V(T_t)\cup V(T_t')$.
It follows that there exists $t^*\in \{0,1,2\}$ such that $|A\setminus (V(T_{t^*}) \cup V(T_{t^*}'))| \geq |A|/3$.

Let $A'=A\setminus (V(T_{t^*}) \cup V(T_{t^*}'))$.
For every $v_{i,j}\in A'$ $v_{i,j-1}\in V(T_{t^*})$ and $v_{i,j+1} \in V(T_{t^*}')$.
Therefore, for every $v\in A$ there exists an edge between $v$ and some vertex in $T_{t^*}$ and an edge between $v$ and some vertex in $T_{t^*}'$.
We can therefore construct a $K_{2,|A'|}$ minor in $H$ as follows.
The right side is $A'$ and
the left side contains a vertex obtained by contracting each one of $T_{t^*}$ and $T_{t^*}'$.
This concludes the proof.
\end{proof}

Using the above lemma we next show how to find large $K_{3,r}$ minors in a 1-apex graph.

\begin{lemma}[$K_{3,r}$ minors in apex graphs]\label{lem:K3r_apex}
Let $r\geq 1$ and let $H$ be the $(r\times r)$-grid.
Let $\partial H$ denote the boundary cycle of $H$.
Let $A\subseteq V(H) \setminus V(\partial H)$.
Let $G$ be the graph obtained by adding a new vertex $a$ to $H$ and connecting it to all vertices in $A$.
That is, $V(G)=V(H)\cup \{a\}$ and $E(G) = E(H) \cup \bigcup_{a'\in A} \{\{a,a'\}\}$.
Then $G$ contains $K_{3,\ell}$ as a minor, for some $\ell\geq |A|/3$.
\end{lemma}

\begin{proof}
By Lemma \ref{lem:K2r_grid}
we can compute a mapping $\mu : V(\Gamma) \to 2^{V(H)}$, where $\Gamma=K_{2,\ell}$ for some $\ell \geq |A|/3$.
Furthermore, if we let $U_1$ and $U_2$ be the left and right sides of $\Gamma$ respectively, then for any $v\in U_2$ we have $\mu(v) \cap A \neq \emptyset$.
Therefore, for each $v\in A$ there exists in $G$ an edge between $a$ and some vertex in $\mu(v)$.
It follows that adding $a$ to $\Gamma$ we obtain a $K_{3,\ell}$ minor, as required.
\end{proof}

We are now ready to prove the main result of this section, concerning the computation of flat grid minors.

\begin{proof}[Proof of Lemma \ref{lem:flat_grid}]
  We first use Lemma \ref{lem:large_tw_planar_piece} to find a set $X \subseteq V(G)$, with $|X| = O(gt \log^{5/2} n)$, and a planar connected component $C$ of $G\setminus X$, such that $C$ contains a
  $(r'\times r')$-grid minor, for some
  $r'=\Omega\left(\frac{t}{g^{3} \log^{5/2} n}\right)$.  Using Lemma \ref{lem:planar_grid_minor_approx} we can
  compute a $(k\times k)$-grid minor $H$ in $C$, for some
  $k=\Omega(r')$.
We may assume that $k\geq 3$, since otherwise the assertion holds trivially.

Fix a mapping $\mu:V(H) \to 2^{V(C)}$ for $H$.
We can choose $\mu$ so that $\mu(H)=C$.
Let $H'$ be the $((k-2)\times (k-2))$-grid obtained by removing the boundary cycle of $H$.
Let $h=c\cdot g \cdot |X|$, for some constant $c>0$ to be specified.
The grid $H'$ contains a collection of $h$ pairwise vertex-disjoint $(k'\times k')$-grids ${\cal H}=\{H_1,\ldots,H_h\}$, for some $k'=\Omega\left(\frac{k}{h^{1/2}}\right) = \Omega\left(\frac{t^{1/2}}{g^{4} \log^{15/4} n}\right)$.

Let $v\in X$ be incident to $\delta$ distinct grids in ${\cal H}$.
By Lemma \ref{lem:K3r_apex} it follows that $G$ contains a $K_{3,\delta/3}$ minor.
By Lemma \ref{lem:K33} we obtain that $\delta=O(g)$.
It follows that every $v\in X$ is incident to at most $\delta = O(g)$ distinct grids in ${\cal H}$.
Thus $X$ is adjacent to at most $O(g\cdot |X|)$ vertices in $G \setminus X$.
Therefore, for a sufficiently large constant $c$, we get that there exists $i\in \{1,\ldots,h\}$, such that $\mu(H_i)$ is not adjacent to $X$.
It follows that the neighborhood of $\mu(H_i)$ is contained in $H$, which implies that $\mu(H_i)$ is flat, concluding the proof.
\end{proof}

\fi

\bibliographystyle{alpha}
\bibliography{bibfile}

\iffull

\appendix

\section{High-level overview of the proof of Lemma \ref{lem:CS_summary}}
\label{sec:CS_summary}

For the sake of completeness of our exposition, we now give a high-level overview of the proof of Lemma \ref{lem:CS_summary}.
For more details we refer the deader to \cite{DBLP:conf/focs/ChekuriS13}.

We us start by recalling some definitions from \cite{DBLP:conf/focs/ChekuriS13}.
Let $G$ be a graph, let $X\subseteq G$ be a subgraph, and let $C\subsetneq X$ be a cycle.
We say that the ordered pair $(X,C)$ is a \emph{patch} (of $G$).
Let now $(X,C)$ be a patch of some graph $G$.
Let $\phi$ be a drawing of $G$ into a surface ${\cal S}$.
We say that $(X, C)$ is a \emph{$\phi$-patch} (of $G$), if there exists a disk ${\cal D} \subset {\cal S}$, satisfying the following conditions:
\begin{description}
\item{(1)}
$\partial {\cal D} = \phi(C)$.
\item{(2)}
$\phi(G) \cap {\cal D} = \phi(X)$.
\end{description}


Let ${\cal A}_1$, ${\cal A}_2$ be the polynomial-time algorithms satisfying  conditions (1) and (2) of Lemma \ref{lem:CS_summary}, respectively.
That is, given a $n$-vertex graph $G$ of treewidth $t$ and an integer $g\geq 1$, algorithm ${\cal A}_1$ either correctly decides that $\eg(G)>g$, or  outputs a flat subgraph $G'\subset G$, such that $X$ contains a $\left(\Omega(r)\times \Omega(r)\right)$-grid minor, for some $r=r(n,g,t)$.
Also, given an $n$-vertex graph $G$, an integer $g$, and some $X\subset V(G)$ such that $G\setminus X$ is planar, algorithm ${\cal A}_2$ either correctly decides that $\eg(G) > g$, or it outputs a drawing of $G$ into a surface of Euler genus at most $\gamma$, for some $\gamma=\gamma(n, g, |X|)$.

We now describe a polynomial-time algorithm ${\cal A}_3$ that satisfies the conclusion of Lemma \ref{lem:CS_summary}.
The input consists of a graph $G$ and some integer $g\geq 1$.
The algorithm proceeds in the following steps:

\begin{description}
\item{\textbf{Step 1: Computing a skeleton.}}
Let $t$ be the treewidth of $G$.
We first describe a polynomial-time procedure which either correctly decides that $\eg(G)>g$ or outputs a collection of pairwise non-overlapping patches $(X_1,C_1),\ldots,(X_{\ell},C_{\ell})$ of $G$, so that the following conditions are satisfied:
\begin{description}
\item{(1)}
If $\eg(G)\leq g$, then there exists a drawing $\phi$ of $G$ into a surface of Euler genus $g$, such that for any $i \in \{1,\ldots, \ell\}$, $(X_i, C_i)$ is a $\phi$-patch.
We emphasize the fact that $\phi$ is not explicitly computed by the algorithm.
\item{(2)}
Let $t'$ be the treewidth of the graph $G\setminus \left(\bigcup_{i=1}^{\ell} (X_i \setminus C_i) \right)$.
Then $r(n,g,t')=O(g)$.
\end{description}

\begin{center}
\scalebox{0.7}{\includegraphics{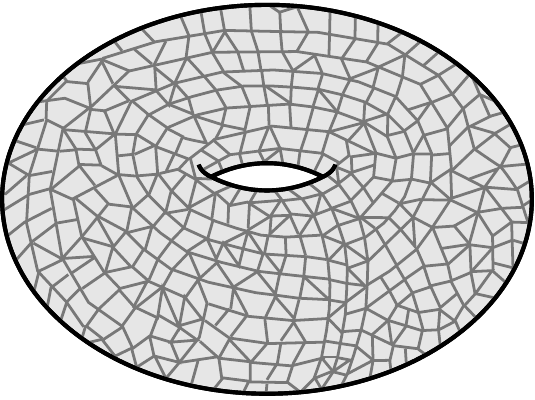}}
\hspace{1cm}
\scalebox{0.7}{\includegraphics{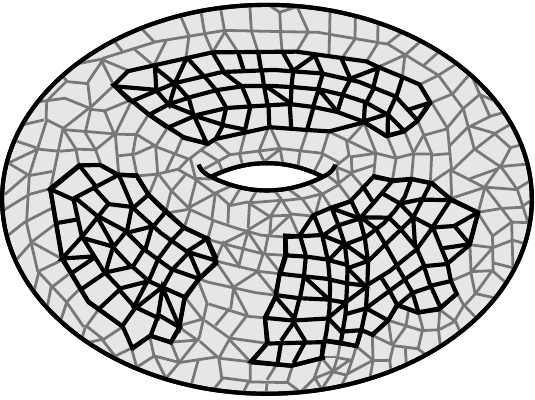}}
\hspace{1cm}
\scalebox{0.7}{\includegraphics{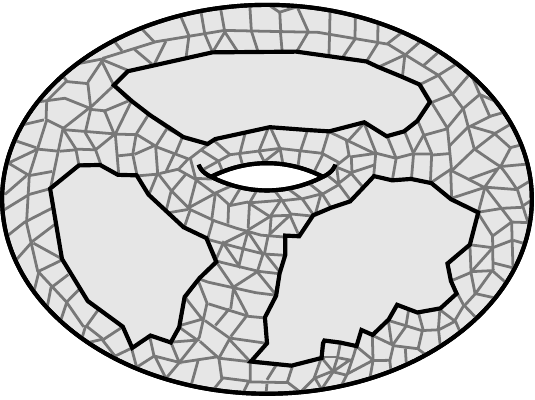}}
\end{center}

If $r(n,g,t)=O(g)$, then we may simply output an empty sequence of patches.
We may therefore assume that $r(n,g,t)=c\cdot g$, for some sufficiently large universal constant $c>0$ to be specified later.
We inductively compute the desired sequence of $\phi$-patches, for some embedding $\phi$ of $G$ into a surface of Euler genus $g$, as follows.
We remark that we do not compute $\phi$; we are only guaranteed that some $\phi$ satisfying the condition exists.
Suppose that we have computed some  $(X_1,C_1),\ldots,(X_i,C_i)$.
Let $G_i=G\setminus \left(\bigcup_{i=1}^{\ell} (X_i \setminus C_i) \right)$.
In particular, we have $G_0=G$.
Let $c$ be a sufficiently large constant, to be specified later.
If the treewidth $t_i$ of $G_i$ satisfies
$r(n,g,t_i) \leq c\cdot g$, then we terminate the sequence of patches by setting $\ell=i$.
Otherwise, we proceed to compute $(X_{i+1}, C_{i+1})$.
Using algorithm ${\cal A}_1$ we either correctly decide that $\eg(G_i)>g$ (and therefore $\eg(G)>g$) or we compute a flat subgraph $G'_i\subset G_i$, such that $G_i'$ contains a $\left(c\cdot g \times c\cdot g\right)$-grid minor.
Using Lemma \ref{lem:planar_grid_minor_approx} we can compute a $(c\cdot c'\cdot g\times c\cdot c'\cdot g)$-grid minor $H_i$ in $G_i'$, for some universal constant $c'>0$.
By setting $c$ to be sufficiently large, the grid $H_i$ contains a sequence at least $g+3$ ``planarly nested'' cycles.
It can be shown that in any embedding of $G$ into a surface of Euler genus $g$, the ``inner-most'' cycle must bound a disk.
This allows us to compute a $\phi$-patch $(X_{i+1}, C_{i+1})$.
With some extra analysis, after ``merging'' some of the patches, we may ensure that the disks bounded by different patches do not intersect.
This completes the computation of the desired $\phi$-patches $(X_1,C_1),\ldots,(X_{\ell},C_{\ell})$, for some embedding $\phi$.

Let $G'=G\setminus \left(\bigcup_{i=1}^{\ell} (X_i \setminus C_i) \right)$.
We refer to $G'$ as the \emph{skeleton}.

\item{\textbf{Step 2: Framing the skeleton.}}
For every $i\in \{1,\ldots,\ell\}$, we have that $C_i$ is a cycle in $G'$.
We add a $(|V(C_i)|\times 3)$-cylinder $K_i$ to $G'$ by identifying the outer face of $K_i$ with $C_i$.
Let $G''$ be the resulting graph.
We refer to $G''$ as the \emph{framed skeleton}.

\begin{center}
\scalebox{0.6}{\includegraphics{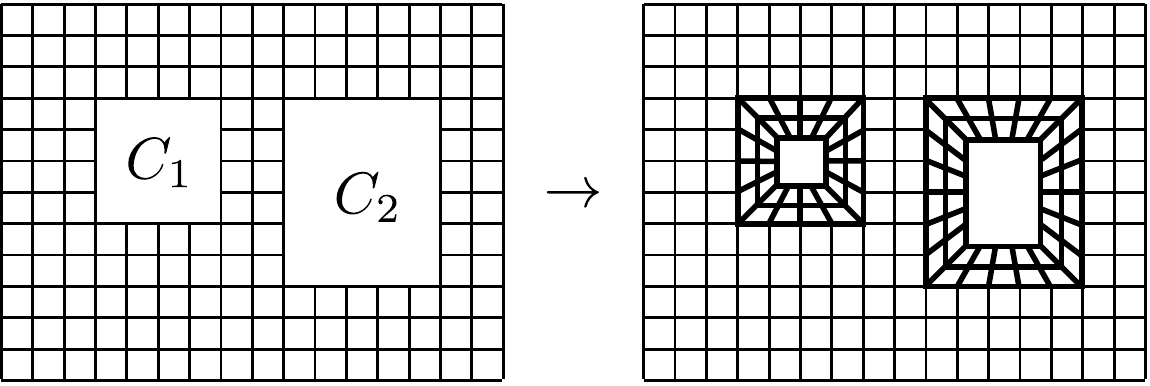}}
\end{center}

\item{\textbf{Step 3: Embedding the framed skeleton.}}
By recursively removing balanced vertex-separators, we can compute some $X\subseteq V(G')$ with $|X|=O(t' g  \log^{3/2}n)$, such that $G''=G'\setminus X$ is planar.
Moreover, we may ensure that for each component $\Gamma$ of $G'$, the ``framing'' of $\Gamma$ (defined in an analogous manner) remains planar.

\begin{center}
\scalebox{0.6}{\includegraphics{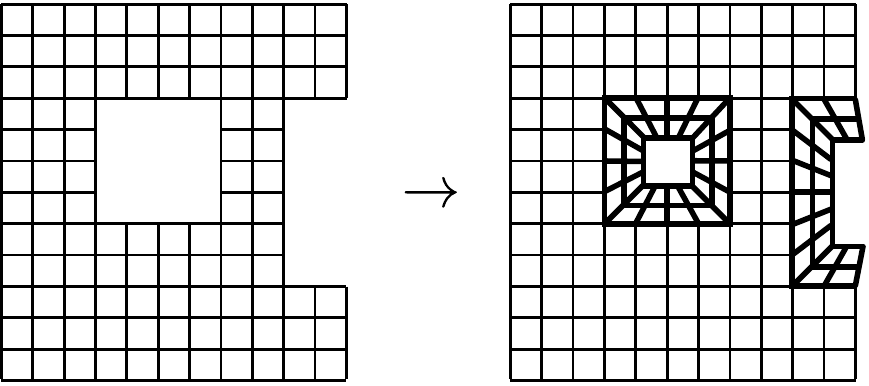}}
\end{center}

Using algorithm ${\cal A}_2$ we either correctly decide that $\eg(G')>g$ (and therefore $\eg(G)>g$),
or we compute an embedding of $G'$ into a surface of Euler genus at most $\gamma(n, g, |X|) = \gamma(n, g, k)$, for some $k=O(t' g  \log^{3/2}n)$.

\item{\textbf{Step 4: Extending the embedding to the original graph.}}
Since $G''$ is obtained from $G'$ by removing $O(t' g  \log^{3/2}n)$ vertices, it follows that the cycles $C_1,\ldots,C_{\ell}$ are partitioned in $G''$ into at most $\ell+O(t' g  \log^{3/2}n)$ cycles and paths.
We may therefore extend the embedding of $G''$ to $G'$ by adding at most $O(t' g  \log^{3/2}n)$ new handles to the surface.
At the end of the algorithm we obtain an embedding of $G'$ into a surface of Euler genus at most $\gamma(n, g, k) + O(t' g  \log^{3/2}n)$, for some $k=O(t' g  \log^{3/2}n)$.



\begin{center}
\scalebox{0.6}{\includegraphics{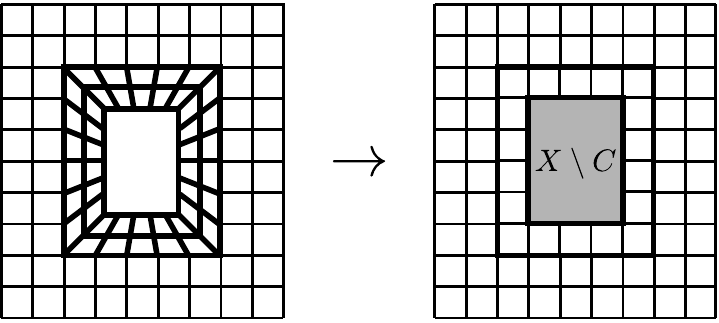}}
\hspace{1cm}
\scalebox{0.6}{\includegraphics{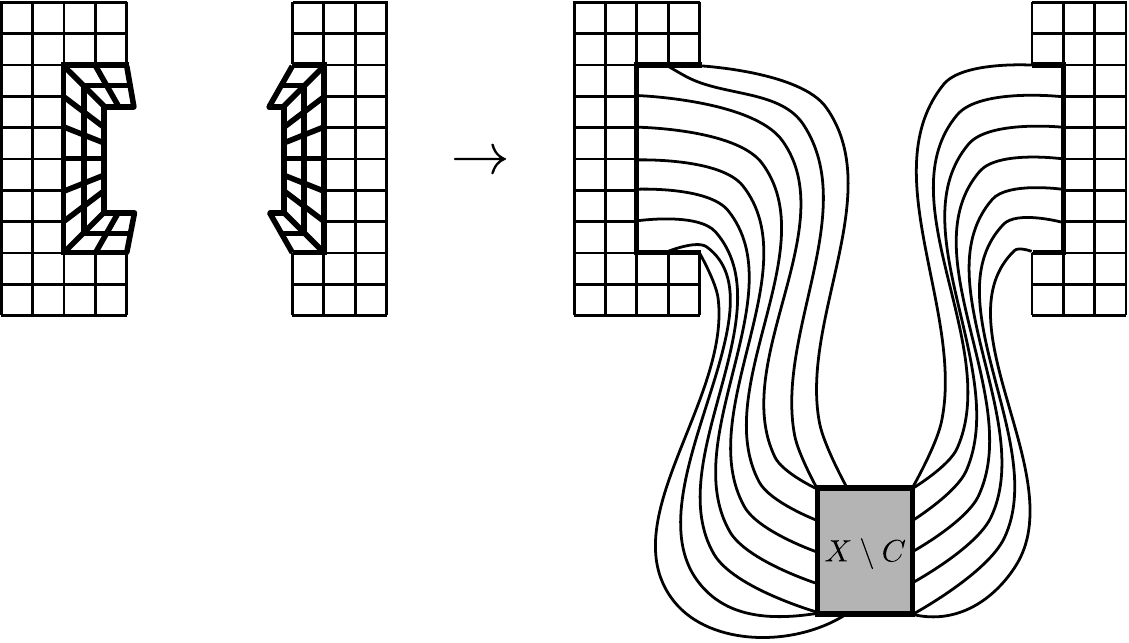}}
\end{center}

This completes the high-level overview of the proof of Lemma \ref{lem:CS_summary}.

\end{description}

\fi

\end{document}